\documentclass[twoside,leqno]{article}
\usepackage[letterpaper]{geometry}
\usepackage{ltexpprt}
\usepackage[hidelinks]{hyperref}

\usepackage{graphicx}
\usepackage{amsmath}
\usepackage{amssymb}
\usepackage{xspace}
\usepackage{epsfig}
\usepackage[dvipsnames]{xcolor}
\usepackage{xfrac}
\usepackage{cancel}
\usepackage{booktabs}
\usepackage{fancyhdr}

\usepackage{wrapfig}

\usepackage[noabbrev, capitalize]{cleveref}
\usepackage{enumerate}

\newcommand {\mm}[1] {\ifmmode{#1}\else{\mbox{\(#1\)}}\fi}

\newcommand{\ignore}[1]{}

\newsavebox{\smallProofsym}                 
\savebox{\smallProofsym}                               %
{
\begin{picture}(6,6)
\put(0,0){\framebox(6,6){}}
\put(0,2){\framebox(4,4){}}
\end{picture} 
}  

\makeatletter
\long\def\@makecaption#1#2{%
  \vskip\abovecaptionskip
  \sbox\@tempboxa{\small #1: #2}%
  \ifdim \wd\@tempboxa >\hsize
    \small #1: #2\par
  \else
    \global \@minipagefalse
    \hb@xt@\hsize{\hfil\box\@tempboxa\hfil}%
  \fi
  \vskip\belowcaptionskip}
\makeatother

\newtheorem{claim}{Claim}[section]

\newtheorem{invariant}{Invariant}
\crefname{invariant}{Invariant}{Invariants}
\Crefname{invariant}{Invariant}{Invariants}

\crefname{@theorem}{Theorem}{Theorems}
\Crefname{@theorem}{Theorem}{Theorems}

\newcommand{\Rspace}        {\mm{{\mathbb R}}}

\newcommand{\UpTree}[1]     {\mm{\rm Up}{({#1})}}
\newcommand{\DnTree}[1]     {\mm{\rm Dn}{({#1})}}

\newcommand{\Point}[2]      {\mm{{\rm Point}{({#1},{#2})}}}
\newcommand{\Arrow}[2]      {\mm{{\rm Arrow}{({#1},{#2})}}}
\newcommand{\Birth}[1]      {\mm{{\rm Bth}{({#1})}}}

\newcommand{\Low}[1]        {\mm{{\sf low}{({#1})}}}
\newcommand{\Dth}[1]        {\mm{{\sf dth}{({#1})}}}
\newcommand{\Dn}[1]         {\mm{{\sf dn}{({#1})}}}
\newcommand{\Up}[1]         {\mm{{\sf up}{({#1})}}}
\newcommand{\In}[1]         {\mm{{\sf in}{({#1})}}}
\newcommand{\Mid}[1]        {\mm{{\sf mid}{({#1})}}}
\newcommand{\Prv}[1]        {\mm{{\sf prv}{({#1})}}}
\newcommand{\Nxt}[1]        {\mm{{\sf nxt}{({#1})}}}

\newcommand{\T}[1]          {\mm{{\tt {#1}}}}

\DeclareMathOperator{\UpOp}{\mathsf{up}}                   
\DeclareMathOperator{\InOp}{\mathsf{in}}

\newcommand{\sign}[1]       {\mm{{\rm sgn}{({#1})}}}

\newcommand{\Dgm}[2]        {\mm{\sf Dgm}_{#1}{({#2})}}
\newcommand{\DgmR}[2]       {\mm{{\sf D}\overrightarrow{\sf gm}}_{#1}{({#2})}}
\newcommand{\Ord}[2]        {\mm{\sf Ord}_{#1}{({#2})}}
\newcommand{\Rel}[2]        {\mm{\sf Rel}_{#1}{({#2})}}
\newcommand{\Ess}[2]        {\mm{\sf Ess}_{#1}{({#2})}}
\newcommand{\Window}[2]     {\mm{W}{({#1},{#2})}}

\newcommand{\Path}[2]       {\mm{\mm{P}{({#1},{#2})}}}

\newcommand{\Lup}           {\mm{L_{\rm up}}}
\newcommand{\Mup}           {\mm{M_{\rm up}}}
\newcommand{\Rup}           {\mm{R_{\rm up}}}
\newcommand{\Ldn}           {\mm{L_{\rm dn}}}
\newcommand{\Mdn}           {\mm{M_{\rm dn}}}
\newcommand{\Rdn}           {\mm{R_{\rm dn}}}

\newcommand{\Stack}         {\mm{Stack}}

\newcommand{\ee}            {\mm{\varepsilon}}

\newcommand{\glueplow}{\hat{p}} %
\newcommand{\gluepbth}{p}       %

\newcommand{\algcancelendpoint} {\texttt{cancel-or-slide-endpoint}}
\newcommand{\alganticancel}     {\texttt{anticancel}}
\newcommand{\algmaxinterchange} {\texttt{max-interchange}}
\newcommand{\algmininterchange} {\texttt{min-interchange}}

\definecolor{blue-green}{rgb}{0.0, 0.87, 0.87}

\usepackage[colorinlistoftodos,prependcaption,textsize=tiny]{todonotes}

\newcommand{\ERCLogo}{%
\begin{wrapfigure}{r}{3.1cm}
        \centering
        \vspace{-0.75cm}
        \includegraphics[width=3cm]{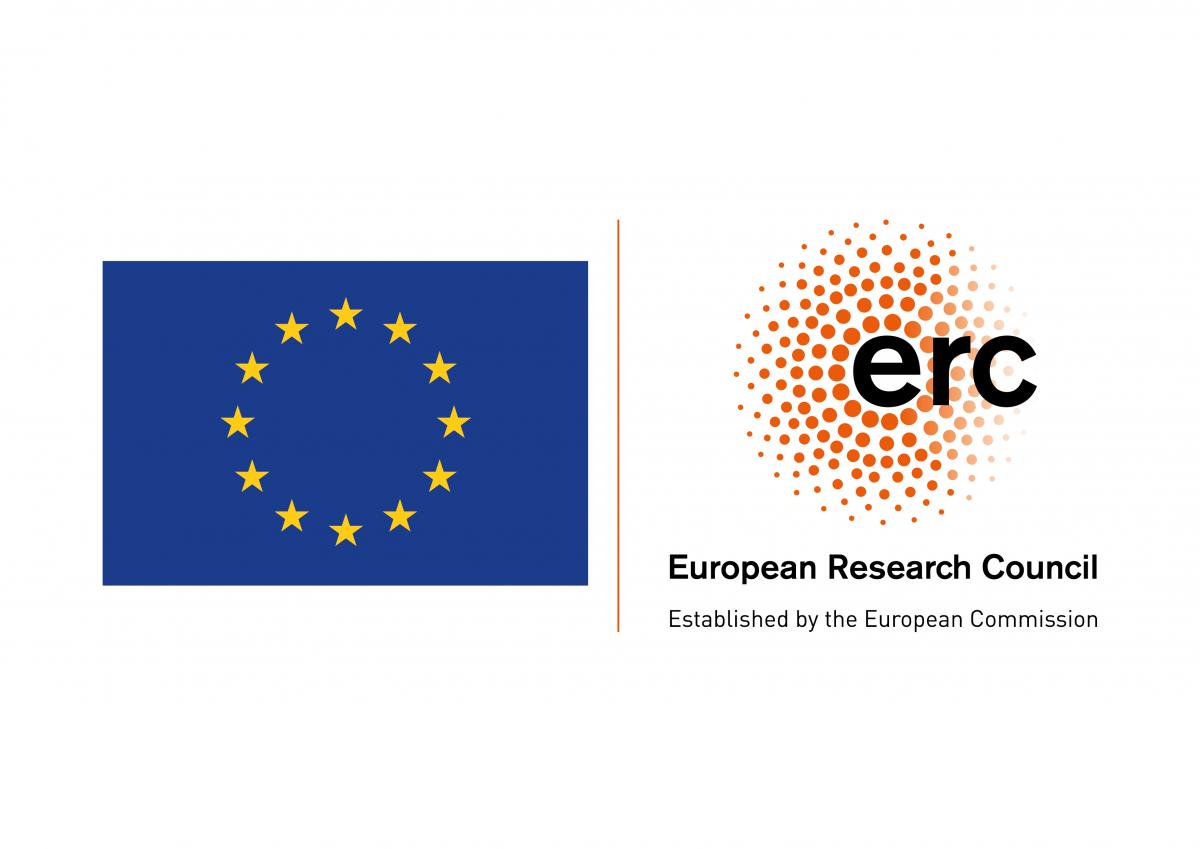}
        \vspace{-1cm}
\end{wrapfigure}
}

\newcounter{algcounter}
\newcounter{alglineno}[algcounter]
\newcommand{\clbegin}{\refstepcounter{alglineno}\scriptsize{\thealglineno}}
\newcommand{\clbeginlbl}[1]{\refstepcounter{alglineno}\label{#1}\scriptsize{\thealglineno}}
\newcommand{\clreset}{\refstepcounter{algcounter}}
\crefname{alglineno}{Line}{Lines}
\Crefname{alglineno}{Line}{Lines}

\title{Dynamically Maintaining the Persistent Homology of Time Series}

\author{
Sebastiano Cultrera di Montesano\thanks{ISTA (Institute of Science and Technology Austria), Klosterneuburg, Austria}
\and
Herbert Edelsbrunner\footnotemark[1]%
\and
Monika Henzinger\footnotemark[1]%
\and
Lara Ost\thanks{Department of Computer Science, University of Vienna, Vienna Austria}
}

\date{}

\newcommand{\thefunding}{
  \ERCLogo

  The first and second authors are funded by the  European Research Council under the European Union's Horizon 2020 research and innovation  programme, ERC grant no.~788183, ``Alpha Shape Theory Extended (Alpha)'',
  by the Wittgenstein Prize, FWF grant no.~Z~342-N31, 
  and by the DFG Collaborative Research Center TRR~109, 
  FWF grant no.~I~02979-N35.
  The third author %
  received funding by the European Research Council under the European Union's Horizon 2020 research and innovation programme, ERC grant no.~101019564, 
  ``The Design of Modern Fully Dynamic Data Structures (MoDynStruct)'',
  and by the Austrian Science Fund through the Wittgenstein Prize with FWF grant no.\ Z 422-N, and also by FWF grant no.~I~5982-N, and
  by FWF grant no.~P~33775-N, with additional funding from the \textit{netidee SCIENCE Stiftung}, 2020--2024. 
  The fourth author is funded by the Vienna Graduate School on Computational Optimization, FWF project no.~W1260-N35.
}

\begin{document}
\maketitle

\fancyfoot[R]{\scriptsize{Copyright \textcopyright\ 2024 by SIAM\\
Unauthorized reproduction of this article is prohibited}}

\begin{abstract}
  We present a dynamic data structure for maintaining the persistent homology of a time series of real numbers.
  The data structure supports local operations, including the insertion and deletion of an item and the cutting and concatenating of lists, each in time $O(\log n + k)$, in which $n$ counts the critical items and $k$ the changes in the augmented persistence diagram.
  To achieve this, we design a tailor-made tree structure with an unconventional representation, referred to as banana tree, which may be useful in its own right.
\end{abstract}

\section{Introduction}
\label{sec:1}

Persistent homology is an algebraic method aimed at the topological analysis of data; see e.g.\ \cite{Car09,EdHa10}.
It applies to low-dimensional geometric as well as to high-dimensional abstract data.
In a nutshell, the method is the embodiment of the idea that features exist on many scale levels, and rather than preferring one scale over another, it quantifies the features in terms of the range of scales during which they appear. 
More precisely, the goal of persistent homology is to compute a representation of the features in a book-keeping data structure, called the \emph{persistence diagram} \cite{EdHa10}.

\smallskip
Importantly, persistent homology has fast algorithms that support the application to large data sets.
Specifically, for arbitrary dimensional inputs,
persistence is generally computed 
by analyzing an underlying complex, with worst-case time cubic in the size of the complex. 
These algorithms are, however,  observed to run much faster in applications; see e.g.\ \cite{Ott17} for a survey on popular implementations.
We restrict ourselves to one-dimensional input data, i.e., a list of $m$ points (or \emph{items}) in an interval of $\mathbb{R}$ with each item $i$ being assigned a \emph{value} $f(i)$. 
Persistent homology has been applied to such time series data in multiple contexts, for example to heart-rate data \cite{CHLW21,GG21}, gene expression data \cite{Deq08,PDHH15}, and financial data \cite{GiKa18}.

The persistent homology of one-dimensional data can be derived from the merge tree \cite{SmMo19}, which records more detailed information about the structure of the persistence diagram, called the \emph{history of the connected components in the filtration of sublevel sets}.
Without recovering this history, the 
persistence information
can be computed in $O(m)$ time \cite{Gli23}.
To the best of our knowledge, the new $O(m)$ time algorithm in this paper is the first linear-time algorithm that can also recover the history, which we store in the \emph{augmented persistence diagram of the filtration of sublevel sets}.
This diagram is the extended persistence diagram of \cite{CEH09} together with a relation that encodes the merge tree.
It is defined formally in Section \ref{sec:2.3}.

\smallskip
As the data may change, it is an interesting question whether persistent homology can be maintained efficiently under update operations.
The historically first such algorithm \cite{CEM06} takes time linear in the complex size per swap in the ordering of the simplices; see also \cite{LuNe21,PiPe21}. 
For one-dimensional data this corresponds to a change of the value of an item, which reduces to a sequence of interchanges of $f$-values, each costing time linear in the size of the persistence diagram. 
More recently, \cite{DeHoPa23} has shown that this can be strengthened to logarithmic time if the input complex is a graph.
The current paper is the first to maintain the persistent homology of dynamically changing one-dimensional input with a tailor-made data structure under a larger suite of update operations, which includes the \emph{insertion} of a new item, the \emph{deletion} of an item, the \emph{adjustment} of the value of an item, the \emph{cutting} of a list of items into two, and the \emph{concatenation} of two lists into one.
The running time per operation is $O(\log n + k)$, in which $n$ is the current number of critical items, and $k$ is the number of changes to the augmented persistence diagram caused by the operation.

\smallskip
Our novel dynamic data structure is based on the characterization of the items in the persistence diagram through \emph{windows}, as recently established in \cite{BCES21}.
The main data structure is a \emph{binary tree}
ordered by position as well as value (see \cite{Vui80} for the introduction of such a binary tree),  a \emph{path-decomposition} of this tree dictated by persistent homology such that each path represents a window, and a final relaxation obtained by splitting each path into a \emph{left trail} of nodes with right children on the path and a \emph{right trail} of nodes with left children on the path.
This split of each path into two trails is crucial for our results, and without it, the update time would have a linear dependence on the depth of the tree, which might be $\Theta(n)$.

\medskip \noindent \textbf{Outline.}
Section~\ref{sec:2} provides the background needed for this paper: lists and maps, persistent homology, and the hierarchy of windows that characterizes the augmented persistent diagram.
In Section~\ref{sec:2A} we present a technical overview of our technique.
Section~\ref{sec:3} introduces the data structures we use to represent a linear list: a doubly-linked list, two dictionaries, and two path-decomposed ordered binary trees whose paths are stored as pairs of trails.
Section~\ref{sec:4} presents the algorithms for maintaining the augmented persistence diagram of a linear list.
Section~\ref{sec:5} concludes the paper.

\section{Background}
\label{sec:2}

This section explains how a linear list can be viewed as a continuous map on a closed interval.
Looking at the sub- and superlevel sets of such a map, we define its augmented persistence diagram and review the hierarchical characterization in terms of windows, as proved in \cite{BCES21}.
A modest amount of topology suffices to describe these diagrams, and we refer to \cite{EdHa10} for a more comprehensive treatment needed to place the results of this paper within a larger context.

\subsection{Lists Viewed as Maps.}
\label{sec:2.1}

By a \emph{linear list} we mean a finite sequence of real numbers, $c_1, c_2, \ldots, c_m$.
To view the list as a continuous map, we set $f(i) = c_i$, for $1 \leq i \leq m$, and linearly interpolate between consecutive values.
The result is a piece-wise linear map on a closed interval, $f \colon [1,m] \to \Rspace$, with \emph{items} $i$ and \emph{values} $c_i$, for $1 \leq i \leq m$.
Important about the items is their ordering and not their precise positions along the interval, so we use consecutive integers for convenience.
To simplify discussions, we will often assume that the map is \emph{generic}, by which we mean that its items have distinct values.
This is no loss of generality since a small perturbation may be simulated and implemented by appropriate tie-breaking rules.

\smallskip
Given $t \in \Rspace$, the \emph{sublevel set} of $f$ at $t$, denoted $f_t = f^{-1} (-\infty, t]$, is the set of points in domain of $f$ such that its value $f(x)$ is not larger than any fixed point $t \in \Rspace$.
Since $f$ is defined on a closed interval, the homology of $f_t$ is fully characterized by the number of connected components.
A point $x$ in the closed interval, $[1,m]$ is a \emph{(homological) critical point} of $f$ if the number of connected components of $f_t$ changes at the moment $t$ passes $f(x)$.
Assuming genericity, a critical point is necessarily an item of $f$, and the only two types in the interior of $[1,m]$ are \emph{minima} and \emph{maxima}.
When $t$ passes the value of a minimum from below, then the number of connected components of $f_t$ increases by $1$, and if $t$ passes the value of a maximum from below, the number of connected components decreases by $1$.
The endpoints of $[1,m]$ are special: the number of connected components changes at an \emph{up-type endpoint} and it remains unchanged at a \emph{down-type} endpoint.
Symmetrically, we call $f^t = f^{-1} [t,\infty)$ the \emph{superlevel set} of $f$ at $t$.
Observe that $f^t$ is the sublevel set of $-f$ at $-t$, and that the minima and maxima of $-f$ are the maxima and minima of $f$.
Similarly, the endpoints swap type.

\subsection{Extended Persistent Homology.}
\label{sec:2.2}

Persistent homology tracks the evolution of the connected components while the sublevel set of $f$ grows, and formally defines when a component is born and when it dies.
Complementing this with the same information for the superlevel sets of $f$, we get what is formally referred to as \emph{extended persistent homology}, which we explain next (see \cite{CEH09} for more details).

\begin{figure}[htb]
  \centering \vspace{0.1in}
  \resizebox{!}{1.45in}{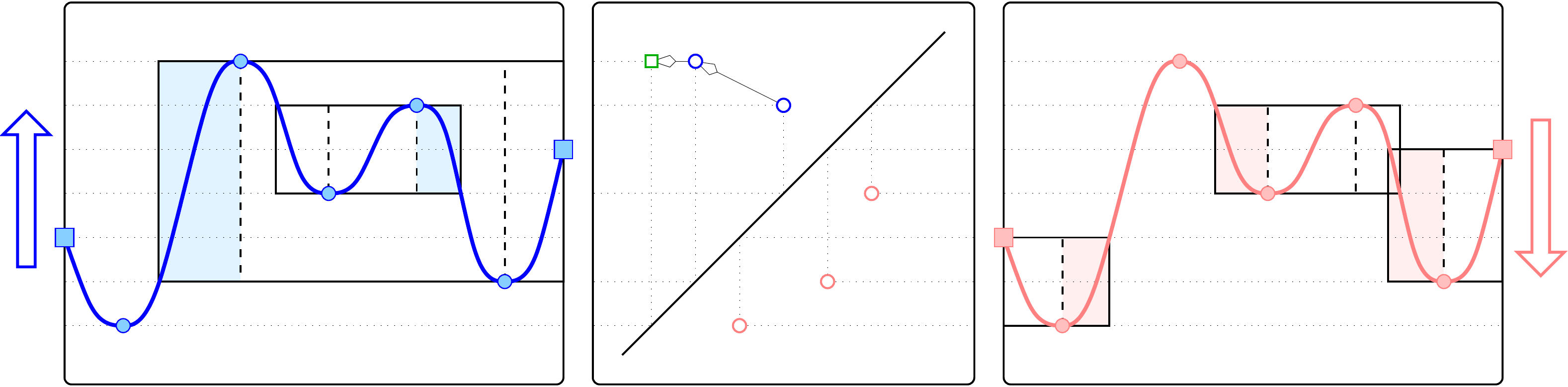}
  \caption{\emph{Left:} a real-valued map on a closed interval, $f$, with three minima and two maxima.
  \emph{Right:} the map $-f$ drawn upside-down.
  \emph{Middle:} the augmented persistence diagram with two (\emph{blue}) points in the ordinary subdiagram $\Ord{}{f}$ above the diagonal, three (\emph{pink}) points in the relative subdiagram $\Rel{}{f}$ below the diagonal, and one (\emph{green}) point in the essential subdiagram $\Ess{}{f}$.}
  \label{fig:diagram}
\end{figure}

\smallskip
In \emph{Phase~One}, we track the connected components of the sublevel set, $f_t$, as $t$ increases from $-\infty$ to $\infty$.
A component is \emph{born} at the smallest value of $t$ at which a point of the component belongs to $f_t$.
This point is necessarily a minimum in the interior of $[1,m]$, or an up-type endpoint. 
The component \emph{dies} when it merges with another component that was born earlier.
The point at which the two components merge is necessarily a maximum in the interior of $[1,m]$.
The \emph{ordinary subdiagram} of $f$, denoted $\Ord{}{f}$, records the birth and death of every component with a point in the plane whose abscissa and ordinate are those values of $t$ at which the component is born and dies, respectively; see Figure~\ref{fig:diagram}.

\smallskip
In \emph{Phase~Two}, we track the connected components of the superlevel set, $f^t$, as $t$ decreases from $\infty$ to $- \infty$.
Birth and death are defined accordingly, and the components are recorded in the \emph{relative subdiagram}, denoted $\Rel{}{f}$.
By construction, the points in $\Ord{}{f}$ lie above and those of $\Rel{}{f}$ lie below the diagonal; see Figure~\ref{fig:diagram}.
The component born at the global minimum of $f$ is special because it does not die during Phase~One.
Instead, it dies at the global minimum of $-f$, which is the global maximum of $f$.
In topological terms, this happens because the one connected component still alive at the beginning of Phase~Two dies in relative homology when its first point enters the superlevel set.\footnote{In the original formulation in \cite{CEH09}, Phase~Two in the construction of the (extended) persistence diagram tracks the $1$-dimensional relative homology groups of the pairs $(f_\infty, f^t)$, which are isomorphic to $0$-dimensional homology groups of the $f^t$.
To appreciate the guiding hand of algebraic topology, we need to understand how the absolute and relative homology groups of different dimensions relate to each other.
However, for the purpose of this paper, this is not necessary and we can track the connected components of the superlevel set in lieu of the relative cycles in the interval modulo the superlevel set.}
This class is represented by the sole point in the \emph{essential subdiagram}, denoted $\Ess{}{f}$; see again Figure~\ref{fig:diagram}.
The \emph{extended persistence diagram}, denoted $\Dgm{}{f}$, is the disjoint union of the three subdiagrams.
Hence, $\Dgm{}{f}$ is a multi-set of points in $\Rspace^2$, and so are the three subdiagrams, unless the map is generic, in which case the diagram is a set.

\subsection{Windows and Augmented Persistence Diagram.}
\label{sec:2.3}

The points of the persistence diagram can be characterized using the concept of windows recently introduced in \cite{BCES21}.
Let $a$ and $b$ be two (homological) critical points of $f$, with values $A = f(a) < f(b) = B$.
Note that $f^{-1} [A,B]$ contains all points in the interval whose values lie between $A$ and $B$.
It consists of one or more connected components, and it is quite possible that $a$ and $b$ belong to different components.
However, if they belong to the same component, then we denote by $[x,y]$ the connected component of $f^{-1} [A,B]$ that contains both $a$ and $b$
and we call
$[x,y] \times [A,B]$  the \emph{frame} with \emph{support} $[x,y]$ spanned by $a$ and $b$.
There are two orientations of the frame: from left to right if $a < b$, and from right to left if $b < a$.
In the former case, we call $x$ the \emph{mirror} of $b$ and $\Window{a}{b}$ a \emph{triple-panel window} if $f(y) = A$.
In the latter case, we call $y$ the \emph{mirror} of $b$ and $\Window{a}{b}$ a \emph{triple-panel window} if $f(x) = A$.
We say $a$ and $b$ \emph{span} $\Window{a}{b}$.
We distinguish a special type of window, referred to as a \emph{global window}, whose support covers the entire interval, $[x,y] = [1,m]$, and which is exempt of the requirement at $x$ or $y$.

\smallskip
In \cite{BCES21}, the authors prove that there is a bijection between the windows and the points in $\Dgm{}{f}$.
Furthermore, each critical point in the interior of the interval spans exactly one window in each phase---even though the pairing may be different in the two phases---while each endpoint spans a window in only one phase.
In Figure~\ref{fig:diagram}, there are five triple-panel windows, two on the left and three on the right.
There is always exactly one global window spanned by the global minimum, $\alpha$, and the global maximum, $\beta$, of $f$.
What follows is a special case of a more general characterization of persistence proved in \cite{BCES21}.
\begin{proposition}[Persistence in Terms of Windows \cite{BCES21}]
  \label{prop:persistence_in_terms_of_windows}
  Let $f \colon [1,m] \to \Rspace$ be a generic piecewise linear map on a closed interval, and let $a, b$ be homological critical points of $f$, or of $-f$, with $f(a) = A$ and $f(b) = B$.
  Then
  \smallskip \begin{enumerate}[(i)]
    \item $(A, B) \in \Ord{}{f}$ iff $\Window{a}{b}$ is a triple-panel window of $f$,
    \item $(B, A) \in \Rel{}{f}$ iff $\Window{b}{a}$ is a triple-panel window of $-f$,
    \item $(A, B) \in \Ess{}{f}$ iff $\Window{a}{b}$ is the global window of $f$.
  \end{enumerate} \medskip
\end{proposition}
\begin{figure}[htb]
  \centering \vspace{0.0in}
  \resizebox{!}{1.8in}{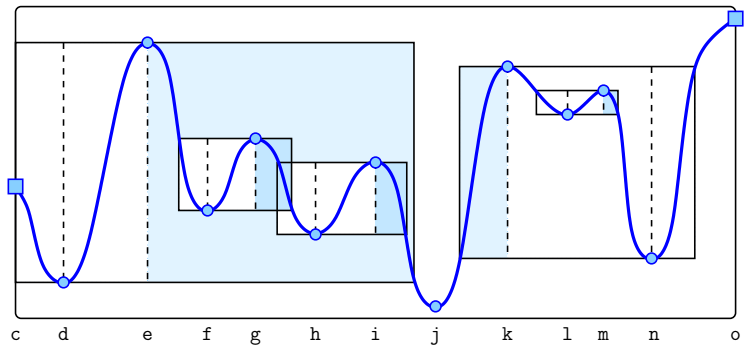}
  \caption{The graph of a generic map on a closed interval.
  All windows shown are with simple wave, except for the \emph{leftmost} window, whose wave is short.
  The global window as well as the (tiny) windows caused by the hooks (which will be introduced in \cref{sec:3}) are not shown.
  The \emph{light-blue shaded} out-panels are part of the triple- but not of the double-panel windows.}
  \label{fig:maponly}
\end{figure}
Suppose $\Window{a}{b}$ is a triple-panel window, with support $[x,y]$.
Unless $a$ is an endpoint of $[1,m]$, cutting $[x,y]$ at $a$ and $b$ produces three segments, $[x,a]$, $[a,b]$, $[b,y]$, in which we assume $a < b$.
We call the corresponding products with $[A,B]$ the \emph{in-panel}, \emph{mid-panel}, \emph{out-panel},
and the graph within the window a \emph{wave}.
This wave is \emph{simple} or \emph{short} depending on whether the function value at the mirror is equal to or smaller that at the maximum: $f(x) = B$ or $f(x) < B$.
In Figure~\ref{fig:maponly}, we have four simple waves (defined by $\T{f}, \T{g}$, by $\T{h}, \T{i}$, by $\T{l}, \T{m}$, and by $\T{n}, \T{k}$) and one short wave (defined by $\T{d}, \T{e}$).
Note that a simple wave of $f$ is also a simple wave of $-f$, albeit upside-down.
A similar statement does not hold for short waves, which explains the violations of symmetry in the persistence diagram; see again Figure~\ref{fig:diagram}.

\smallskip
The \emph{double-panel} version of a triple-panel window consists of the in-panel and the mid-panel but drops the out-panel.
Any two double-panel windows of $f$ are either disjoint or nested, and all these windows are nested inside the global window \cite[Lemma~3.3]{BCES21}.
We use this property to augment the persistence diagram with the merge tree information; that is: we draw an arrow between two points if the corresponding double-panel windows are nested without any other window nested between them.
The result is the \emph{augmented persistence diagram}, denoted $\DgmR{}{f}$. 

\smallskip
In this paper, we consider operations that maintain the augmented persistence diagram to reflect the change from a map, $f$, to any other map, $g$.
We quantify this difference by the number of points and arrows that change or, more formally, belong to the symmetric difference between the two diagrams.
What are typical differences?
Every local change in the function reduces to a sequence of \emph{interchanges} between two minima or two maxima, possibly ending in a \emph{cancellation} or beginning with an \emph{anti-cancellation}.
Every cancellation removes a point from the diagram, and every anti-cancellation adds a point to the diagram, so the diagrams before and after differ by a constant amount.
An interchange may or may not swap two critical values between points (two coordinates, which are values of two critical items) or the reversal of an arrow, and if it does not, then this should not incur any cost.
On the other hand, if it does, then this can cost a constant amount of time.
Our declared goal is to have operations that run in time $O(\log n + k)$, in which $n$ is the number of critical points and $k$ measures the difference between the augmented persistence diagrams before and after the operations.

\section{Technical Overview}
\label{sec:2A}

To achieve this goal, we propose a novel collection of data structures---some classic and some new---for maintaining nested windows.
We show that using these data structures allows us to maintain the augmented persistence diagram to reflect the change from one map to another in the desired time bound.
We maintain the following data structures:
\smallskip \begin{enumerate}
  \item[(1)] a \emph{doubly-linked list} of all items, critical or not, ordered by their positions in the interval;
  \item[(2)] two \emph{binary search trees}, called \emph{dictionaries}, one storing the minima and the other storing the maxima, both ordered by their positions in the interval;
  \item[(3)] two \emph{banana trees} (described next) representing the information in the augmented persistence diagram by storing the minima and maxima while reflecting their ordering by position as well as by function value.
\end{enumerate} \smallskip
Note that the minima and maxima are subsets of all items and are therefore represented in all of the above data structures.
To reduce the special cases in the algorithms, we add two artificial items, called \emph{hooks}, at the very beginning and the very end of the interval.
They make sure that the formerly first and last items are proper minima or maxima, but we ignore them and the technicalities involved in this overview and give slightly simplified descriptions of the data structures and the algorithms.

\medskip \noindent \textbf{Banana Trees.}
We introduce these trees in three stages.
In the \emph{first stage}, we organize all windows in a full binary tree, whose leaves are the minima and whose internal nodes are the maxima, such that (i) the in-order traversal of the tree visits the nodes in increasing order of their positions in the interval, and (ii) the nodes along any path from a leaf to the root are ordered by increasing function value.
Such a tree always exists and it is unique.
In the following, we do not distinguish between a node in the tree and the critical item it represents.

\smallskip
In the \emph{second stage}, we add a \emph{special root} labeled with value larger than the global maximum whose only child is the previous root.
Call the resulting tree $T$ and note that it has an equal number of leaves and internal nodes.
We then form paths, each starting at a leaf, $a$, and ending at an internal node, $b$, such that $a$ and $b$ span a window, $\Window{a}{b}$.
Call this path $\Path{a}{b}$.
Based on structural properties of $T$, this leads to a partition of $T$ into edge-disjoint (but not vertex-disjoint) paths; see the left drawing in Figure~\ref{fig:bananatree}.
The node $b$ ending the path that starts at $a$ is locally determined: it is the first internal node encountered while walking up from $a$, such that $a$ and the descending leaf with minimum function value lie on different sides (in different subtrees) of this internal node.
Hence, every maximum on the path from $a$ to $b$ spans a window that is immediately nested in $\Window{a}{b}$.
It follows that every maximum, $b$, belongs to two paths: $\Path{a}{b}$ and $\Path{p}{q}$ such that $\Window{a}{b}$ is immediately nested in $\Window{p}{q}$.
This even holds for the root (of the full binary tree), which spans a path and also belongs to the path that ends at the special root.
Given a map, $f$, the tree, $T$, and its partition into paths are unique.
The strict dependence on $f$ may force $T$ to be unbalanced, and indeed have linear depth, so that efficient maintenance algorithms are challenging.
This is why we need another modification.
\begin{figure}%
  \centering \vspace{0.0in}
  \resizebox{!}{2.2in}{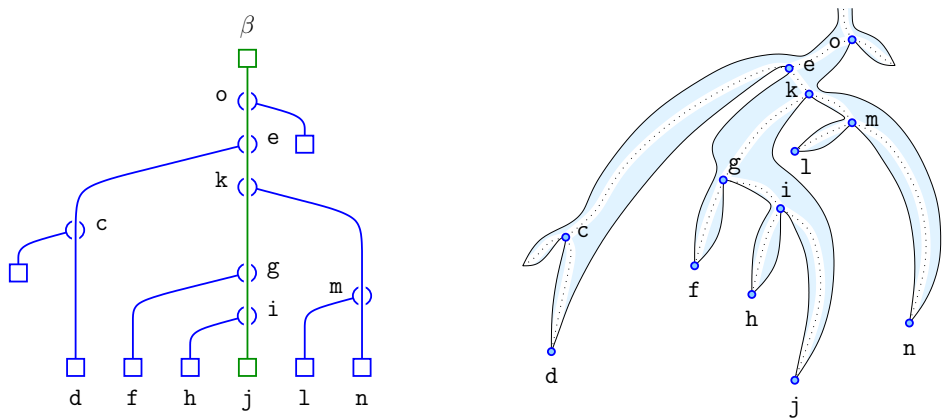}
  \caption{The path-decomposed binary tree associated to the map in Figure \ref{fig:maponly} with special root, $\beta$, on the \emph{left}, and the corresponding banana tree on the \emph{right}.}
  \label{fig:bananatree}
\end{figure}

\smallskip
In the \emph{third stage}, we split each path into two trails; see the right drawing in Figure~\ref{fig:bananatree}.
The \emph{left trail} of $\Path{a}{b}$ contains $a$ and every maximum $u$ on $\Path{a}{b}$ with $u \leq a$, while the \emph{right trail} contains $a$ and every maximum $u$ on $\Path{a}{b}$ with $a \leq u$.
A node $v$ on $\Path{a}{b}$ is the right child of its parent, $u$, in the binary tree iff $u$ is on the left trail.
Furthermore, $v$ is the left child of $u$ iff $u$ is on the right trail.
Thus, to which trail $u$ belongs to can be decided based on local information at $u$.
To simplify language, we also give a second name to the trails.
If $a < b$, then $\Window{a}{b}$ consists of the in-panel on the left, the mid-panel in the middle, and the out-panel on the right.
In this case, $u$ belongs to the left trail iff the window it spans is nested inside the in-panel, so we alternatively call the left trail the \emph{in-trail}, and we alternatively call the right trail the \emph{mid-trail}.
If $b<a$, then the right trail is the in-trail and the left trail is the mid-trail.
So what happened to the out-trail?
It is indeed needed, but only if $\Window{a}{b}$ is with simple wave.
Such a window corresponds to $\Path{a}{b}$ in the banana tree of $f$ and to $\Path{b}{a}$ in the banana tree of $-f$.
The out-panel of $\Window{a}{b}$ is the in-panel of $\Window{b}{a}$, so we can get information about the out-panel from the in-panel of the same window in the other banana tree when we need it.

\smallskip
To speed up the algorithms, we maintain various pointers, such as between the occurrences of the same critical item in the banana trees, the dictionaries, and the doubly-linked list, for example.
Importantly, every internal node, $b$, stores a pointer to the descending leaf with minimum function value, $\Low{b}$, and every leaf, $a$, stores a pointer to the endpoint of its path, $\Dth{a}$.
Observe that $q = \Dth{\Low{b}}$ is the maximum that spans the window in which $\Window{a}{b}$ is immediately nested, so this window can be obtained in constant time.

\medskip \noindent \textbf{Construction.}
While the main focus of this paper is the maintenance of the data structures through local updates, we also consider the construction from scratch.
It is straightforward to derive the augmented persistence diagram during a single traversal of the banana trees in linear time, so the question we study is how fast this diagram can be constructed from a given 
one-dimensional input list.
Assuming all non-critical items have been removed and we are given the remaining sequence of $n$ critical items, there are standard algorithms that can be adapted to construct the banana trees in $O(n \log n)$ time.
There is also an $O(n)$ time algorithm for computing the persistence diagram \cite{Gli23}, but this algorithm does not extend to the augmented persistence diagram.
To the best of our knowledge  our algorithm is the first to construct the augmented persistence diagram in $O(n)$ time.

\smallskip
The main structure of the algorithm is a left-to-right scan of the data.
We interpret the item $i$ with value $c_i = f(i)$ as the point $(i, c_i)$ in the plane and maintain a decreasing staircase such that all processed items are points on or below the staircase.
Each step of the staircase corresponds to an unfinished banana.
One of the difficulties is that before a banana is completed, we do not know whether it will be attached to a left or a right trail.
We tentatively assume it will be attached to a left trail but are prepared to move the banana to the other side when this turns out to be necessary.
When we process the next item, we may remove any number of steps, turning each into a finished banana, but we can add at most one new step.
Since a step that is removed was added earlier, this proves that the algorithm runs in $O(1)$ amortized time per item, and therefore in $O(n)$ time altogether.

\medskip \noindent \textbf{Local Maintenance.}
Given a list of $m$ items with real function values, we consider the operations that \emph{insert} a new item, \emph{delete} an item, and \emph{change the value} of an item.
All three operations reduce to a sequence of \emph{interchanges}---which can be between two maxima or between two minima---possibly preceded by an \emph{anti-cancellation} or a \emph{slide}, and possibly succeeded by a \emph{cancellation} or a \emph{slide}.
In a slide, a minimum or maximum next to a non-critical item becomes non-critical, and the non-critical item becomes a minimum or maximum, respectively.
Similarly in a cancellation, a minimum and a neighboring maximum simultaneously become non-critical.
Here we will focus on the interchanges, because they are most common as well as most interesting, and on the anti-cancellations, because they pose an unexpected challenge.

\smallskip
Consider two maxima, $b$ and $q$, of $f$, and assume $f(b) < f(q)$ before the interchange.
To avoid confusing language, we write $g(b)$ and $g(q)$ for the values after the operation but assume that $f$ and $g$ agree on all items except for $b$.
Furthermore, we assume that $g(b) > g(q)$ and that $b$ and $q$ are the only two items for which the ordering by $f$-value differs from the ordering by $g$-value.
In many cases, the interchange of $b$ and $q$ does not affect the banana trees.
Indeed, only if $b$ is a child of $q$ is it necessary to update the order of the two nodes.
And even if $b$ and $q$ are consecutive maxima on a path, there is no structural change unless $b$ and $q$ also belong to a common trail.
This is the main reason for splitting each path into two trails as explained in the third stage of the introduction of the banana trees: to avoid any cost to occur for interchanges that have no structural affect on the augmented persistence diagram.
When $b$ and $q$ interchange while belonging to different trails, then the banana tree is oblivious to this change and requires no update.
On the other hand, if $b$ and $q$ belong to the same trail, then they swap positions along this trail, and there is a change of the augmented persistence diagram that pays for the time it takes to update the banana tree.

\smallskip
The interchange of two minima, $a$ and $p$, is quite different because it does not affect the ordered binary tree at the first stage of the banana tree.
However, the interchange affects the path-decomposition, so some of the bananas may have to be updated.
To be specific, assume $f(a) > f(p)$ before and $g(a) < g(p)$ after the operation.
As before, we also assume that $f$ and $g$ agree on all items except for $a$, and that $a$ and $p$ are the only two items for which the ordering by $f$-value differs from the ordering by $g$-value.
Let $b$ and $q$ be the internal nodes so that $\Path{a}{b}$ and $\Path{p}{q}$ are paths in the decomposition of the banana tree of $f$.
The interchange of $a$ and $p$ has no effect on the banana tree, unless $b$ is a node on $\Path{p}{q}$.
If $b$ lies on $\Path{p}{q}$, then we extend $\Path{a}{b}$ to $\Path{a}{q}$ and we shorten $\Path{p}{q}$ to $\Path{p}{b}$.
The nodes $u$ on the path from $b$ to the child of $q$ change their pointer from $\Low{u} = p$ to $\Low{u} = a$.
There can be arbitrarily many such nodes, but each change causes the adjustment of an arrow in the extended persistence diagram, to which it can be charged in the running time analysis.

\smallskip
This shows that it is possible to perform an interchange of two minima within the desired time bound, but it is not clear how to find them.
Considering the scenario in which $a$ decreases its value continuously, it may cause a sequence of interchanges with other minima, but since these minima are not sorted by function value, it is not clear how to find them, and how to ignore the ones without structural consequences.
Here is where the relation between the banana trees of $f$ and $-f$ becomes important.
The minima of $f$ are the maxima of $-f$, so the interchange of two minima in the banana tree of $f$ corresponds to the interchange of two maxima in the banana tree of $-f$.
We already know how to find the interchanges of two maxima and how to ignore the ones without structural consequences, so we use them to identify the interchanges of minima.
More precisely, we prove that the interchange of the minima, $a$ and $p$ of $f$, affects the structure of the banana tree of $f$ only if the interchange of the maxima, $a$ and $p$ of $-f$, affect the structure of the banana tree of $-f$.
The converse does not hold, but the implication suffices since the interchange of the maxima of $-f$ leads to a change in the extended persistence diagram which can be charged for the interchange of the minima, which costs only $O(1)$ time if there are no structural adjustments.

\medskip
Next, we sketch what happens in the remaining operations.
A \emph{slide} occurs if a maximum decreases its value so that it becomes non-critical, while a neighboring non-critical item becomes a maximum.
However, if this neighboring item is a maximum, then both items become non-critical at the same time, in which case the operation is called a \emph{cancellation}.
There are also the symmetric operations in which a minimum increases its value and becomes non-critical, while a neighboring item changes from non-critical to minimum (a \emph{slide}) or from maximum to non-critical (a \emph{cancellation}).
The corresponding updates are easily performed within the required time bounds.

\smallskip
A more delicate operation is the \emph{anti-cancellation}, in which two neighboring non-critical items become critical at the same moment in time.
Let $a$ and $b$ be these two items and assume $a$ is a minimum and $b$ is a maximum after the operation, so $g(a) < g(b)$ and, by assumption, $f(a') > f(a) > f(b) > f(b')$, in which $a', a, b, b'$ are four consecutive items in the list.
Hence, $a$ and $b$ are not present in the banana tree of $f$, but $\Path{a}{b}$ is a path in the path-decomposed tree of $g$, so $a, b$ form a minimal banana in the banana tree of $g$.
All we need to do is find the correct place to attach this banana, but this turns out to be difficult.
We first explain why it is difficult, then present an algorithm that works within the current data structure, and finally sketch a modification of the data structure that accelerates the anti-cancellation to $O(\log n)$ time.
While this leads to a speed-up in the worst case for this type of operation, it slows down other operations by a logarithmic factor.

We use mirrors to explain in what situation an anti-cancellation is difficult.
Their representation in the banana tree is indirect: the maximum is the upper end of a mid-trail, and its mirror is the upper end of the matching in-trail.
In the tree, the two upper ends are the same node, but if we traverse the trails of the banana trees in sequence, we visit them at different times, and these times correspond to the positions along the interval, which are different for the maximum and its mirror.
Every maximum has at most one mirror, so adding all mirrors to the list increases it to less than double the original size.
Nevertheless, it is easy to construct a case in which there is a long subsequence of mirrors, and these mirrors are consecutive with the exception of $a$ and $b$ appearing somewhere in their midst.
In this situation, finding the correct attachment of $\Path{a}{b}$ in the banana tree of $g$ needs the mirrors immediately to the left and right of $b$.
The data structure as described provides no fast search mechanisms among the mirrors, so we find the correct attachment by linear search starting from the first maximum to the left of $a$.
Suppose we pass $k$ mirrors before we find this attachment.
Then we have a nested sequence of $k$ windows, and $\Window{a}{b}$ is nested inside all of them.
Hence, there is one new immediate nesting pair, while the transitive closure of the nesting relation gains $k$ new pairs.
The cost of the anti-cancellation can be charged to the change of this transitive closure.
In many cases, the change in the transitive closure will be comparable in size to that of the transitive reduction, but it can also be significantly larger, like in the case we just described.
We thus entertain alternatives.

There is a modification of the data structure that allows for fast searches among subsequences of mirrors:  for any consecutive min-max pair in the list, store the mirrors between them in a binary search tree.
In practice, there will be many very small such trees, but it is possible that the mirrors accumulate and produce a few large trees.
With this modification, an anti-cancellation can be performed in $O(\log n)$ time.
Note however, that these binary search trees have to be maintained, which adds a factor of $O(\log n)$ to the time of every operation, in particular to every interchange, of which there can be many.

\medskip \noindent
\textbf{Topological Maintenance.}
We call an operation \emph{topological} if it cuts the list of points into two, or it concatenates two lists into one.
More challenging topological operations, such as gluing the two ends of a list to form a cyclic list, or gluing several lists to form a geometric tree or more complicated geometric network are feasible but beyond the scope of this paper.

\smallskip
When we \emph{cut} the list of data into two, we \emph{split} the banana trees of the map and its negative into two each.
We mention both banana trees since we need information from the other to be able to split one within the desired time bound.
Splitting one banana tree is superficially similar to splitting a binary search tree, but more involved.
Let $f \colon [1,m] \to \Rspace$ be the map before the operation and $g \colon [0,\ell] \to \Rspace$ and $h \colon [\ell+1, m] \to \Rspace$ the maps after the operation.
We write $z = \ell + \frac{1}{2}$ for the position at which the time series is cut.
Since the banana trees are determined by the maps, we need to understand the difference between the windows of $f$ and those of $g$ and $h$.
Whether or not a frame of $f$ is also a window depends solely on the restriction of $f$ to the support of the frame.
To decide about the future of a window of $f$, let $[x,y]$ be its support.
This window is also a window of $g$ if $y < z$, and it is also a window of $h$ if $z < x$.
The windows that need attention are the ones with $x < z < y$.
A triple-panel window consists of three panels, so we distinguish between three cases:
\smallskip \begin{itemize}
  \item $\Window{a}{b}$ suffers an \emph{injury} if $z$ cuts through the in-panel.
  Then $a$ and $b$ lie on the same side of $z$ and they still span a window, albeit with short wave.
  \item $\Window{a}{b}$ suffers a \emph{fatality} if $z$ cuts through the mid-panel.
  Then $a$ and $b$ lie on different sides of $z$ and need to find new partnering critical items to span new windows.
  \item $\Window{a}{b}$ suffers a \emph{scare} if $z$ cuts through the out-panel.
  These windows are difficult to find in the banana tree of $f$, but they are easy to find in the banana tree of $-f$, in which $\Window{b}{a}$ suffers an injury.
\end{itemize} \smallskip
The splitting of the banana tree proceeds in three steps: first, find the smallest banana that suffers an injury, fatality, or scare; second, find the remaining such bananas; and third, split the banana tree of $f$ into the banana trees of $g$ and $h$.
We address all three steps and highlight the most interesting feature in each.

\smallskip
The first step is difficult because of the lack of an appropriate search mechanism in the banana tree.
To explain this, we recall that the banana tree stores the mirrors implicitly, as the upper ends of the in-trails.
Like in the case of an anti-cancellation, the challenging case is when mirrors accumulate and we have to locate $z$ in their midst. 
As before, we locate $z$ by linear search, scanning the mirrors in the order of decreasing function value.
In contrast to the anti-cancellation, we can now charge the cost for the search to the changes in the extended persistence diagram.
Indeed, every mirror we pass belongs to a window that experiences a scare.
Such a window of $f$ is neither a window of $g$ nor of $h$, so its spanning critical items will re-pair and span new windows after the operation.

The second step traverses a path upward from the smallest affected banana we just identified.
In each step of the traversal, we determine the corresponding window and push it onto the stacks of windows that experience an injury, fatality, or scare.
A window that experiences an injury remains a window, but now with short instead of simple wave.
Whether this window with short wave belongs to the banana tree of $g$ or that of $h$ depends on whether the spanning minimum is to the right or the left of the spanning maximum.
A window that experiences a fatality falls apart, with one of the two spanning critical items in $g$ and the other in $h$.
Finally, a window that experiences a scare stops to be a window, in spite of having both critical points in $g$ or in $h$.
As mentioned earlier, such a window is difficult to find in the banana tree of $f$, but it is easy to find in the banana tree of $-f$, where it experiences an injury.
We thus process both banana trees simultaneously, and distributed the windows or their corresponding bananas as needed.
The upward traversal halts when we reach the \emph{spine} of the tree, by which we mean the sequence of left children that start at the root, or the sequence of right children that start at the  root.
This is important because moving further until we reach the root is not necessary and can be costly because the spine can be arbitrarily long and the steps towards the root cannot be charged to changes in the extended persistence diagram.

The third step re-pairs the critical items of the windows that experience a fatality or scare.
All new windows are with short wave, for else they would be windows of $f$ before the splitting, which is a contradiction.
Hence, their bananas belong to the spines of the banana trees of $g$ and $h$.
The bananas in the spines have the particularly simple structure that their critical items come in sequence.
For the right spine of the banana tree of $g$, this sequence increases from right to left, for the left spine of the banana tree of $h$, the sequence increases from left to right, and both are consistent with the sequence in which we collect the corresponding windows in the second step.
It follows that the available critical items can be paired up in sequence, which takes $O(1)$ time per item.

\smallskip
Corresponding to the concatenation of two lists, we pairwise \emph{glue} the banana trees of the two maps.
The operation is the inverse of splitting, so we omit further details.
A final word about the cost paid by the changes in the augmented persistence diagram.
How do we compare one diagram with two, which we get after splitting?
To deal with this issue, we consider the maps $g$ and $h$ to be \emph{one} map, namely a map on two intervals.
Then we still have one augmented persistence diagram and the symmetric difference between the diagrams before and after the operation is well defined.

\section{Data Structures}
\label{sec:3}

\medskip
We use five inter-connected data structures to represent a linear list or, equivalently, the implied piecewise linear map: a doubly-linked list of all items, two dictionaries storing the minima and maxima for quick access, and two ordered trees representing the information in the augmented persistence diagram for quick update.
The novel aspect is the implementation of the ordered trees as \emph{banana trees}, to be described shortly.
There are two such trees, storing the homological critical points of the map and its negation, which are subsets of all items.
Indeed, the non-critical items are stored only in the doubly-linked list, but they enter the dictionaries and banana trees when they become critical.
The difference between critical and non-critical items is defined in terms of the piecewise linear map implied by the items.
Let $m \geq 2$ and write $c_1, c_2, \ldots, c_m$ for the list of values, which we assume are distinct.
We construct the map, $f \colon [0,m+1] \to \Rspace$, by setting $f(i) = c_i$, for $1 \leq i \leq m$, and adding artifical ends by setting
\begin{align}
    f(0)  &=  c_1 + \ee \cdot \sign{c_2 - c_1} , \\
    f(m+1)  &=  c_{m} + \ee \cdot \sign{c_{m-1} - c_{m}} ,
\end{align}
in which $\sign{c} = \pm 1$ depending on whether $c > 0$ or $c < 0$, and $\ee > 0$ is arbitrarily small and in any case smaller than the absolute difference between any two of the $c_i$.
We call these artificial ends \emph{hooks}; they make sure that items $1$ and $m$ are proper minima or maxima.

\subsection{Dictionaries.}
\label{sec:3.1}

We maintain the items in a \emph{doubly-linked list} ordered according to their position in  the interval.
For reasons that will become clear shortly, we additionally store the minima and maxima in a \emph{dictionary} each, both ordered like the doubly-linked list.
Besides searching, the dictionaries support the retrieval of the minimum or maximum immediately to the left or the right of a given $x \in \Rspace$.
It will be convenient to write $n$ for the number of maxima so that $n-1$, $n$, or $n+1$ is the number of minima.
In addition, we will pretend that the items---and sometimes just the minima and maxima---are at consecutive integer locations along the real line.
Obviously, this cannot be maintained as we insert and delete items, but it makes sense locally and simplifies the notation and discussions without causing any confusion.

\smallskip
An opportune data structure for the dictionary is a binary search tree, which supports \texttt{access}, \texttt{insertion}, and \texttt{deletion} of an item in $O(\log n)$ time each.
Similarly it supports the \texttt{cutting} of a dictionary into two, and the \texttt{concatenation} of two dictionaries into one---provided all items in one dictionary are smaller than all items in the other---again in $O(\log n)$ time.
For comparison, doubly-linked lists can be \texttt{cut} and \texttt{concatenated} in constant time, provided the nodes where the cutting and concatenation is to happen are given.

\subsection{Banana Trees.}
\label{sec:3.2}

Let $f \colon [0,m+1] \to \Rspace$ be a generic piecewise linear map constructed from a list of $m$ distinct values.
The main data structure consists of two trees, one for $f$ and the other for $-f$, which we explain in three steps.
In order to ensure that all critical items are represented in both trees, we require homological critical points of $f$ to also be homological critical points of $-f$, and vice versa.
This affects how we treat up- and down-type items: up-type items in $f$ are already homological critical points; in $-f$, however, they are down-type items, which are not homological critical points.
We use the hooks at down-type items to treat them as proper maxima, and ignore the hooks at up-type when constructing the trees.
To describe the tree for $f$, let $a_0 < b_1 < a_1 < \ldots < b_n < a_n$ be the homological critical points of $f$, and note that the $a_i$ are
minima---with the possible exception of $a_0$ and $a_n$, which may be up-type endpoints that are artificially added as hooks.
On the other hand, all $b_i$ are proper maxima.

\medskip \noindent \textbf{First step:}
we construct a full binary tree whose nodes are the homological critical values.
Among the many choices, we arrange these values such that
\smallskip \begin{description}
  \item[{\sf I.1}] the in-order sequence of the nodes in the tree is the ordering of the homological critical points in the linear list, i.e.\ of their position in the interval;
  \item[{\sf I.2}] the nodes along any path from a leaf to the root are ordered by increasing values.
\end{description} \smallskip
It is not difficult but important to see that there is a unique full binary tree that satisfies Conditions~{\sf I.1} and {\sf I.2}.
Indeed, the largest value is a maximum, $b_j$, which is necessarily the root of the tree.
The left subtree is the recursively defined tree of $a_0, b_1, \ldots, a_{j-1}$, and the right subtree is the recursively defined tree of $a_j, b_{j+1}, \ldots, a_n$.
By induction, the two subtrees are unique, so the entire tree is unique.

\medskip \noindent \textbf{Second step:}
we decompose the tree into edge-disjoint paths, each connecting a leaf to an internal node.
Since there is one extra leaf, we add a \emph{special root}, $\beta$, whose only child is the root, as an extra internal node.
We define $\beta$ to be greater than all items, both in terms of function value and along the interval; that is:
$f(i) < f(\beta)$ and $i < \beta$ for all items $i$.
With this small change, the paths define a bijection between the leaves and the internal nodes.
Again there are many choices, and we pick the paths such that
\smallskip \begin{description}
  \item[{\sf II}] each path connects a leaf, $a$, to the nearest ancestor, $b$, for which $a$ does not have the smallest value in the subtree of $b$, and to the special root, if no such ancestor exists.
\end{description} \smallskip
After fixing the tree in the first step, Condition~{\sf II} implies a unique partition of the edges into $n+1$ paths; see Figure~\ref{fig:updowntrees}.
Comparing with the windows introduced in Section~\ref{sec:2}, we note that each path corresponds to a double-panel window:
the lowest and highest nodes are the minimum and maximum spanning the window, and all other nodes of the path are maxima of nested windows in the in-panel or the mid-panel of the double-panel window.
The out-panel is indirectly represented by the requirement that its highest node is interior to another path whose lowest node has smaller value than the lowest node of the current path.
If the window is with simple wave, then the triple-panel window is also represented in the down-tree, except that in- and out-panels are exchanged. 

\begin{figure}[htb]
  \centering \vspace{0.1in}
  \resizebox{!}{1.9in}{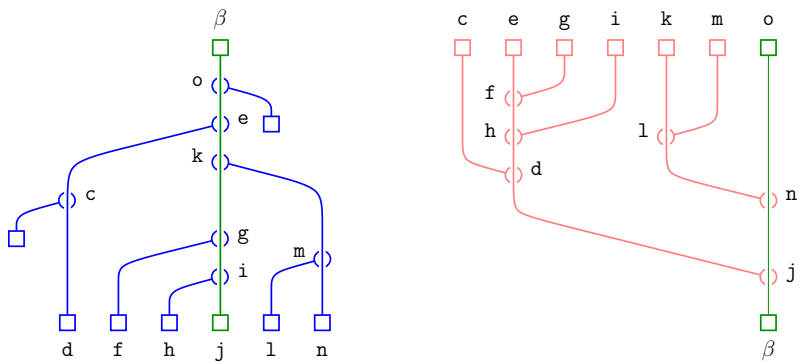}
  \caption{\emph{Left:} the path-decomposition of the binary tree for the map, $f$, displayed in Figure~\ref{fig:maponly}.
  The nodes in its spine (to be defined shortly) are $\T{c}$, $\T{e}$, $\beta$, and $\T{o}$.
  \emph{Right:} upside-down drawing of the path-decomposition of the binary tree for $-f$.
  The nodes in its spine are $\T{d}$, $\T{j}$, and $\beta$.}
  \label{fig:updowntrees}
\end{figure}

\medskip \noindent \textbf{Third step:}
we split every path into two parallel trails.
Let $a = q_0, q_1, \ldots, q_{\ell-1}, q_\ell = b$ be such a path, and recall that $f(q_i) < f(q_{i+1})$ for $0 \leq i \leq \ell-1$.
Assuming $a < b$, the \emph{in-trail} consists of $a$, every $q_i$ whose right child is $q_{i-1}$, and $b$, and the \emph{mid-trail} consists of $a$, every $q_i$ whose left child is $q_{i-1}$, and $b$.
The items in the in-trail are smaller than those in the mid-trail, which motivates the alternative terminology of the \emph{left trail} for the former and the \emph{right trail} for the latter.
If $a > b$, the in-trail is the right trail and the mid-trail is the left trail, so the in-trail corresponds to the in-panel and the mid-trail to the mid-panel in either case.
Both trails \emph{start} at $a$ and \emph{end} at $b$.
All other $q_i$ are \emph{interior} to the trails.
For the special root, $\beta$, and the corresponding lower end of its path, $\alpha$, the order is not defined, so we make an arbitrary choice and call the left trail the \emph{in-trail} and the right trail the \emph{mid-trail}.
We call a trail \emph{empty} if it contains no interior nodes.
In both trails, the items as well as their values are sorted.
To state this formally, assume again that $a < b$:
\smallskip \begin{description}
  \item[{\sf III.1}] writing $a = u_0, u_1, \ldots, u_j = b$ for the nodes along the in-trail,
                we have $u_i > u_{i+1}$ for $0 \leq i \leq j-2$
                and $f(u_i) < f(u_{i+1})$, for $0 \leq i \leq j-1$;
  \item[{\sf III.2}] writing $a = v_0, v_1, \ldots, v_k = b$ for the nodes along the mid-trail,
                we have $v_i < v_{i+1}$ and $f(v_i) < f(v_{i+1})$, for $0 \leq i \leq k - 1$.
\end{description} \smallskip
The respective first inequalities in {\sf III.1} and {\sf III.2} imply that the nodes of the in-trail precede the nodes of the mid-trail.
In contrast, there is no particular order between the values of nodes in different trails.
We draw the trails roughly parallel, like the outline of a banana.
Letting $a, b$ be the lower and upper ends, we call the pair of trails the \emph{banana} spanned by $a, b$.
Except if $b$ is the special root, $b$ is also internal to another trail, so $b$ is where the banana spanned by $a,b$ connects to another banana; see Figure~\ref{fig:conversion}.
We summarize the crucial properties of bananas in a lemma whose proof follows directly from the construction and is thus omitted.
\begin{lemma}[Bananas and Windows]
  \label{lem:bananas_and_windows}
  Let $a, b$ be critical points of a generic map, $f$, with $f(a) < f(b)$.
  The two points span a banana in $\UpTree{f}$ iff $\Window{a}{b}$ is a window of $f$.
  Furthermore, another window, $\Window{p}{q}$, is immediately nested in the in-panel or mid-panel of $\Window{a}{b}$
  iff its maximum, $q$, is an interior node of the in-trail or mid-trail of the banana spanned by $a$ and $b$, respectively.
\end{lemma}
\begin{figure}[b!]
  \centering \vspace{0.0in}
  \resizebox{!}{2.2in}{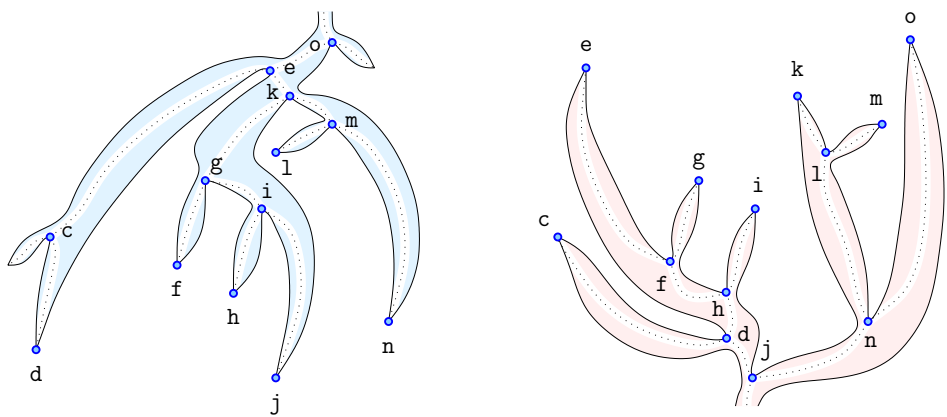}
  \caption{\emph{Left:} the banana tree of the map in Figure~\ref{fig:maponly}, with dotted curves showing the tree before splitting its paths; compare with the \emph{left} drawing in Figure~\ref{fig:updowntrees}.
  \emph{Right:} upside-down drawing of the banana tree for the negated map; compare with the \emph{right} drawing in Figure~\ref{fig:updowntrees}.}
  \label{fig:conversion}
\end{figure}
Each node stores pointers to its neighbors along the trails:
for a node $p$, $\In{p}$ and $\Mid{p}$ point to the first node on the in- and mid-trails beginning at $p$, respectively; a maximum $q$ has additional pointers $\Up{q}$ and $\Dn{q}$, which point to the node above and below $q$ on the same trail.
The banana tree uses additional pointers to connect the ends of its trails:
letting $a$ and $b$ be the lower and upper ends of a trail, we store $\Dth{a} = b$ and $\Low{p} = a$, in which $p$ is any node in the banana other than $b$.
For convenience, this includes $p = a$, which thus stores a pointer to itself.
To obtain the lower end of a banana from its upper end we define $\Birth{b} = \Low{\In{b}} = \Low{\Mid{b}}$.
Observe that each of the two banana trees stores exactly one node per critical item and thus requires $O(n)$ space. The dictionaries and the doubly-linked list require additional $O(m)$ space and hence the total space used by our data structure is linear in the total number of items.

\smallskip
We conclude this section with the assertion that the banana tree of a linear list is unique.
This will be useful in proving some of the algorithms in the subsequent section correct.
\begin{lemma}[Uniqueness of Banana Tree]
  \label{lem:uniqueness_of_banana_tree}
  Given a linear list of distinct values, there is a unique path-decomposed binary tree satisfying Conditions~{\sf I.1}, {\sf I.2}, and {\sf II}, and a resulting unique banana tree satisfying Conditions~{\sf III.1} and {\sf III.2}.
\end{lemma}
\begin{proof}
  The algorithm that constructs the banana tree of a linear list is deterministic and thus computes a unique such tree, which we denote $B$.
  It therefore suffices to prove that no other banana tree satisfies Conditions \textsf{I.1}, \textsf{I.2}. \textsf{II}, \textsf{III.1}, and \textsf{III.2}.
  We have already seen that there is a unique ordered binary tree that satisfies \textsf{I.1} and \textsf{I.2}.
  Similarly, there is a unique decomposition of this tree into paths that satisfies \textsf{II}.

  To get a contradiction, assume $B' \neq B$ is another banana tree that satisfies \textsf{I.1} to \textsf{III.2}.
  Write $B_2$ and $B_2'$ for the path-decomposed ordered binary tree we get by merging the two trails of each banana in $B$ and $B'$, respectively.
  To satisfy Condition \textsf{I.2}, the merging must preserve the ordering by value, and by Conditions \textsf{III.1} and \textsf{III.2} this is indeed possible.
  But the step from $B'$ to $B_2'$ is deterministic, and so is the step from $B$ to $B_2$.
  We have already established that $B_2' = B_2$, which implies $B' = B$, as claimed.
\end{proof}

The reverse of Lemma~\ref{lem:uniqueness_of_banana_tree} does not hold for the trivial reason that the banana trees do not store the non-critical items.
There is also the more subtle reason that different path-decomposed ordered binary trees can map to the same banana tree.
Indeed, the banana trees do not specify the order of values between parallel trails, so merging them can result in different ordered paths.

\subsection{String and Spine.}
\label{sec:3.3}

It is possible to connect the trails of a banana tree into a single curve such that the homological critical points are listed from left to right according to their positions in the interval.
Considering a banana tree as a graph, the minima and the special root have two neighbors, and the maxima have four neighbors each, so the maxima would appear twice, but one appearance is the mirror of the maximum, namely the upper end of the in-trail.
We list the maximum when we reach the upper end of the mid-trail.
Equivalently, we adopt the following rule for a banana spanned by $a, b$:
if $b < a$, we first list $b$, then walk down the interior nodes of the left trail, then list $a$, and finally walk up the interior nodes of the right trail, and if $b > a$, we do the same except that we list $b$ at the end rather than at the beginning.
The sub-bananas are listed recursively when their upper ends are encountered.

\smallskip
As an example consider the banana tree sketched in Figure~\ref{fig:conversion} on the left:  after starting at the special root, we first encounter the left hook, then $c$, $d$, $e$ and so on until $n$, $o$, and finally the right hook before returning to the special root.
We call this the \emph{string} of the banana tree.
It is not difficult to see that this is also the in-order traversal of the tree after the first step of the banana tree construction.

\smallskip
We continue with the definition of an important subset of nodes in a banana tree.
The \emph{left spine} consists of the special root, $\beta$, as well as the first interior node of every left trail with upper end in the left spine.
Symmetrically, the \emph{right spine} consists of $\beta$ as well as the first interior node of every right trail with upper end in the right spine.
The two overlap in $\beta$, and the \emph{spine} is the union of the left and the right spines; see Figure~\ref{fig:updowntrees} for an example.
The nodes in the spine can also be characterized in terms of the windows they span.
\begin{lemma}[Spines and Windows]
  \label{lem:spines_and_windows}
  Let $\UpTree{f}$ be the banana tree of the map $f$.
  Then
  \smallskip \begin{enumerate}[(i)]
    \item the special root $\beta$ and the global minimum span the unique global window of $f$;
    \item a node $b \neq \beta$ of the left spine spans a window with wave that is short on the left;
    \item a node $b \neq \beta$ of the right spine spans a window with wave that is short on the right;
    \item all other internal nodes of $\UpTree{f}$ span windows with simple waves.
  \end{enumerate}
\end{lemma}
\begin{proof}
  It is not necessary to give an argument for \textsf{(i)}.
  Since \textsf{(ii)} and \textsf{(iii)} are symmetric, it suffices to prove \textsf{(ii)}.
  By the first step of the construction of a banana tree in Section~\ref{sec:3.2}, a node $b$ belongs to the left spine iff $f(b) > f(q)$ for all items $q < b$.
  It follows that the window spanned by $b$ is not constrained on the left, so it extends to the beginning of the list.
  In other words, the window is with short wave, and the wave is short on the left.

  \smallskip
  The condition $f(b) > f(q)$ for all $q < b$ is also necessary to have a short wave on the left, so all such windows are spanned by nodes in the left spine.
  Symmetrically, all windows with short wave on the right are spanned by nodes in the right spine.
  Hence, all other internal nodes span windows with simple wave, which proves \textsf{(iv)}.
\end{proof}

According Lemma~\ref{lem:spines_and_windows}, the special root spans the global window, the remaining nodes in the spine span windows with short wave, and the remaining internal nodes span windows with simple wave.
This holds for $\UpTree{f}$ as well as for $\UpTree{-f} = \DnTree{f}$.
Since a window with short wave of $f$ is not a window of $-f$, and vice versa \cite[Theorem~4.2]{BCES21}, this implies that the min-max pairs with the maximum in the spine (other than the special root) are distinct in the two trees.
In contrast, the min-max pairs with the maximum not in the spine are the same in $\UpTree{f}$ and $\DnTree{f}$.

\smallskip
The algorithms for splitting and gluing banana trees in Section~\ref{sec:4.3} need to recognize nodes in the spine.
We thus label these nodes and maintain the labeling when changes occur.

\section{Algorithms}
\label{sec:4}

Call the banana trees of a map and its negative the \emph{up-tree} and \emph{down-tree} of $f$, denoted $\UpTree{f}$ and $\DnTree{f}$.
We describe the construction of both trees, the extraction of the augmented persistence diagram from these trees, and operations that maintain the trees subject to local changes of the data.
We begin with the construction of the up-tree, which takes time linear in the number of items.
Together with the extraction, this gives a linear-time algorithm for the augmented persistence diagram; compare with the algorithm of Glisse \cite{Gli23}, which constructs the persistence diagram in linear time but not the augmentation.
We prove the correctness of the local and topological maintenance operations in \cref{sec:local-maintenance-correctness,sec:topo-correctness}.

\subsection{Construction.}
\label{sec:4.1}

We explain the construction of the up-tree for a generic piecewise linear map defined by a list of $m$ items: $f(i) = c_i$ for $1 \leq i \leq m$.
For convenience, we add items $0$ and $m+1$ at the two ends, with $f(0) = \infty$ and $f(m+1) = \infty - 1$.
The algorithm processes the map from left to right and uses a stack to store a subset of the minima and maxima so far encountered.
Specifically, while processing $j$, the stack stores pairs $(a_0, b_0), (a_1, b_1), \ldots, (a_k, b_k)$ such that
\smallskip \begin{itemize}
   \item $0 = a_0 = b_0 < a_1 < b_1 < a_2 < \ldots < a_k < b_k < j$,
   \item $f(i) \leq f(b_\ell)$ for all $1 \leq \ell \leq k$ and $b_{\ell-1} < i < j$,
   \item $f(i) \geq f(a_\ell)$ for all $1 \leq \ell \leq k$ and $b_{\ell-1} < i \leq b_\ell$.
\end{itemize} \smallskip
The first two properties imply that the points $(b_\ell, f(b_\ell))$ form a descending staircase in the plane, and all points $(i, f(i))$ with $i < j$ that are not steps of the staircase lie below the staircase.
The third property says that item $a_\ell$ minimizes the value among all items $i$ vertically below the step of $b_{\ell}$.

\smallskip
The only relevant items are the critical points, so we eliminate non-critical items in a preliminary scan and connect the remaining items with prv- and nxt-pointers.
Here we consider the artificially added items as maxima, so $0$ is the first item in this list, followed by $\Nxt{0}$, $\Nxt{\Nxt{0}}$, etc., until we reach item $m+1$.
For later use, each remaining item in the list is provided with the appropriate subset of initially \texttt{null} in-, mid-, up-, dn-, low-, and dth-pointers.
During the main scan the algorithm maintains unfinished bananas, each corresponding to a pair on the stack, as well as a value $A$, which, after the current item has been processed, is the item with the minimum value to the right of the top-most item on the stack. 
Whenever we pass a  minimum, $j$, we set $A=j$, as the immediately preceding item must have been a maximum.
Whenever we pass a maximum, the stack is accessed through the standard functions \textsc{Push} and \textsc{Pop}, as well as through function \textsc{Top}, which returns the item $b$ at the top of the stack, but without removing the pair $(a,b)$ from the stack.

\smallskip
Suppose $(a, b)$ is the pair at the top of the stack, and we pop it off because $f(b) < f(j)$.
If $f(A) < f(a)$, we now know that the unfinished banana spanned by $a,b$ is correct, so we finalize it in $\textsc{FixBanana}$.
Otherwise, $A, b$ span a banana, which must belong to a right trail as $b < A$. 
We, thus, detach $b$ on the left, attach it on the right, and finalize the banana in $\textsc{FixBanana}$.
Then we set $A = a$ and push $(A,j)$ on the stack after creating its temporary banana and attaching it on the left:
\begin{tabbing}
  m\=m\=m\=m\=m\=m\=m\=m\=m\=\kill
  \> \> $\Dn{0} = 1$; $\textsc{Push} (a_0=0, b_0=0)$; $j = 0$; \\*
  \> \> \texttt{repeat} \= $j = \Nxt{j}$; \\*
  \> \> \> \texttt{if} $j$ is minimum \texttt{then} $A = j$ \texttt{endif}; \\*
  \> \> \> \texttt{if} \= $j$ is maximum \texttt{then} \\*
  \> \> \> \> \texttt{while} \= $f(j) > f(\textsc{Top})$ \texttt{do}
  $(a, b) = \textsc{Pop}$; \\*
  \> \> \> \> \> \texttt{if} \= $f(A) < f(a)$ \texttt{then} $\textsc{FixBanana} (a, b)$ \\*
  \> \> \> \> \> \> \texttt{else} \= attach $b$ below $j$ on the right; \\*
  \> \> \> \> \> \> \> $\textsc{FixBanana} (A, b)$;
                       $A = a$ \\*
  \> \> \> \> \> \texttt{endif} \\*
  \> \> \> \> \texttt{endwhile}; \\*
  \> \> \> \> attach $j$ below $b$ on the left; $\textsc{Push} (A, j)$; \\*
  \> \> \> \> \texttt{if} $j = m+1$ \= \texttt{then} $\textsc{FixBanana} (A, j)$ \texttt{endif} \\*
  \> \> \> \texttt{endif} \\*
  \> \> \texttt{until} $j = m+1$.
\end{tabbing}
To ``attach $j$ below $b$ on the left'', where  $j$ is the currently processed item and $b$ is a past item,
we create an unfinished banana with upper end $j$ and temporary in- and mid-pointers, and set the relevant pointers of $j$, as illustrated in Figure~\ref{fig:attach} on the left:
\begin{tabbing}
  m\=m\=m\=m\=m\=m\=m\=m\=m\=\kill
  \> \> $\Up{j} = b$; $\In{j} = \Dn{b}$;
  $\Mid{j} = \Prv{j}$; $\Dn{j} = \Nxt{j}$; \\*
  \> \> $\Dn{b} = \Up{\In{j}} = \Up{\Mid{j}} = j$.
\end{tabbing}
Whenever the pair containing $j$ is later popped off the stack, and a comparison of the values shows that $j$ belongs to a left trail, then these pointers remain unchanged in the final banana.

\begin{figure}[htb]
  \centering \vspace{-0.0in}
  \resizebox{!}{1.4in}{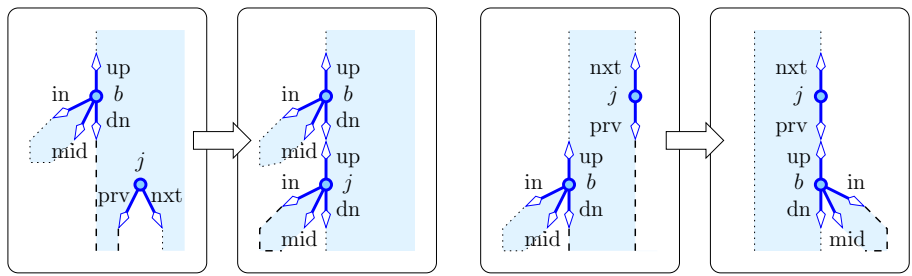}
  \caption{\emph{Left:} item $j$ it temporarily attached on the left, which creates a banana whose upper end, $j$, is interior to a left trail.
  \emph{Right:} the temporary attachment of item $b$ on the left is resolved, and $b$ is attached on the right, which creates a banana whose upper end, $b$, is interior to a right trail.}
  \label{fig:attach}
\end{figure}
\noindent
If, however, this is not the case, we execute ``attach $b$ below $j$ on the right'', which suitably updates these values to reflect the fact that $j$ belongs to a right trail.
Specifically, we detach $b$ on the left and attach it on the right, as illustrated in Figure~\ref{fig:attach} on the right:
\begin{align*}
  &\Dn{\Up{b}} = \In{b}; \Up{\In{b}} = \Up{b}; \\*
  &\Up{b} = j; \In{b} = \Prv{j}; \textsf{aux} = \Dn{b}; \Dn{b} = \Mid{b}; \Mid{b} = \textsf{aux}; \\*
  &\Prv{j} = \Up{\In{b}} = b.
\end{align*}
To fix a banana, we set the low-pointers of its interior nodes and the minimum, as well as the in-, mid-, and dth-pointers of the minimum:
\begin{tabbing}
  m\=m\=m\=m\=m\=m\=m\=m\=m\=\kill
  \> \texttt{function} $\textsc{FixBanana} (a, b)$: \\*
  \> \> $q = b$; $p = \In{b}$; \= \texttt{while} $p \neq a$ \texttt{do} $\Low{p} = a$; $q=p$; $p = \Dn{p}$ \texttt{endwhile}; $\In{a} = q$; \\*
  \> \> $q = b$; $p = \Mid{b}$; \texttt{while} $p \neq a$ \texttt{do} $\Low{p} = a$; $q=p$; $p = \Dn{p}$ \texttt{endwhile}; $\Mid{a} = q$; \\*
  \> \> $\Low{a} = a$; $\Dth{a} = b$.
\end{tabbing}

\noindent
To argue the correctness of the algorithm, we formulate three invariants maintained by the algorithm.
Let $j$ denote the current item, and distinguish between \emph{past items}, $b < j$, and \emph{future items}, $b > j$.
\smallskip \begin{enumerate}[(i)]
  \item The prv-pointers of the future items are unchanged, except possibly the first one, which points to $j$.
  In contrast, all nxt-pointers remain unchanged throughout.
  \item If a past item, $b$, is on the stack, then the banana rooted at $b$ is temporary and unfinished.
  The latter means that the in- and mid-pointers of $b$ and the up- and dn-pointers of $b$ and all interior nodes are in place, but not necessarily the remaining pointers that define the banana.
  In contrast, all its sub-bananas are complete.
  \item Right before taking $b$ off the stack, it satisfies (ii) except that the banana rooted at $b$ is no longer temporary.
  After taking $b$ off the stack, its banana is complete, which means that all pointers of its nodes are in place.
\end{enumerate} \smallskip
It is easy to see that \textsf{(i)} is maintained, as there is only one place where a prv-pointer is altered, and there is no place where a nxt-pointer is altered.
Invariant \textsf{(ii)} holds because all interior nodes of the banana rooted at $b$ have been taken off the stack in the past, and their bananas are complete by Invariant \textsf{(iii)}.
In particular, this implies that each unfinished banana spanned by $a, b$ has a path from $\In{b}$ along dn-pointers to $a$ defining its in-trail, and similarly from $\Mid{b}$ along dn-pointers to $a$ defining its mid-trail.
Assuming \textsf{(ii)}, function \textsc{FixBanana} adds the necessary pointers to complete the banana, which implies \textsf{(iii)}.

\medskip \noindent \textbf{Extraction.}
We next discuss how to compute the augmented persistence diagram from the up- and the down-tree of $f$.
Each banana in $\UpTree{f}$ corresponds to a point in $\Ord{}{f}$, with the exception of the banana of the special root, which corresponds to a point in $\Ess{}{f}$.
We find these points and the nesting relation by recursively walking along the trails from bottom to top.
We are handed the upper end of each pair of trails, so first we jump to the lower end.
The recursive function that enumerates all points and arrows is called with two parameters: the upper end of two parallel trails, and the point in the persistence diagram that corresponds to the pair of trails that contain that upper end as an interior node.
Initially, the upper end is the special root, and the point is empty:
\begin{tabbing}
  m\=m\=m\=m\=m\=m\=m\=m\=m\=\kill
  \> \> \texttt{function} $\textsc{Walk} (b, pnt)$: \\*
  \> \> \> $a = \Birth{b}$;
           output $\Point{a}{b} = (f(a), f(b))$ and $\Arrow{a}{b} = (\Point{a}{b}, {pnt})$; \\*
  \> \> \> $x = \In{a}$;
  \texttt{while} $x \neq b$ \texttt{do} $\textsc{Walk} (x, \Point{a}{b})$; $x = \Up{x}$ \texttt{endwhile}; \\*
  \> \> \> $x = \Mid{a}$;
  \texttt{while} $x \neq b$ \texttt{do} $\textsc{Walk} (x, \Point{a}{b})$; $x = \Up{x}$ \texttt{endwhile}.
\end{tabbing}
To complete the augmented persistence diagram, we also construct $\DnTree{f}$ and apply the recursive function to its parallel trails, with the only difference that the banana of the special root does not correspond to any point in $\Dgm{}{f}$.

\medskip \noindent \textbf{Summary.}
After removing all non-critical items, which takes $O(m)$ time, the construction of the two banana trees as well as the extraction of the augmented persistence diagram takes only $O(n)$ time.
Indeed, the main scan of the construction algorithm completes each banana only once, and altogether touches each item only a constant number of times.
We summarize our findings for later reference.
\begin{theorem}[Time to Construct]
  \label{thm:time_to_construct}
  Let $f \colon [0,m+1] \to \Rspace$ be a generic piecewise linear map with $n$ maxima.
  After removing all non-critical items in $O(m)$ time, $\UpTree{f}$ and $\DnTree{f}$ can be constructed in $O(n)$ time, and the augmented persistence diagram of $f$ can be computed from these trees in $O(n)$ time.
\end{theorem}

\subsection{Local Maintenance.}
\label{sec:4.2}

Recall that the banana trees store all critical items but none of the non-critical items.
As a temporary exception, %
we \emph{delete} items by first adjusting their values until they become non-critical.
To formalize the operation that \emph{adjusts} the value of an item, $j$, from $c_j$ to $d_j$, we write $f, g \colon [0,m+1] \to \Rspace$ for the maps before and after adjustment, so $g(i) = f(i)$ for $0 \leq i \neq j \leq m+1$, and $f(j) = c_j$, $g(j) = d_j$.
To avoid complications near the endpoints, assume $3 \leq j \leq m-2$. We give details on how to treat endpoints in \cref{sec:scenario-correctness}.
We modify the trees while following the \emph{straight-line homotopy} from $f$ to $g$, which is the $1$-parameter family of maps $h_\lambda \colon [0,m+1] \to \Rspace$
defined by $h_\lambda (x) = (1-\lambda) f(x) + \lambda g(x)$ for $0 \leq \lambda \leq 1$.
Clearly, $h_0 = f$ and $h_1 = g$.

\smallskip
The homotopy reduces to a sequence of \emph{interchanges}---of two maxima or two minima---followed or preceded by a \emph{cancellation} or an \emph{anti-cancellation}, which reflect the disappearance or appearance of a point into or from the diagonal of the persistence diagram, or a \emph{slide}, which occurs when a critical item becomes non-critical due to a non-critical neighbor becoming critical.
We will discuss these operations in detail below. 
Among the adjustments, we consider two scenarios:
\smallskip \begin{description}
  \item[{\sc A:}] item $j$ is non-critical in $f$ and increases its value;
  \item[{\sc B:}] item $j$ decreases its value until it becomes non-critical in $g$.
\end{description} \smallskip
The scenario in which $j$ is non-critical and decreases its value is symmetric to A, and either of the two is applied after we insert a new item.
The scenario in which $j$ increases its value until it becomes non-critical is symmetric to B, and either of the two is needed before we delete an item.
Other adjustments are subsequences or compositions of Scenarios~A and B or their symmetric versions.

\medskip \noindent \textbf{Scenario~A.}
We increase the value of a non-critical item, $j$, and we assume that it becomes critical, else up-tree and down-tree would not be affected.
We consider the case in which $f(j-1) < f(j) < f(j+1)$ and $f(j-1) < g(j) > f(j+1)$, i.e., the item $j$ becomes a maximum in $g$.
If $j+1$ is non-critical in $f$, it becomes a minimum in $g$ and the number of critical items increases by two.
If $j+1$ is critical (a maximum), then it becomes non-critical in $g$ and item $j$ replaces it as maximum; the number of critical items does not change.
In the former case we perform an anti-cancellation to introduce the banana spanned by $j+1$ and $j$;
in the latter case we perform a slide.
Afterwards we fix the position of $j$ in the up-tree and down-tree by performing a sequence of interchanges:
\begin{tabbing}
  m\=m\=m\=m\=m\=m\=m\=m\=m\=\kill
  \> \> \texttt{if} $g(j) > f(j+1)$ \texttt{then} \\*
  \> \> \> \texttt{if} $f(j+1) < f(j+2)$ \= \texttt{then} anti-cancel $j$ and $j+1$ in $\UpTree{f}$ and $\DnTree{f}$ \\*
  \> \> \>                                    \> \texttt{else} slide $j+1$ to $j$ in $\UpTree{f}$ and $\DnTree{f}$ \\*
  \> \> \> \texttt{endif}; \\*
  \> \> \> set $q = \Up{j}$ in $\UpTree{f}$; \\*
  \> \> \> \texttt{while} $f(q) < g(j)$ \texttt{do} \=
    interchange $j$ and $q$ in $\UpTree{f}$ and $\DnTree{f}$; \\*
  \> \> \> \> $q = \Up{j}$ in $\UpTree{f}$ \\*
  \> \> \> \texttt{endwhile} \\
  \> \> \texttt{endif}.
\end{tabbing}
Note that we iterate with $q = \Up{j}$ and not with $q = \Up{q}$.
This is not a mistake because the interchange in $\UpTree{f}$ is of two maxima, $j$ and $q$, which swaps the two nodes (see below).
In contrast, the interchange in $\DnTree{f}$ is of two minima, albeit they are the same two items, $j$ and $q$.
Performing these interchanges simultaneously, we save the effort of independently finding the next relevant interchange of minima, which would be costly.
The correctness of this strategy is guaranteed by Lemma~\ref{lem:coupling_of_interchanges}, which we will state and prove after formalizing the notion of an interchange.

\medskip \noindent \textbf{Scenario~B.}
The symmetric version of Scenario~A---in which $j$ decreases its value---is implemented accordingly, by switching the roles of the up-tree and the down-tree.
We use the simultaneous interchange of maxima and minima also in the implementation of this inverse of Scenario~A.
There are again symmetric cases, and we consider the case in which $f(j-1) < f(j+1)$ and $f(j-1) < f(j) > f(j+1)$.
Furthermore, we assume that item $j$ is interior to a left trail.
The operation begins with a sequence of interchanges of maxima in the up-tree and of minima in the down-tree, followed by a cancellation or by a slide:
\begin{tabbing}
  m\=m\=m\=m\=m\=m\=m\=m\=m\=\kill
  \> \> \texttt{loop} $q = \arg\max\{f(\Dn{j}), f(\In{j}), f(\Mid{j})\}$ in $\UpTree{f}$; \\*
  \> \> \> \texttt{if} $f(q) > f(j+1)$ \= \texttt{then}
  interchange $q$ and $j$ in $\UpTree{f}$ and $\DnTree{f}$ \\*
  \> \> \>                             \> \texttt{else} \texttt{exit}
  \texttt{endif} \\*
  \> \> \texttt{forever}; \\
  \> \> \texttt{if} $f(j+1) < f(j+2)$ \= \texttt{then}
         cancel $j$ with $j+1$ in $\UpTree{f}$ and $\DnTree{f}$ \\*
  \> \>                               \> \texttt{else}
         slide $j$ to $j+1$ in $\UpTree{f}$ and $\DnTree{f}$ \\*
  \> \> \texttt{endif}.
\end{tabbing}
Note that $q$ is either a maximum or its value is less than $f(j+1)$ as $f(j-1) < f(j+1)$, so we will not attempt to interchange a maximum, $j$, with a minimum, $q$, which would be impossible indeed.
We continue with the details for the interchanges and (anti-)cancellations.

\medskip \noindent \textbf{Interchange of maxima.}
An interchange between two maxima with $f(j) < f(q)$
has structural consequences only if $q = \Up{j}$, as this is the only case in which a uniqueness condition, namely {\sf III.1} or {\sf III.2}, is violated.
We will see that the algorithm does not ever get into the situation of attempting any other interchange.
Assume that $q$ is on a left trail, which is the situation depicted in Figure~\ref{fig:maxinterchange}.
Let $i$ and $p$ be such that $j = \Dth{i}$ and $q = \Dth{p}$.
If $j = \Dn{q}$, we distinguish two cases depending on the order of $f(i)$ and $f(p)$.
\smallskip \begin{description}
  \item[{\sc Case 1:}] $f(i) < f(p)$.
    Remove $q$ from its trail and add it right below $j$ in $j$'s in-trail, as illustrated in Figure~\ref{fig:maxinterchange} in the top left.
    Adjust the pointers of the involved nodes accordingly.
  \item[{\sc Case 2:}] $f(i) > f(p)$.
    Exchange $j$ and $q$ as upper ends of their respective bananas, remove $q$ from its trail and add it right below $j$ in $j$'s mid-trail, as illustrated in Figure~\ref{fig:maxinterchange} in the top right.
    Adjust pointers, and in particular set $\Dth{i} = q$ and $\Dth{p} = j$.
    The in-trail of $i$ becomes its mid-trail and vice versa.
\end{description} \smallskip
In both cases, $j$ joins the left spine of the up-tree iff $q$ is a node of the left spine already before the operation.
The cases where $j = \Mid{q}$ or $j = \In{q}$ are similar to the reverse of the cases with $j = \Dn{q}$, and are illustrated in the bottom row of \cref{fig:maxinterchange}.
There are also the inverse operations (reading Figure~\ref{fig:maxinterchange} from right to left), which apply when $j = \Up{q}$,
i.e., the case $f(j) > f(q)$, and the symmetric cases where $q$ is on a right trail.
\begin{figure}[htb]
  \centering \vspace{0.0in}
  \resizebox{!}{2.9in}{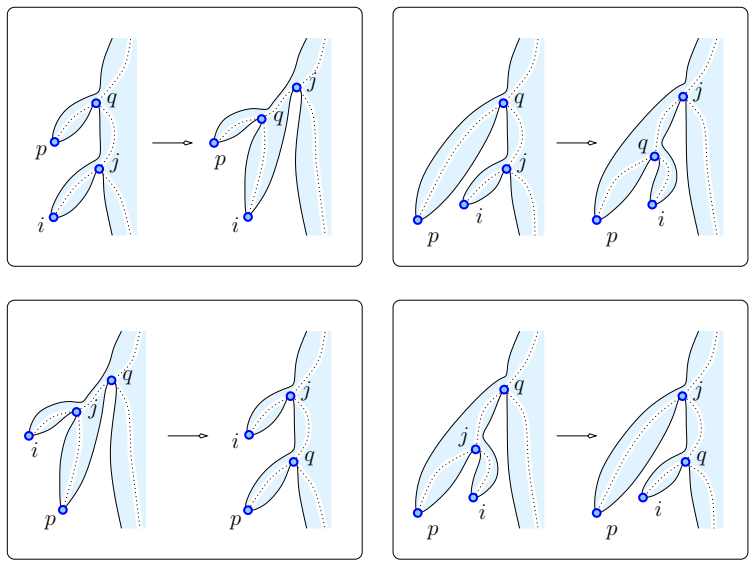}
  \caption{The interchange of two maxima, $j$ and $q$.
  In the \emph{dotted} underlying tree, the operation corresponds to a rotation.
  Before the interchange, the pairs are $i, j$ and $p, q$.
  \emph{Top left:} $f(i) < f(p)$ which preserves the pairs.
  \emph{Top right:} $f(i) > f(p)$ and the pairs change to $i, q$ and $p, j$.
  \emph{Bottom left:} $j = \In{q}$ which preserves the pairs.
  \emph{Bottom right:} $j = \Mid{q}$ and the pairs change to $i, q$ and $p, j$.}
  \label{fig:maxinterchange}
\end{figure}

\medskip \noindent \textbf{Interchange of minima.}
An interchange of maxima in $\UpTree{f}$ is always done in parallel with an interchange of minima in $\DnTree{f}$.
It can, however, happen that the interchange of $j$ and $q$ has a structural effect on $\UpTree{f}$ but not on $\DnTree{f}$.
An example is the interchange of nodes $i$ and $g$ in Figure~\ref{fig:updowntrees} and \ref{fig:conversion}.
They are internal nodes in the up-tree, so the effect of the interchange is as depicted in Figure~\ref{fig:maxinterchange} on the left.
The two nodes are leaves in the down-tree whose bananas do not meet, so the interchange of minima has no effect.

\smallskip
In general, the interchange of two minima with $f(j) < f(q)$ has structural consequences only if $p = \Dth{q}$ is interior to the banana spanned by $j$ and $i = \Dth{j}$, as this is the only case in which $g(j) > g(q)$ leads to a violation of the uniqueness conditions, namely of {\sf II}.
To describe how the trails are updated, we assume that $p$ is interior to the left trail.
We split this trail into the upper part above $p$ and the lower part below $p$ (with neither part including $p$).
Similarly, we split the parallel right trail connecting $i$ to $j$ into upper and lower, with values larger and smaller than $f(p)$, respectively.
Then we concatenate the upper part of the left trail with the left trail connecting $p$ to $q$ (but without the upper end, which is $p$), and we concatenate the upper part of the right trail with the right trail connecting $p$ to $q$ (this time including $p$).
Finally, we join the two lower parts to form parallel trails connecting $p$ to $j$.

\smallskip
Observe that splitting the left trail is easy because it contains $p$ as an interior node.
We split the right trail by traversing it one node at a time from the upper end until we reach the first node with value less than $f(p)$. Note that all the travsersed node on the right trail will have a change in their low-pointer and we will use this fact
when we analyze the running-time of the operation at the end of this subsection.

\medskip \noindent \textbf{Simultaneity of interchanges.}
Next, we prove the correctness of coupling interchanges as described in Scenarios~A and B.
The moment two maxima of $f$ interchange is of course also the moment at which they interchange as minima of $-f$.
However, many interchanges are irrelevant, in the sense that they cause no structural changes to the banana trees.
Many of these irrelevant interchanges go unnoticed by our algorithm (and fortunately so), but it is important that no relevant interchange is overlooked.
We claim that every relevant interchange of minima corresponds to a relevant interchange of maxima.

\smallskip
To formalize this claim, let $F \colon [0,m+1] \to \Rspace$ be a piecewise linear map determined by its values at the integers, and assume that these values are distinct, with the exception of $F(j) = F(q)$ at maxima $j \neq q$ of $F$.
Let $f, g \colon [0,m+1] \to \Rspace$ be piecewise linear maps defined by the same values at the integers, except that $f(j) = F(j) - \ee$ and $g(j) = F(j) + \ee$ for a sufficiently small $\ee > 0$.
The straight-line homotopy from $f$ to $g$ is an interchange of maxima, namely of $j$ and $q$, and that from $-f$ to $-g$ is an interchange of minima, again of $j$ and $q$.
We call the former \emph{relevant} if $\UpTree{f}$ and $\UpTree{g}$ differ by a rotation, namely of $j$ and $q$, and we call the latter \emph{relevant} if $\DnTree{f}$ and $\DnTree{g}$ differ by a change in the pairing.
\begin{lemma}[Coupling of Interchanges]
  \label{lem:coupling_of_interchanges}
  Let $f, g \colon [0,m+1] \to \Rspace$ as introduced above.
  If the straight-line homotopy from $-f$ to $-g$, which is an interchange of minima, is relevant, then so is the straight-line homotopy from $f$ to $g$, which is an interchange of maxima.
\end{lemma}
\begin{proof}
  We prove the contrapositive: that irrelevant interchanges of maxima imply irrelevant interchanges of minima.
  The interchange of the maxima $j \neq q$ of $f$ is irrelevant if there is a banana in $\UpTree{f}$ so that $j$ and $q$ belong to opposite trails, or to sub-bananas rooted on opposite trails of this banana.
  Let this banana be spanned by $a, b$.
  Assuming $j < q$, this implies $j < a < q$ and $b$ is either to the left or the right of the three items.
  Letting $i$ and $p$ be the lower ends of the bananas spanned by $j$ and $q$, respectively, we observe that $i < a < p$, $f(a) < f(i)$,  and $f(a) < f(p)$.

  \smallskip
  For the negated function, $j, q$ are minima and $i,a,p$ are maxima satisfying $i,j < a < p,q$, $-f(a) > -f(i)$,  and $-f(a) > -f(p)$.
  If $\Window{i}{j}$ and $\Window{p}{q}$ are both with simple waves, then $\Window{j}{i}$ and $\Window{q}{p}$ are triple-panel windows of $-f$ that are separated by $a$.
  It follows that the interchange of minima is irrelevant.
  Similarly, if one or both of these windows are with short wave, then the separation by $a$ implies that the pairing stays constant, so again the interchange is irrelevant.
\end{proof}

There is a special case to consider when the separating banana is spanned by the special root, and $j, q$ are the maxima with the two largest values.  For $f$, we have $\beta = q$ and for $g$ we have $\beta = j$.
In words, the interchange of the maxima $j$ and $q$ does not change the pairing but it causes a replacement of the special root.

\medskip \noindent \textbf{Cancellations.}
The last step in Scenario~B is the cancellation of items $j$ and $j+1$, which is implemented by removing the two nodes from the up-tree and down-tree.
We argue that the implementation is really this easy.

\smallskip
Recall that the value of $j$ during the homotopy from $f$ to $g$ is $h_\lambda (j) = (1-\lambda) f(j) + \lambda g(j)$.
By assumption, $f(j-1) < f(j+1)$, so the cancellation happens at $f(j-1) < g(j) < f(j+1)$.
Since $f(j) > f(j+1) > g(j)$ and $f(j+2) > f(j+1)$, there exists $0 \leq \mu \leq 1$ such that $f(j+2) > h_\mu (j) > f(j+1)$.
At this stage of the homotopy, $j+1$ is a child of $j$ and $\Dth{j+1} = j$ in $\UpTree{f}$.
In other words, $j+1$ and $j$ are the upper and lower ends of a pair of empty trails.
We can therefore simply remove this pair of trails, while forming a direct link between $\Dn{j}$ and $\Up{j}$.
The situation is symmetric in $\DnTree{f}$.

\medskip \noindent \textbf{Anti-cancellations.}
The first step in Scenario~A is the anti-cancellation of items $j$ and $j+1$.
We describe how to perform the anti-cancellation in the up-tree; the anti-cancellation in the down-tree is symmetric.
By assumption, $f(j-1) < g(j) > f(j+1) < f(j+2)$ and $g(j)$ is such that there is no item with value between $f(j+1)$ and $g(j)$.
We first identify the maximum $b$ closest to $j$ such that the new minimum $j+1$ lies between $j$ and $b$, i.e., $j < j+1 < b$.
Note that there can be no other critical point between $j+1$ and $b$.

\smallskip
To insert $j$, we walk along the path from $b$ to the closest minimum $a$ with $a < j < j+1 < b$,
until we find a node $q$ with smaller value than $j$.
The node $j$ is then inserted as a parent of $q$:
\begin{tabbing}
    m\=m\=m\=m\=m\=m\=m\=m\=m\=m\=m\=\kill
    \> \> \texttt{if} $\Birth{b} < j < b$ \texttt{then} $q = \Mid{b}$ \texttt{else} $q = \Dn{b}$ \texttt{endif}; \\*
    \> \> \texttt{while} $f(q) > g(j)$ \texttt{do} $q = \In{q}$ \texttt{endwhile}; \\*
    \> \> \texttt{if} $q$ is a leaf \= \texttt{then} 
      \= \texttt{if} $q = \Birth{b}$
                \= \texttt{then} insert $j$ as $\Mid{q}$ \\*
    \> \> \> \> \> \texttt{else if} \= $q$ is the global minimum \texttt{and} $q < j$
                        \texttt{then} \\*
    \> \> \> \> \> \>       insert $j$ as $\Mid{q}$ \\*
    \> \> \> \> \> \texttt{else} insert $j$ as $\In{q}$ \\*
    \> \> \> \> \texttt{endif} \\*
    \> \> \> \texttt{else} insert $j$ as $\Up{q}$ \\*
    \> \> \texttt{endif.}
\end{tabbing}
After inserting $j$, we construct a banana spanned by $j+1$ and $j$, which includes setting $\In{j+1} = \Mid{j+1} = \Dth{j+1} = j$ and $\In{j} = \Mid{j} = j+1$.

\medskip \noindent \textbf{Slides.}
Consider the case in which the value of a maximum, $j$, is decreased, and assume that $f(j-1) < f(j+1) > f(j+2)$. 
Note that $j+2$ is non-critical.
If $f(j-1) < g(j) < f(j+1)$, then $j+1$ becomes a maximum and $j$ becomes non-critical.
The number of critical items remains the same in $f$ and $g$, but criticality is transferred from one item in $f$ to a neighboring item in $g$.
As mentioned earlier, this is what we call a \emph{slide}.

\smallskip
The same situation occurs when
(1) the value of a minimum increases above the value of a non-critical neighbor,
(2) the value of a non-critical item increases above the value of a neighboring maximum, or
(3) the value of a non-critical item decreases below the value of a neighboring minimum.
A slide has no impact on the structure of the banana-tree and only requires to update the association between items and nodes.

\medskip \noindent \textbf{Summary.}
We summarize the findings in this subsection by stating the running-time for adjusting the value of an item in terms of the number of critical points and the number of changes caused to the augmented persistence diagram.
To this end, we define the \emph{accumulated change} along a straight-line homotopy as the total number of changes to the points and arrows in the augmented persistence diagram that occur during the homotopy.

\begin{theorem}[Time to Adjust]
  \label{thm:time_to_adjust}
  Let $f \colon [0,m+1] \to \Rspace$ be a generic piecewise linear map with $n$ maxima.
  The time to adjust the value of an item is $O(\log n +k)$, in which $k$ is the accumulated change during the adjustment, plus $O(k')$ if the adjustment includes an anti-cancellation, in which $k'$ is the difference in the transitive closures of the nesting relation before and after the anti-cancellation.
\end{theorem}
\begin{proof}
  We prove the claimed bound on the running-time by charging most steps to the change in augmented persistence diagrams they cause.

  \smallskip
  An \emph{interchange of maxima} reduces to a rotation plus possibly a swap in the pairing.
  The rotation takes $O(1)$ time, which we charge to the changing arrow, and the swap takes $O(1)$ time, which we charge to the two points in the diagram that exchange coordinates.
  An \emph{interchange of minima} reduces to a swap in the pairing followed by resetting the low-pointers along the path connecting the two maxima involved in the swap.
  As before, the swap is charged to the two points that exchange coordinates, and each resetting of a low-pointer is charged to the corresponding change in the arrow.
  It is also possible that the interchange of minima has no effect on the binary tree, namely when the corresponding paths are disjoint.
  This is detected in $O(1)$ time, and this time is charged to the changes caused by the simultaneous interchange of maxima in the other banana tree.
  A \emph{cancellation} takes $O(1)$ time, and there are at most two cancellations per value adjustment of an item.
  In contrast, an \emph{anti-cancellation} takes $O(\log n)$ time to find the maximum, $b$, closest to the new pair of critical items.
  The while-loop in the anti-cancellation takes as many iterations as there are nodes on the path from $b$ to the newly inserted node.
  For each of those nodes an arrow appears in the transitive closure of the nesting relation.
  Thus, the time to insert the new banana is $O(k')$.
  The insertion itself takes constant time.

  \smallskip
  In addition, we take $O(\log n)$ time to update the dictionaries whenever an item changes from critical to non-critical, or the other way round, and there are only $O(1)$ such changes per value adjustment of an item.
  All this adds up to $O(\log n + k)$ time, if the adjustment does not requires any anti-cancellation, and to $O(\log n + k + k')$ time, if it does.
\end{proof}

\subsection{Topological Maintenance.}
\label{sec:4.3}

Given a collection of linear lists, we next study the maintenance of the augmented persistence diagram subject to cutting and concatenating the lists.
Recall that a list $c_1, c_2, \ldots, c_m$ induces a piecewise linear map, $f \colon [0,m+1] \to \Rspace$, with $f(i) = c_i$ for $1 \leq i \leq m$.
The values at $i=0, m+1$ are added to create the computationally convenient hooks introduced in Section~\ref{sec:3}.
To \emph{cut} $f$, we split the list into $c_1, c_2, \ldots, c_{\ell}$ and $c_{\ell+1}, c_{\ell+2}, \ldots, c_m$ and let $g \colon [0,\ell+1] \to \Rspace$ and $h \colon [\ell, m+1] \to \Rspace$ be the corresponding piecewise linear maps.
We need at least two items to construct the hooks, so we require $2 \leq \ell \leq m-1$.
We describe the operation for the up-tree, and consider the down-tree only to the extent it provides information to update the up-tree.
For ease of reference, we write $x$ for the midpoint between $\ell$ and $\ell+1$ and set $f(x) = \frac{1}{2} (f(\ell) + f(\ell+1))$.

\medskip \noindent \textbf{Splitting a banana tree.}
We introduce terminology before describing the splitting algorithm in three steps.
A banana suffers an \emph{injury}, \emph{fatality}, \emph{scare} if $x$ cuts through the in-, mid-, out-panel of the window, respectively; see Figure~\ref{fig:injuries}.
A triple-panel window with simple wave is shared by $f$ and $-f$, except that in- and out-panels are exchanged.
Hence, a scare in $\UpTree{f}$ is an injury in $\DnTree{f}$, so we exploit the down-tree to find the scares in the up-tree.
\begin{figure}[htb]
  \centering \vspace{0.0in}
  \resizebox{!}{0.9in}{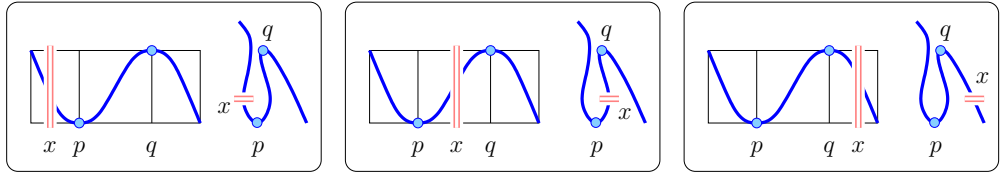}
  \caption{\emph{From left to right:} an injury, fatality, and scare.
  Correspondingly, $x$ cuts the banana in its in-trail, in its mid-trail, or after the maximum but still above the value of the minimum.}
  \label{fig:injuries}
\end{figure}

\noindent \textsc{Step 1:} \textbf{Find smallest banana.}
To enumerate the bananas that suffer injuries or fatalities, we first find the smallest such banana.
To this end, we search the dictionaries to locate the smallest interval between two critical points that satisfy $a < x < b$.
Note however that neither $b$ is necessarily the upper end nor $a$ is necessarily the lower end of this banana.
Assuming $b$ is a maximum, we use it to find the smallest banana in the up-tree such that $x$ lies between its lower and upper ends.
The search uses a pair of nodes, $q, r$, maintaining that the upper end of the desired smallest banana
has an ancestor that  is an interior node of the right trail with upper end $q$, and $r$ with $f(r) > f(x)$ is a node on this right trail.
The iteration advances $q$ and $r$ toward $x$ and halts when these conditions fail.  
To repeatedly go down the right trail of $r$, which itself is interior to a right trail, we use the $\In{r}$ pointer.
The iteration is slightly different in the first case, when $b$ is interior to a right trail, and in the second case, when $b$ is interior to a left trail:
\begin{tabbing}
  m\=m\=m\=m\=m\=m\=m\=m\=m\=\kill
  \> \texttt{function} $\textsc{SmallestBanana} (x)$: \\*
  \> \> find $a < x < b$ and assume that $b$ is a maximum; \\* 
  \> \> \texttt{if} $\Dn{b} < b$ \= \texttt{then} 
    $q = \Dth{\Low{b}}$; $r = \Dn{b}$
  \texttt{else} $q = b$; $r = \Mid{b}$
  \texttt{endif}; \\*
  \> \> \texttt{while} $r \neq a$ \texttt{and} $f(x) < f(r)$ \texttt{do}
        $q = r$; $r = \In{r}$
        \texttt{endwhile}; \\*
  \> \> \texttt{return} $(\Birth{q},q)$.
\end{tabbing}
The time is $O(\log n + k)$, in which $k$ is the number of inspected bananas.
Each of these bananas causes a change in the augmented persistence diagram, which pays for the visit.

\medskip \noindent {\sc Step 2:} \textbf{Stack bananas.}
After locating the smallest banana in $\UpTree{f}$ that suffers an injury or fatality, we find the others by traversing the tree upward.
In the process, we load three initially empty stacks with the injuries and fatalities.
In doing so, we distinguish spanning min-max pair to the left of $x$, separated by $x$, to the right of $x$, and denote the corresponding stacks $\Lup$, $\Mup$, $\Rup$, respectively.
On its way up, the algorithm pushes each banana on one of the three stacks, until it encounters the first banana with maximum in the spine.
\begin{tabbing}
  m\=m\=m\=m\=m\=m\=m\=m\=m\=\kill
  \> \texttt{function} {\sc LoadStacks} $(x)$: \\*
  \> \> $(p, q) = \textsc{SmallestBanana} (x)$; \\
  \> \> \texttt{loop} \= \texttt{case} \= $p < x$ \texttt{and} $q < x$: $\textsc{Push} (\Lup, (p,q))$; \\*
  \> \>              \>            \> $p < x$ \texttt{xor} $q < x$: $\textsc{Push} (\Mup, (p,q))$;  \\*
  \> \>              \>            \> $p > x$ \texttt{and} $q > x$: $\textsc{Push} (\Rup, (p,q))$ \\*
  \> \> \> \texttt{endcase}; \\
  \> \> \> \texttt{if} $q$ is in spine of $\UpTree{f}$ \texttt{then} \texttt{exit} \texttt{endif}; \\*
  \> \> \> $p = \Low{q}$; $q = \Dth{p}$ \\*
  \> \> \texttt{forever}.
\end{tabbing}
Similarly, we load the initially empty stacks $\Ldn$, $\Mdn$, $\Rdn$ with bananas in $\DnTree{f}$.
When we split the up-tree, it will be convenient to reverse the pairs defining the bananas in the down-tree, and when we split the down-tree, we reverse the pairs defining the bananas in the up-tree.
We call a stack with bananas $(p_0, q_0)$ at the bottom to $(p_j, q_j)$ at the top \emph{sorted} if $f(p_{j}) < \ldots < f(p_0) < f(q_0) < \ldots < f(q_j)$.
The stacks satisfy the following properties:
\begin{lemma}[Sorted Stacks]
  \label{lem:sorted_stacks}
  Let $c_1, c_2, \ldots, c_{m}$ be a sequence of distinct values, $f$ the thus defined piecewise linear map, $x$ such that $c_2 < x < c_{m-1}$, and $\Lup, \Mup, \Rup, \Ldn, \Mdn, \Rdn$ the stacks as returned by {\sc LoadStacks} holding the bananas with injury, fatality, and scare in $\UpTree{f}$ and $\DnTree{f}$.
  Then
  \smallskip \begin{enumerate}[(i)]
    \item all six stacks are sorted;
    \item all bananas on the stacks are with simple wave, except for the last banana pushed onto $\Lup, \Mup, \Rup$, whose maximum belongs to the spine of $\UpTree{f}$, and the last banana pushed onto $\Ldn, \Mdn, \Rdn$, whose maximum belongs to the spine of $\DnTree{f}$;
    \item the set of simple wave bananas on $\Mup$ is the same as that on $\Mdn$;
    \item $\Lup$, $\Mup$, $\Ldn$ can be merged into a single sorted stack, and so can $\Rup$, $\Mup$, $\Rdn$ and $\Lup$, $\Mup$, $\Rup$ and $\Ldn, \Mdn, \Rdn$.
  \end{enumerate}
\end{lemma}
\begin{proof}
  Property~{\sf (i)} is implied by Condition~{\sf II} on the path-decomposition and the fact that the algorithm visits the bananas from bottom to top.
  Property~{\sf (ii)} holds because all bananas are with simple wave, other than the ones with maxima in the spine, and except for the respective last ones, no banana pushed onto the stacks have their maxima in the spine.
  To see Property~{\sf (iii)} recall that a simple wave of $f$ is also a simple wave of $-f$.
  Hence, simple wave bananas belong to both trees, except that they share the mid-trail but not the in-trail.
  Each banana on $\Mup$ and $\Mdn$ suffers a fatality, which cuts
  cuts the mid-trail and therefore also the corresponding banana on the other stack.
  Property~{\sf (iv)} follows from the fact that all bananas on one of these stacks must correspond to nested windows.
  Note, however, that it is not true that all six stacks can be merged to a sorted stack.
  The reason is that triple-panel windows of bananas in $\Lup$ and $\Rdn$ can overlap without being nested, and so can triple-panel windows of bananas in $\Ldn$ and $\Rup$; see the windows spanned by $\T{f}, \T{g}$ and $\T{h}, \T{i}$ in Figure~\ref{fig:maponly}.
  However, if $x$ cuts through any two such overlapping and not nested windows, then it also separates the critical points that span one from those that span the other.
  Such pairs neither exist for $\Lup, \Mup, \Ldn$, nor for any of the other three triplets of stacks.
\end{proof}

There are up to two windows with short wave which need special treatment:
First, the window of $-f$ with its maximum on the spine of $\DnTree{f}$ and $x$ in its in-panel;
second, the window of $f$ with its maximum on the spine of $\UpTree{f}$ and $x$ in neither its in-panel or mid-panel.
The former appears in $\Ldn$ or $\Rdn$, but is not a window of $\UpTree{f}$ and we ignore it when splitting $\UpTree{f}$.
The latter is not loaded into $\Lup$, $\Mup$ or $\Rup$, is not a window in $\DnTree{f}$ and is thus not loaded onto any stack.
However, it may have $x$ in its out-panel, in which case it suffers a scare and should be pushed onto $\Ldn$ or $\Rdn$.
We can identify this case by examining $q' = \In{q_j}$, where $(p_j, q_j)$ is the topmost banana.
If $(\Birth{q'}, q')$ pushed onto $\Ldn$ or $\Rdn$ preserves the properties of \cref{lem:sorted_stacks} and $f(\Birth{q'}) < f(x) < f(q')$,
then indeed $x$ is in the out-panel of this window and we place it on $\Ldn$ or $\Rdn$ as appropriate when splitting $\UpTree{f}$.
Otherwise, no window with $x$ in its out-panel is missing from the stacks.
Symmetrically, when splitting $\DnTree{f}$ we ignore in $\Lup \cup \Rup$ the window with maximum on the spine of $\UpTree{f}$
and push the missing window in $\DnTree{f}$ onto $\Lup$ or $\Rup$.

After pushing the bananas onto the stacks, we remove them from the top.
Write $(p,q) = \textsc{Top} (\Lup)$ for the topmost banana on $\Lup$, with $(p,q) = (\texttt{nil}, \texttt{nil})$ if $\Lup$ is empty, and similarly for the other stacks.
We set $f(\texttt{nil}) = - \infty$.
The overall top banana is the one whose maximum has the largest value and is returned by function  \textsc{TopBanana}:
\begin{tabbing}
  m\=m\=m\=m\=m\=m\=m\=m\=m\=\kill
  \> \texttt{function} \textsc{TopBanana}: \\*
  \> \> $\Stack = \texttt{nil}$; $(p,q) = (\texttt{nil}, \texttt{nil})$; \\*
  \> \> \texttt{for} \= \texttt{each} $S \in \{\Lup, \Rup, \Mup, \Ldn, \Rdn\}$ \texttt{do}
           $(p',q') = \textsc{Top}(S)$; \\*
  \> \> \> \texttt{if} $f(q') > f(q)$ \texttt{then} $\Stack = S$; $(p,q) = (p',q')$ \texttt{endif} \\* 
  \> \> \texttt{endfor};  \texttt{return} $(\Stack, (p,q))$.
\end{tabbing}

\noindent \textsc{Step 3:}
\textbf{Split up-tree.}
This operation splits the up-tree into two and in the process adds at most four new nodes:  a second special root, a minimum and a maximum if $\ell$ and $\ell+1$ are non-critical points prior to the cut, and the up-type endpoint of the two new hooks on both sides of the cut.
The algorithm uses the loaded stacks to visit the nodes that need repair from top to bottom, i.e.\ from the largest to the smallest banana as determined by the stacks.
The iteration ends when the stacks are empty:
\begin{tabbing}
  m\=m\=m\=m\=m\=m\=m\=m\=m\=\kill
  \> \texttt{function} $\textsc{Split} (x)$: \\*
  \> \> $\textsc{LoadStacks} (x)$; \\*
  \> \> \texttt{loop} \= $(\Stack, (p,q)) = \textsc{TopBanana}$; \\*
  \> \> \> \texttt{if} $(p,q) = (\texttt{nil}, \texttt{nil})$ \= \texttt{then} \texttt{exit} \texttt{endif}; \\ 
  \> \> \> \texttt{case} \= $\Stack \in \{ \Lup, \Rup\}$: $\textsc{DoInjury} (p,q)$; \\*
  \> \> \>               \> $\Stack = \Mup$: $\textsc{DoFatality} (p,q)$; \\*
  \> \> \>               \> $\Stack \in \{ \Ldn, \Rdn \}$: $\textsc{DoScare} (p,q)$ \\*
  \> \> \> \texttt{endcase}; $\textsc{Pop} (\Stack)$ \\*
  \> \> \texttt{forever}.
\end{tabbing}
Finally, we add $\ell$ and $\ell+1$ as new nodes if they become critical in the process, and we do the final adjustment to the values of the new hooks.
In each iteration, we have two banana trees, one on the left and the other on the right.
An injury transfers part of one tree to the other, and so does a fatality.
The latter also adjusts the pairing between the critical points, while a scare only adjusts the pairing.
To initialize this setting, we construct a second banana tree consisting of a single banana with two empty trails connecting its special root with a dummy leaf, $\alpha$.
If we cut the original banana tree on the left spine or the special banana, then this new tree becomes the left tree and we define its special root to be at $-\infty$ along the interval.
Otherwise, the new tree becomes the right tree and we define the special root of the original tree to be at $-\infty$, swapping in- and mid-trails of the special banana to satisfy the uniqueness conditions.
After splitting is complete, we reset the special roots to $-\infty$ and swap the in- and mid-trails of the special banana as needed.%
For reasons that will become clear shortly, we set $f(\alpha) = f(p_j) - \varepsilon$, in which $p_j, q_j$ are the nodes that span the top banana on the stacks, and $\varepsilon > 0$ is smaller than the difference between the values of any two items.
We observe that the top banana in the first iteration either suffers an injury or a fatality.
Indeed, if it suffered a scare, then $x$ would cut through its out-panel and therefore lie between $q_j$ and $\Low{q_j}$.
But then $x$ cuts through the in- or mid-panel of the banana spanned by $\Low{q_j}$ and $\Dth{\Low{q_j}}$, and this banana would have been the top banana on the stacks, which is a contradiction.
To see the correctness of the splitting operation, we note that Function \textsc{Split} maintains the following invariants:
\smallskip \begin{enumerate}[(i)]
    \item both banana trees satisfy the uniqueness conditions \textsf{I.1}, \textsf{I.2}, \textsf{II}, \textsf{III.1}, and \textsf{III.2};
    \item any node $u$ that is not on any stack and does not have an ancestor on the stacks is in the left tree, if $u < x$, and in the right tree, if $u > x$.
\end{enumerate} \smallskip
Invariant \textsf{(ii)} implies that once the stacks are empty, all nodes to the left of $x$ are in the left tree and all nodes to the right of $x$ are in the right tree.
By invariant \textsf{(i)} the right and left trees are the unique trees representing the maps $g$ and $h$.
Assuming $x < b$ for the root of the banana tree, the first transfer will be from right to left (as in Figure~\ref{fig:healing}), so we let the existing banana tree be the right tree and the banana with the two empty trails the left tree.
We are now ready to discuss the actions specific to the injury, fatality, and scare of a banana.

\medskip
An \emph{injury} occurs when $x$ cuts the in-trail of the banana spanned by $p, q$.
To simplify the discussion, assume $\Stack = \Rup$, so $x < p < q$, as illustrated in Figure~\ref{fig:injuries}, left, and Figure~\ref{fig:healing}, middle and left.
The case $\Stack = \Lup$ and $q < p < x$ is symmetric.

\smallskip
The injury causes a possibly empty portion of the in-trail to split off and append to the rightmost banana spanned by a spine node of the tree on the left; see again Figure~\ref{fig:healing}.
The cut off portion consists of all interior nodes $j < x$ and their sub-bananas.
Note that it is quite possible that $x$ cuts the string somewhere in a sub-banana with upper end on the in-trail.
In this case, the sub-banana does not transfer as its upper end is to the right of $x$, and it is instead subject to a later repair operation.
Because the bananas taken from the stacks are sorted, we have $f(i) > f(\alpha)$ for all transferred nodes $i$.
We maintain $\alpha$ as an up-type endpoint that spans a banana with $b$, as before the operation.

\begin{figure}[htb]
  \centering \vspace{0.0in}
  \resizebox{!}{1.4in}{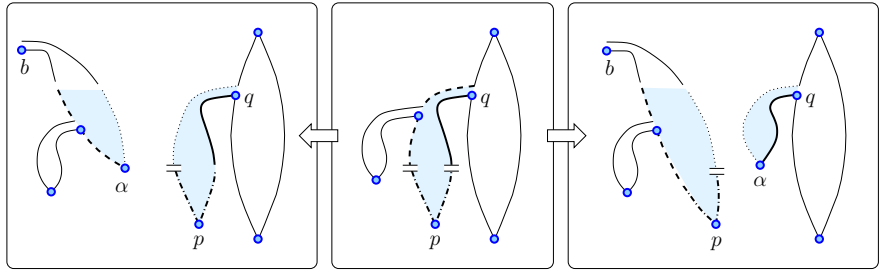}
  \caption{\emph{Middle:} the banana spanned by $p, q$ suffers an injury (cut in in-trail) or a fatality (cut in mid-trail).
  \emph{Left:} to process the injury, we move the \emph{dashed} portion and append it to the rightmost spine banana of the left tree.
  The up-type node, $\alpha$, is paired with the upper end, $b$, of the expanded banana.
  \emph{Right:} to process the fatality, we move the entire in-trail together with the \emph{dash-dotted} portion of the mid-trail.
  The new up-type endpoint, $\alpha$, is paired with $q$ in the right tree, and $p$ is paired with the upper end, $b$, of the expanded banana.
  Dotted portions of trails are empty.}
  \label{fig:healing}
\end{figure}

\medskip
A \emph{fatality} occurs when $x$ cuts the mid-trail of the banana spanned by $p, q$.
Hence, $\Stack = \Mup$, and we assume $p < x < q$, which is the case illustrated in Figure~\ref{fig:injuries}, middle, and Figure~\ref{fig:healing}, middle and right.
The case $q < x < p$ is symmetric.

The fatality causes the entire in-trail and a possibly empty portion of the mid-trail to split off and append to the rightmost banana spanned by a spine node of the tree on the left; see again Figure~\ref{fig:healing}.
The transfer includes the minimum, $p$, and all interior nodes, $j$, with $j < x$.
As in the case of an injury, it is possible that $x$ cuts a sub-banana with upper end on the mid-trail.
This sub-banana is not transferred and instead subject to later repair operation.
Since $p$ moves from the right to the left, we move $\alpha$ from the left to the right and adjust the pairing accordingly: $p$ spans a banana with $b$ in the left tree, and $\alpha$ spans a banana with $q$ in the right tree.
To ensure that the path-decomposition remains correct, we set $f(\alpha) = f(p) - \varepsilon$.

\medskip
A \emph{scare} occurs when $x$ lies in the out-panel of the triple-panel window spanned by $p, q$, so $x$ cuts neither trail of the corresponding banana.
However, there is a banana spanned by $q, p$ in $\DnTree{f}$, and $x$ cuts the in-trail of that banana.
Assume $\Stack = \Ldn$, so $p < q < x$, which is the case illustrated in Figure~\ref{fig:injuries}, right.
The case $\Stack = \Rdn$ and $x < q < p$ is symmetric.

\smallskip
The scare does not affect the right up-tree, and it affects the left up-tree only indirectly, namely by triggering a change in the pairing.
In other words, it preserves the underlying ordered binary tree (First Step of the banana tree construction), but it changes the path-decomposition (Second Step), and therefore also the organization of the bananas (Third Step).
This is done by executing an interchange of two minima, namely of $p$ and $\alpha = \Low{\Dth{p}}$.
To justify this interchange, we adjust the value of $\alpha$ to $f(\alpha) = f(p) + \varepsilon$.

\medskip \noindent \textbf{Gluing two banana trees.}
Concatenating the lists $c_1, c_2, \ldots, c_{\ell}$ and $c_{\ell+1}, c_{\ell+2}. \ldots, c_{m}$ is the inverse of cutting $c_1, c_2, \ldots, c_{m}$ into these two lists.
Equivalently, we can think of concatenating $g \colon [0, \ell+1] \to \Rspace$ and $h \colon [\ell, m+1] \to \Rspace$ to get $f = g \cdot h \colon [0,m+1] \to \Rspace$.
A triple-panel window with simple wave of $g$ is still a triple-panel window with simple wave of $f$, and similarly for $h$ and $f$.
Also a triple-panel window whose wave is short at the left end of the domain of $g$ is still a triple-panel window with short wave of $f$, and similarly for $h$ and the right end of the domain of $h$.
It follows that the only windows that need repair are the global windows of $g$ and $h$ and the windows whose waves are short at the right end of the domain of $g$ or the left end of the domain of $h$.
These windows correspond to the bananas rooted at nodes of the right spine of $\UpTree{g}$ and the left spine of $\UpTree{h}$.

\smallskip
In a nutshell, the algorithm visits the nodes in the relevant portions of the two spines from bottom to top.
For $\UpTree{g}$, we get a sequence of nested short wave windows ending with the global window, and we list the corresponding critical points from right to left as $a_0, b_0, a_1, b_1, \ldots, a_i, b_i = \beta_g$.
Symmetrically, we list the critical points we get from $\UpTree{h}$ from left to right as $a_0', b_0', a_1', b_1', \ldots. a_j', b_j' = \beta_h$.
To have a specific setting, we assume that $a_0$ is a proper minimum of $g$ (an up-type endpoint after removing the hook), while $a_0'$ is a hook that is an up-type endpoint of $h$.
Because the windows are nested, we have
\begin{align}
  f(a_i) < \ldots < f(a_0) < f(x) < f(b_0) < \ldots < f(b_i); \\*
  f(a_j') < \ldots < f(a_1') < f(x) < f(b_0') < \ldots < f(b_j'),
\end{align}
in which the hook, $a_0'$, has been removed from the second list; see Figure~\ref{fig:simplified}.
The two orderings imply that the minimum with largest value is either $a_0$ or $a_1'$, and the maximum with smallest value is either $b_0$ or $b_0'$.
To get a sorted list of nested windows, the algorithm pairs up the lowest maximum with the highest minimum, removes the two, and iterates.
\begin{figure}[htb]
  \centering \vspace{0.0in}
  \resizebox{!}{0.9in}{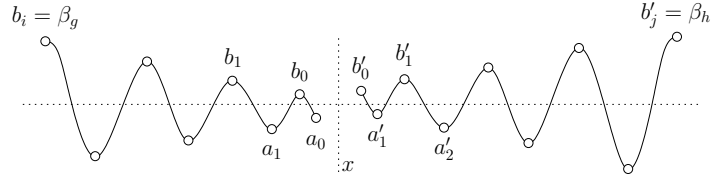}
  \vspace{-0.05in}
  \caption{The simplified graph of $g$ to the \emph{left} of $x$, which connects the critical points that span bananas in the right spine of $\UpTree{g}$.
  The symmetrically simplified graph of $h$ to the \emph{right} of $x$, which connects the critical points that span bananas in the left spine of $\UpTree{h}$.}
  \label{fig:simplified}
\end{figure}

Given $a_0$, the next maximum and then the next minimum to its left are $b_0 = \Dth{a_0}$ and $a_1 = \Low{b_0}$.
Symmetrically, $a_1' = \Low{b_0'}$ and $b_1' = \Dth{a_1'}$.
Assuming the configuration in Figure~\ref{fig:simplified}, namely $f(a_0) < f(x) < f(b_0')$, $a_0$ remains a minimum and $b_0'$ remains a maximum after connecting the two items monotonically.
In other cases, $a_0$ and $b_0'$ may become non-critical.
As a pre-processing step to gluing, we remove up-type and down-type endpoints that become non-critical when we connect the two lists.
Whatever the case, the algorithm assumes that $a_0$ and $b_0'$ are the first critical points to the left and right of $x$, and that $a_0$ is a minimum and $b_0'$ is a maximum.
We write $\alpha$ for the hook at $b_0'$ and use it as a dummy leaf, like in cutting.
For the purposes of this algorithm we define $f(\beta_h) < f(\beta_g)$ for tie-breaking.
We also define $\Low{\beta_g} = \texttt{nil}$ and $\Low{\beta_h} = \texttt{nil}$ with $f(\texttt{nil}) = -\infty$.
Again, we define the special root of the left tree, $\beta_g$, to be at $-\infty$ along the interval and swap the in- and mid-trails of the special banana of the left tree.
This is reset after the gluing is complete.%
In each iteration, the maximum with lower value is paired with a minimum to its left or its right.
If the maximum in the left tree, $q$, is processed, then the candidate minima are $\glueplow = \Low{q}$, which is on the left of $q$, and $\gluepbth$, which is on the same side of $q$ as $\Birth{q}$.
If the maximum in the right tree is processed, then the situation is symmetric.
The iteration ends when one of the two trees is empty, i.e., consists only of a special root and the dummy $\alpha$. 
\begin{tabbing}
  m\=m\=m\=m\=m\=m\=m\=m\=m\=\kill
  \> \texttt{function} $\textsc{Glue} (a_0, b_0', \beta_g, \beta_h)$: \\*
  \> \> set $q = b_0 = \Dth{a_0}$; $q' = b_0'$; \\*
  \> \> \texttt{repeat} \= \texttt{if} \= $f(q) < f(q')$ \texttt{then} \\*
  \> \> \> \> $\glueplow = \Low{q}$; \texttt{if} $\gluepbth \neq \alpha$ \texttt{then} $\gluepbth = \Birth{q}$ \texttt{else} $\gluepbth = \Birth{q'}$ \texttt{endif}; \\
  \> \> \> \> \texttt{case} \= $f(\gluepbth) > f(\glueplow)$ \texttt{and} $\gluepbth = \Birth{q}$: $\textsc{UndoInjury}(\gluepbth,q)$; \\*
  \> \> \> \>               \> $f(\gluepbth) > f(\glueplow)$ \texttt{and} $\gluepbth = \Birth{q'}$: $\textsc{UndoFatality}(\gluepbth,q)$; \\*
  \> \> \> \>               \> $f(\glueplow) > f(\gluepbth)$: $\textsc{UndoScare}(\glueplow,q)$ \\*
  \> \> \> \> \texttt{endcase};
              $q = \Dth{\Low{q}}$ \\
  \> \> \> ~~\texttt{else} symmetric cases for $f(q) > f(q')$ \\*
  \> \> \> \texttt{endif} \\
  \> \> \texttt{until} $\In{\beta_g} = \Mid{\beta_g} = \alpha$ or $\In{\beta_h} = \Mid{\beta_h} = \alpha$; \\*
  \> \> \texttt{if} $\In{\beta_g} = \Mid{\beta_g} = \alpha$ \= \texttt{then} discard $\beta_g, \alpha$ \texttt{else} discard $\beta_h, \alpha$ \texttt{endif}.
\end{tabbing}
The main three subroutines are inverses of the earlier functions and undo injuries, fatalities, and scares.
Function \textsc{UndoInjury} extends the short wave banana beginning at $q$ to a simple wave banana by inserting maxima into the in-trail beginning at $q$; see Figure~\ref{fig:healing} but read the change backward, from the left to the middle.
Function \textsc{UndoFatality} extends the short wave banana beginning at $q$ to a simple wave banana while also changing the pairing of $q$; see again Figure~\ref{fig:healing} now reading the change from the right to the middle.
Finally, function \textsc{UndoScare} adjust the value of the up-type endpoint, $\alpha = \Birth{q}$, to $f(\alpha) = f(\glueplow) - \varepsilon$,
thereby extending the out-panel of the window associated with $q$ and changing the pairing.

\medskip \noindent \textbf{Summary.}
We summarize by stating the running-time for cutting and concatenating lists in terms of the number of critical points and the number of changes to the augmented persistence diagrams.
\begin{theorem}[Time to Cut and Concatenate]
  \label{thm:time_to_cut_and_concatenate}
  Let $f \colon [0,m+1] \to \Rspace$ be a generic piecewise linear map with $n$ maxima, and let $g$ and $h$ such that $f = g \cdot h$.
  The time to cut $f$ into $g$ and $h$ is $O(\log n + k)$, in which $k$ is the size of the symmetric difference between $\DgmR{}{f}$ and $\DgmR{}{g} \sqcup \DgmR{}{h}$, and so is the time to concatenate $g$ and $h$ to form $f$.
\end{theorem}
\begin{proof}
  We focus on the algorithm for cutting $f$ at $x$, and omit the argument for glueing, which is symmetric.
  It takes $O(\log n)$ time to find the two consecutive critical points that sandwich $x$ between then, and it takes $O(\log n)$ time to cut the dictionaries accordingly.
  The time needed to split $\UpTree{f}$ and $\DnTree{f}$ is proportional to a constant plus the number of nodes visited during the iteration.
  
  \smallskip
  We argue that we can charge each such node, $q$, to a change in the augmented persistence diagram in such a way that each change is charged at most twice, once by a node in $\UpTree{f}$ and once by a node in $\DnTree{f}$.
  Let $q$ be a visited maximum in $\UpTree{f}$, and let $p$ be the minimum with $\Dth{p} = q$.
  Then $\Window{p}{q}$ is a window with simple or short wave cut by $x$.
  If $x$ cuts through the mid-panel or the out-panel, then $\Window{p}{q}$ is neither a window of $g$ nor of $h$, and we charge $q$ to the disappearance of the point $(f(p), f(p))$ from the ordinary subdiagram.
  So suppose $x$ cuts through the in-panel.
  If $\Window{p}{q}$ is with simple wave, $\Window{q}{p}$ is a window of $-f$ and $x$ cuts through its out-panel, so we charge $q$ to the disappearance of $(f(q), f(p))$ from the relative subdiagram.
  Finally, if $\Window{p}{q}$ is a window with short wave and $x$ cuts through its in-panel, then the algorithm has reached the spine of $\UpTree{f}$.
  By Lemma~\ref{lem:spines_and_windows}, all ancestors of $q$ are of the same type, so we can end the traversal here. Thus, this case causes only constant work which is charged to the cutting operation directly.
\end{proof}

We remark that the $O(\log n + k)$ time bound for cutting and concatenating would not hold if in splitting and gluing we traversed the banana trees all the way to the respective special roots.
Let $q$ be the node on the spine where the algorithm halts.
The path connecting $q$ to the special root may be arbitrarily long, and none of these nodes corresponds to a change in the augmented persistence diagram.
A particular map for which this happens is the \emph{damped sine function}, defined by $f(x) = x \sin x$; see the right half of Figure~\ref{fig:simplified} for a sketch.
We get a triple-panel window for each min-max pair, each with short wave on the left.
If we cut $f$ close to $0$, then all these windows get insubstantially smaller but remain to be spanned by the same min-max pairs.
These changes are not visible in the augmented persistence diagrams of $f$ and the functions created by cutting $f$ at $x$.

\section{Discussion}
\label{sec:5}

The main contribution of this paper is a dynamic data structure for maintaining the augmented persistence diagram of linear lists.
The data structure starts with the Cartesian tree of Vuillemin \cite{Vui80} (see also Aragon and Seidel \cite{ArSe96}), which it decomposes into paths and then splits into pairs of parallel trails, arriving at an unconvential representation referred to as banana tree.
For the efficiency of the data structure, it is essential to maintain a pair of banana trees, which store complementary information about the list.
The most important next step is the implementation of the data structure and its algorithms.
With such an implementation, we can deepen our understanding of the persistence of random lists, which in turn can be used as a baseline for our understanding of non-random time series.

\smallskip
The theoretical foundations for our dynamic algorithms are described in \cite{BCES21}, where maps on $1$-dimensional domains, which include but go beyond intervals are described.
Can the banana trees be extended to geometric trees (which allow for bifurcations) and geometric networks (which allow for bifurcations as well as loops) without deterioration of the asymptotic running time?
Because of the unconventional representation of trees in terms of bananas (pairs of parallel trails), we finally ask whether there are other dynamic data structure questions that can benefit from this representation.

\clearpage

\appendix

\section{Correctness of Local Maintenance}
\label{sec:local-maintenance-correctness}

In this section, we prove the correctness of the local operations and of the algorithms for Scenarios~A and B (see \cref{lem:scenario-A,lem:scenario-B} in \cref{sec:scenario-correctness}).
We also give details on how to treat changes in value that lead to an up-type item becoming down-type or vice versa.
We make extensive use of the homotopy $h_\lambda \colon [0,m+1] \to \Rspace$ defined in \cref{sec:4.2}
as $h_\lambda (x) = (1-\lambda) f(x) + \lambda g(x)$ for $0 \leq \lambda \leq 1$, where $f$ and $g$ differ in the value of a single item.
The change in function value from $f$ to $g$ can be arbitrary and the transformation of $\UpTree{f}$ into $\UpTree{g}$ is achieved by a sequence of local operations,
which we define in terms of small changes in the function value of a single item.
To make the notion of small change more precise, we define \emph{contiguity} of items in terms of function value.
\begin{Definition}[Contiguity]
    \label{def:contiguous-value}
    Items $p$ and $q$ are \emph{contiguous in $f$} if there exists no item $u \neq \{p,q\}$ such that $f(u) \in [\min \{f(p), f(q)\}, \max \{f(p), f(q)\}]$.
\end{Definition}
Each local operation involving items $p$ and $q$ is then defined as the transformation of $\UpTree{h_\sigma}$ into $\UpTree{h_\theta}$,
where $p$ and $q$ are contiguous in $h_\sigma$ and $h_\theta$.
We prove the correctness of local operations in terms of such a small change,
then show how to extend this to larger changes where items are not necessarily contiguous.
This allows us to combine the local operations to form the Scenarios~A and B.
It is easy to see that when items remain contiguous upon exchanging values then they swap places in the order of items by function value.
This is stated in the following lemma.

\begin{lemma}[Local Changes]
    \label{lem:stable-value-order}
    Let $0 \leq \sigma < \theta \leq 1$.
    Let $p$ and $q$ be items that are contiguous in both $h_\sigma$ and $h_\theta$ with $h_\sigma(p) < h_\sigma(q)$ and $h_\theta(p) > h_\theta(q)$.
    The order of items by function value is the same in $h_\sigma$ and $h_\theta$ except that $p$ and $q$ are swapped.
\end{lemma}

In some of the correctness proofs, we analyze how a change in the function value of an item affects the structure of windows
and leverage the correspondence between windows and bananas established in \cref{lem:bananas_and_windows}.
In other correctness proofs, we use the fact that there is a unique banana tree satisfying the conditions 
proved in \cref{lem:uniqueness_of_banana_tree}.
To this end we reformulate the uniqueness conditions in terms of individual nodes and their pointers. This requires the following definitions.
\begin{Definition}[Ancestor and Descendant]
    A node $p$ is an \emph{ancestor} of an internal node $q$ if $p = \UpOp^*(q)$
    and of a leaf $r$ if $p = \UpOp^*(\In{r})$ or $p = \UpOp^*(\Mid{r})$, where $\UpOp^*$ denotes a sequence of $\Up{\cdot}$-pointers.
    A node $s$ is a \emph{descendant} of a node $t$ if $t$ is an ancestor of $s$.
\end{Definition}
\begin{Definition}[Banana subtree]
    The \emph{banana subtree rooted at a maximum $q$} is the banana tree consisting of the banana spanned by $\Birth{q}$ and $q$
    and all bananas whose maximum has an ancestor other than $q$ on the in-trail or mid-trail starting at $\Birth{q}$ and ending at $q$.
\end{Definition}
We now state the new uniqueness conditions as \cref{inv:condition-I,inv:condition-II,inv:condition-III-max}.
\begin{invariant}
    \label{inv:condition-I}
   For each maximum $q$ of $f$, except for the special root, it holds that
    \begin{enumerate}
        \item if $\Birth{q} < q$, then all nodes $u \neq q$ in the banana subtree rooted at $q$ satisfy $u < q$
            and all descendants $v$ of $\Dn{q}$, including $\Dn{q}$, satisfy $v > q$;
        \item if $\Birth{q} > q$, then all nodes $u \neq q$ in the banana subtree rooted at $q$ satisfy $u > q$
            and all descendants $v$ of $\Dn{q}$, including $\Dn{q}$, satisfy $v < q$.
    \end{enumerate}
\end{invariant}
\begin{invariant}
    \label{inv:condition-II}
    For each minimum $p$ of $f$, except for the global minimum, it holds that $f(p) > f(\Low{\Dth{p}})$.
\end{invariant}
\begin{invariant}
    \label{inv:condition-III-max}
    For each maximum $q$ of $f$, except for the special root, it holds that
    \begin{enumerate}
        \item $f(\Up{q}) > f(q) > f(\Dn{q})$;
        \item if $q \neq \In{\Up{q}}$, then $\Up{q} < q < \Dn{q}$ or $\Dn{q} < q < \Up{q}$;
        \item if $q = \In{\Up{q}}$, then
            either $\Up{q} < q$ and $\Dn{q} < q$ or $\Up{q} > q$ and $\Dn{q} > q$.
    \end{enumerate}
\end{invariant}

We show that there is a unique banana tree satisfying \cref{inv:condition-I,inv:condition-II,inv:condition-III-max} for a given map $f$.
We first show that a banana tree satisfying \cref{inv:condition-I,inv:condition-III-max} also satisfies Conditions {\sf III.1} and {\sf III.2}.
\begin{lemma}[Invariants imply Banana Conditions]
    \label{lem:inv-1-and-3-are-cond-3}
    A banana tree $\mathcal{B}$ that satisfies \cref{inv:condition-I,inv:condition-III-max}
    also satisfies Conditions {\sf III.1} and {\sf III.2}.
\end{lemma}
\begin{proof} 
    Let $p$ and $q$ be a minimum and a maximum spanning a banana in $\mathcal{B}$.
    Assume $p < q$; the other case is symmetric.
    The requirement on function values increasing along the trails from $p$ to $q$ follows immediately from $f(\Dn{u}) < f(u) < f(\Up{u})$.
    We now show that the trails are also ordered along the interval.
    By \cref{inv:condition-I}, $\Mid{q} < q$ and $\In{q} < q$.
    Write $p=v_0,v_1,\dots,v_\ell=q$ for the nodes along the mid-trail with $v_1 = \Mid{p}$, $v_{\ell-1} = \Mid{q}$
    and $v_i = \Up{v_{i-1}}$ for $2 \leq i \leq \ell$ and $v_i = \Dn{v_{i+1}}$ for $0 \leq i \leq \ell - 2$.
    If $v_{\ell-1} = \Mid{q}$ is a maximum,
    then by \cref{inv:condition-III-max} and $\Mid{q} < q$ it satisfies $\Dn{v_{\ell-1}} < v_{\ell-1} < \Up{v_{\ell-1}} = q$.
    If $v_{\ell-2}$ is a maximum,
    then it in turn satisfies $\Dn{v_{\ell-2}} < v_{\ell-2} < \Up{v_{\ell-2}} = v_{\ell-1}$ by \cref{inv:condition-III-max}.
    Repeating this argument for all $v_i$ for $1 \leq i \leq {\ell - 3}$ it follows that
    $p < v_1 < \dots < v_\ell$.

    \smallskip
    Write $p = u_0,u_1,\dots,u_k$ for the nodes along the in-trail with $u_1 = \In{p}$, $u_{k-1} = \In{q}$
    and $u_i = \Up{u_{i-1}}$ for $2 \leq i \leq k$ and $u_i = \Dn{u_{i+1}}$ for $0 \leq i \leq k - 2$.
    By \cref{inv:condition-III-max} and $\In{q} < q$ we have $u_{k-1} < \Dn{u_{k-1}}$.
    Following a similar argument as for the $v_i$ we get $\Dn{u_i} > u_i > \Up{u_i}$ for all $1 \leq i \leq k-2$
    and it follows that $p = u_0 > u_1 > \dots > u_{k-1}$.
    Thus, as required by Conditions {\sf III.1} and {\sf III.2}, the trails are ordered along the interval and by function value,
    i.e., $\mathcal{B}$ satisfies Conditions {\sf III.1} and {\sf III.2}.
\end{proof}

The next lemma states that a path-decomposed binary tree can be transformed into a banana tree satisfying the invariants.
\begin{lemma}[Conditions imply Banana Invariants]
    \label{lem:pdbt-to-banana-tree}
    If a path-decomposed binary tree $T$ fulfills the Conditions {\sf I.1}, {\sf I.2} and {\sf II},
    then the banana tree obtained from $T$ as described in \cref{sec:3.2} satisfies \cref{inv:condition-I,inv:condition-II,inv:condition-III-max}.
\end{lemma}
\begin{proof}
    Let $T$ be a path-decomposed binary tree satisfying Conditions {\sf I.1}, {\sf I.2} and {\sf II}
    and let $\mathcal{B}$ be the banana tree obtained from $T$ as described in \cref{sec:3.2}.
    We begin by proving that the banana tree $\mathcal{B}$ satisfies \cref{inv:condition-III-max}.
    Recall that $\mathcal{B}$ satisfies Conditions {\sf III.1} and {\sf III.2}
    and that by \cref{lem:uniqueness_of_banana_tree} it is unique.
    To see that this implies \cref{inv:condition-III-max} consider any banana spanned by a minimum $a$ and a maximum $b$.
    Assume $a < b$; the argument for $a > b$ is symmetric.
    Let $a = u_0,u_1,\dots,u_j = b$ be the nodes along the in-trail
    and $a = v_0,v_1,\dots,v_k = b$ be the nodes along the mid-trail.
    By Conditions {\sf III.1}, {\sf III.2} we have
    \smallskip \begin{enumerate}
        \item $u_i > u_{i+1}$ for all $0 \leq i \leq j-2$ and $f(u_i) < f(u_{i+1})$ for all $0 \leq i \leq j-1$,
        \item $v_i < v_{i+1}$ for all $0 \leq i \leq k-1$ and $f(v_i) < f(v_{i+1})$ for all $0 \leq i \leq k-1$.
    \end{enumerate} \smallskip
    The pointers $\Dn{\cdot}$ and $\Up{\cdot}$ are defined such that
        $\Up{u_i} = u_{i+1}$, $\Dn{u_i} = u_{i-1}$ for all $1 \leq i \leq j-1$
    and $\Up{v_i} = v_{i+1}$, $\Dn{v_i} = v_{i-1}$ for all $1 \leq i \leq k-1$.
    Together with Conditions {\sf III.1} and {\sf III.2} this implies that
    all maxima $q$ on a trail between $a$ and $b$ satisfy
    $f(\Dn{q}) < f(q) < f(\Up{q})$ 
    and, except $u_{j-1}$, they satisfy
    $\Dn{q} < q < \Up{q}$ or $\Up{q} < q < \Dn{q}$.
    By Condition {\sf III.1} it holds that $\Dn{u_{j-1}} = u_{j-2} > u_{j-1}$,
    and by the assumption that $a < b$ it holds that $b = \Up{u_{j-1}} > u_{j-1}$.
    This shows that \cref{inv:condition-III-max} is satisfied for
    all $u_i$ with $1 \leq i \leq j-1$ and $v_i$ with $1 \leq i \leq k-1$,
    i.e, for all nodes internal on a trail from $a$ to $b$.
    Since all maxima except the special root are internal to some trail
    it follows that $\mathcal{B}$ satisfies \cref{inv:condition-III-max}.

    \smallskip
    We now prove that $\mathcal{B}$ satisfies \cref{inv:condition-II}.
    We need to show that for every minimum $p$ it holds that $f(p) > f(\Low{\Dth{p}})$.
    Assume by contradiction that for some minimum $p$ this is not the case and let $a = \Low{\Dth{p}}$, i.e., we assume $f(a) = f(\Low{\Dth{p}}) > f(p)$.
    The inequality $f(\Low{\Dth{p}}) > f(p)$ implies that $\Dth{p} = q$ is an ancestor of $a$ for which $a$ does not have the smallest value in the subtree of $q$.
    Condition~{\sf II} requires that there is a path $\Path{a}{q}$ and this is a contradiction as $q \neq \Dth{a}$.
    It follows that $f(p) > f(\Low{\Dth{p}})$ for every minimum $p$ and that $\mathcal{B}$ satisfies \cref{inv:condition-II}.

    \smallskip
    It remains to show that the banana tree $\mathcal{B}$ also satisfies \cref{inv:condition-I}.
    Let $q$ be some maximum except the special root, let $Q$ be the path ending at $q$ and let $p$ be the leaf at the other end of $Q$.
    Assume $p < q$; the other case is symmetric.
    The path $Q$ is contained entirely in a subtree $S_1$ of $q$ and the construction of $\mathcal{B}$ is such that this subtree becomes the banana subtree rooted at $q$.
    By Condition {\sf I.1} either for all $s \in S_1 \colon s < q$ or for all $s \in S_1 \colon s > q$.
    Since $p \in S_1$, it follows that all $s \in S_1$ satisfy $s < q$. 
    The subtree $S_1$ becomes the banana subtree rooted at $q$ in $\mathcal{B}$
    and thus the condition that $u < q$ for all $u$ in the banana subtree rooted at $q$ is satisfied.
    Let $S_2$ be the other subtree of $q$ in $T$.
    By Condition {\sf I.1} for all $t \in S_2$ we have $t > q$.
    The node $\Dn{q}$ in $\mathcal{B}$ is one of the nodes in $S_2$ and thus $\Dn{q} > q$.
    By Conditions {\sf III.1} and {\sf III.2} all descendants $v$ of $\Dn{q}$ on the same trail as $\Dn{q}$ satisfy $v > \Dn{q}$ and thus $v > q$.
    All these nodes $v$ are also nodes in $S_2$, since they are on the same path as $\Dn{q}$ and have smaller value.
    The banana subtree of each $v$ is obtained from the subtree of $v$ in $T$ and thus these banana subtrees consist of nodes in subtrees of $S_2$.
    It follows that every descendant $t$ of $\Dn{q}$ including $\Dn{q}$ satisfies $t > q$, as required by \cref{inv:condition-I}.
    This concludes the proof that $\mathcal{B}$ satisfies \cref{inv:condition-I}.
\end{proof}

We now give an algorithm that turns a banana tree into path-decomposed binary tree.
\begin{lemma}[Merging Trails]
    \label{lem:banana-tree-to-pdbt}
    There exists a deterministic algorithm to merge the trails of each banana in a banana tree satisfying \cref{inv:condition-I,inv:condition-II,inv:condition-III-max},
    such that the resulting path-decomposed binary tree satisfies Conditions {\sf I.1}, {\sf I.2} and {\sf II}.
\end{lemma}
\begin{proof}
    Let $\mathcal{B}$ be a banana tree satisfying \cref{inv:condition-I,inv:condition-II,inv:condition-III-max}.
    Define the height of a banana subtree to be the longest simple path from the root of the banana subtree to a leaf following only $\In{\cdot}$, $\Mid{\cdot}$ and $\Dn{\cdot}$ pointers and with function value decreasing along the path.
    We show by induction over the height
    that any banana subtree of $\mathcal{B}$ satisfying \cref{inv:condition-I,inv:condition-II,inv:condition-III-max}
    can be merged into a path-decomposed binary tree satisfying Conditions {\sf I.1}, {\sf I.2} and {\sf II}.
    Then, since the banana subtree rooted at the special root is equivalent to the banana tree $\mathcal{B}$,
    it follows that there is a deterministic algorithm to obtain a path-decomposed binary tree satisfying Conditions {\sf I.1}, {\sf I.2} and {\sf II} from $\mathcal{B}$.

    \smallskip
    We now give the algorithm for merging banana trees into path-decomposed binary trees.
    By induction over the height $h$ we show that a banana subtree of $\mathcal{B}$ of height $h$ with root $q$
    can be merged into a path-decomposed binary tree $T(q)$ such that Conditions {\sf I.1}, {\sf I.2} and {\sf II} are satisfied.
    To achieve this we also show that $\Birth{q} < q$ implies that $q$ has no right child in $T(q)$ and $\Birth{q} > q$ implies that $q$ has no left child in $T(q)$.
    This is needed to ensure that after assembling the subtrees of smaller height on a path, the nodes in the resulting binary tree are ordered correctly.
    Note that a banana tree consists of at least two nodes spanning a banana with empty trails, i.e., trails without interior nodes, so the base case is $h = 1$.

    \medskip
    \noindent\textbf{Base case:} A banana tree of height $1$ consists of a single banana spanned by a minimum $p$ and a maximum $q$ with $f(p) < f(q)$.
    We obtain a path-decomposed binary tree as follows: $q$ becomes the tree with $p$ as left child of $q$ if $p<q$ and right child otherwise.
    There is a single path between $p$ and $q$.
    This tree clearly satisfies Conditions {\sf I.1}, {\sf I.2} and {\sf II}.
    Furthermore, since $\Birth{q} = p$ if $p < q$, then $q$ has no right child and if $p > q$ then $q$ has no left child, so the claim holds.

    \medskip
    \noindent\textbf{Induction hypothesis:}
    For some $h \geq 1$, for all $1 \leq j \leq h$ a banana subtree of $\mathcal{B}$ of height $j$ can be merged into a path-decomposed binary tree $T$ satisfying Conditions {\sf I.1}, {\sf I.2} and {\sf II}.
    Furthermore, $\Birth{q} < q$ implies that $q$ has no right child in $T$ and $\Birth{q} > q$ implies that $q$ has no left child in $T$.

    \medskip
    \noindent\textbf{Induction step:} Assume the induction hypothesis for some $h \geq 1$.
    We show that the claim holds for banana subtrees of height $h+1$.
    Let $q$ be the root of a banana subtree of $\mathcal{B}$ of height $h+1$ and $p = \Birth{q}$.
    Banana subtrees rooted at a node on the banana spanned by $p$ and $q$ have height at most $h$,
    and by the induction hypothesis these banana subtrees can be merged into path-decomposed binary trees satisfying Conditions {\sf I.1}, {\sf I.2} and {\sf II}.
    Write $T(u)$ for the path-decomposed binary tree obtained from the banana subtree rooted at node $u$.

    Let $s_0 = p, s_1,\dots,s_\ell = q$ be the nodes on the trails between $p$ and $q$ ordered by function value.
    This ordering is unique since function values are distinct.
    Consider a node $s_i$ for $1 \leq i \leq \ell-1$. We show that for all $0 \leq j \leq i-1$, if $\Birth{s_i} < s_i$ then $s_i < s_j$ and $s_i < u$ for any descendant $u$ of $s_j$,
    and if $\Birth{s_i} > s_i$ then $s_i > s_j$ and $s_i > u$ for any descendant $u$ of $s_j$.
    Assume $\Birth{s_i} < s_j$; the other case is symmetric.
    By \cref{inv:condition-I} we have $s_i < \Dn{s_i}$ and for all descendants $u$ of $\Dn{s_i}$ it holds that $s_i < u$.
    Note that the nodes $s_j$ for $0 \leq i-1$ that are on the same trail as $s_i$ are descendants of $\Dn{s_i}$ or $\Dn{s_i}$.
    By Conditions~{\sf III.1} and {\sf III.2} the nodes $s_k$ for $k \leq i-1$ that are on the other trail than $s_i$ satisfy $s_i < s_k$ and $\Dn{s_k} < s_j$.
    For these nodes $s_k$ \cref{inv:condition-I} implies $s_k < \Birth{s_k}$ and all descendants $v$ of $s_k$ satisfy $s_k < v$.
    This proves that for all $0 \leq j \leq i-1$ both $s_i < s_j$ and $s_i < u$ for any descendant of $u$.
    
    We iteratively assemble the path-decomposed binary trees rooted at the $s_i$ for $1 \leq i \leq \ell-1$ into a path-decomposed binary tree rooted at $q$ by linking $T(s_i)$ to the tree assembled from the $T(s_j)$ for $j \leq i - 1$.
    
    We begin with $i = 1$.
    Note that $\Dn{s_1} = s_0$, since $s_1$ must be the bottom-most node of the in- or mid-trail between $p$ and $b$.
    By \cref{inv:condition-I}, if $\Birth{s_1} < s_1$ then $\Dn{s_1} > s_1$ and if $\Birth{s_1} > s_1$ then $\Dn{s_1} < s_1$.
    Assume $\Birth{s_1} < s_1$; the other case is symmetric.
    By the induction hypothesis the root $s_1$ of the tree $T(s_1)$ has no right child in this tree.
    We make $s_0 = \Dn{s_1}$ the right child of $s_1$ and since $s_0 = \Dn{s_1} > s_1$ this ensures that the resulting tree satisfies Condition {\sf I.1}.
    By \cref{inv:condition-III-max}, $f(\Dn{s_1}) < f(s_1)$, so Condition {\sf I.2} is also satisfied.
    
    Now consider any $i$ with $1 \leq i \leq \ell-1$.
    In iteration $i-1$ we have already linked the $T(s_j)$ for $1 \leq j \leq i$ into a tree $T_{i-1}$ satisfying Conditions {\sf I.1} and {\sf I.2}.
    The process to combine $T_{i-1}$ with $T(s_i)$ is similar to that for $i = 1$.
    By \cref{inv:condition-I} we have either $\Birth{s_i} < s_i < \Dn{s_i}$ or $\Birth{s_i} > s_i > \Dn{s_i}$.
    Assume again that $\Birth{s_i} < s_i < \Dn{s_i}$; the other case is symmetric.
    The nodes in $T_{i-1}$ are exactly the nodes $s_j$ for $j \leq i-1$ and their descendants,
    so as discussed above, for all $t \in T_{i-1}$ it holds that $\Birth{s_i} < s_i < t$.
    By the induction hypothesis $s_i$ has no right child in $T(s_i)$.
    We can thus attach $T_{i-1}$ as the right subtree of $s_i$ in $T(s_i)$ and the resulting tree $T_i$ satisfies Condition {\sf I.1}.
    The tree $T_i$ also satisfies Condition {\sf I.2}, as $T_{i-1}$ satisfies Condition {\sf I.2} and $f(s_i) > f(s_{i-1})$.

    The algorithm finishes by making $T_{\ell-1}$ the left subtree of $q$ if $p < q$ and the right subtree otherwise,
    which yields the tree $T(q)$.
    Furthermore, we make $p = s_0,s_1,\dots,s_\ell=q$ a path in $T(q)$.
    Since all nodes $u$ in $T_{\ell-1}$ are in the banana subtree rooted at $q$ they satisfy $u < q$ if $p < q$ and $u > q$ if $p > q$.
    Thus $T(q)$ satisfies Condition {\sf I.1}.
    Condition {\sf I.2} is also satisfied since $T_{\ell-1}$ satisfies this condition and $f(q) > f(s_{\ell-1})$.
    Note also that $T(q)$ has no right child if $p := \Birth{q} < q$ and no left child if $p :=\Birth{q} > q$,
    since the banana subtree rooted at $q$ satisfies \cref{inv:condition-I}.

    It remains to show that $T(q)$ satisfies Condition {\sf II}.
    There can be no minimum $a \neq p$ in the banana subtree rooted at $q$ with $f(a) < f(p)$,
    as otherwise there would exist a minimum $c$ in the banana subtree rooted at $q$ violating $f(\Low{\Dth{c}} < f(c)$ required by \cref{inv:condition-II}.
    This implies that $p$ has the smallest value in the banana subtree rooted at $q$.
    It is connected to the root $q$ by a path, as required by Condition {\sf II}.
    Since all $T(s_i)$ for $1 \leq i \leq \ell-1$ satisfy Condition {\sf II} by the induction hypothesis
    we now only need to show that the path ending at each $s_i$ satisfies Condition {\sf II}.
    Write $t_i$ for the other end of the path ending at $s_i$.
    We need to show that $s_i$ is the nearest ancestor to $t_i$ such that $s_i$ has a descendant with smaller value than $t_i$ in $T(q)$.
    Since $f(p) < f(t_i)$, as discussed above, $s_i$ is indeed such an ancestor, and since $t_i$ is connected to the root of $T(s_i)$ by a path
    there exists no other such ancestor in $T(s_i)$.
    Thus, Condition {\sf II} holds.
    This concludes the induction and the proof of the lemma.
\end{proof}

Finally, we state our result on uniqueness of banana trees.
\begin{corollary}[Invariants Determine the Banana Tree]
    \label{cor:unique-bananas-local}
    Given a linear list of distinct values there is a unique banana tree satisfying \cref{inv:condition-I,inv:condition-II,inv:condition-III-max}.
\end{corollary}
\begin{proof}
    By \cref{lem:uniqueness_of_banana_tree} there exists a unique path-decomposed binary tree for the linear list that satisfies Conditions {\sf I.1}, {\sf I.2} and {\sf II}.
    By \cref{lem:pdbt-to-banana-tree,lem:banana-tree-to-pdbt} it follows that there is also a unique banana tree that satisfies \cref{inv:condition-I,inv:condition-II,inv:condition-III-max}.
\end{proof}

\subsection{Slides.}
A \emph{slide} occurs when an internal critical item becomes non-critical and its non-critical neighbor becomes critical.
We distinguish two cases, based on whether the critical item is a maximum or a minimum.
\smallskip \begin{description}
    \item[Max Slide:] Let $p$ be a non-critical item, $q$ a neighbor of $p$ and a maximum.
        The value of $p$ increases above the value of $q$ or the value of $q$ decreases below the value of $p$.
    \item[Min Slide:] Let $p$ be a non-critical item, $q$ a neighbor of $p$ and a minimum.
        The value of $p$ decreases below the value of $q$ or the value of $q$ increases above the value of $p$.
\end{description} \smallskip
In both cases, the number of critical items remains unchanged.
If the order of critical items by function value is unaffected by the slide apart from $q$ being replaced by $p$,
then the banana tree can be updated by replacing the formerly critical item with the new critical item in the tree,
which changes the label of the node associated with $q$ to refer to $p$ instead.
We now prove that this correctly maintains the up-tree if the slide is caused by a sufficiently small change in value.

\begin{lemma}[Max-Slide]
    \label{lem:max-slide}
    Let $p$ and $q$ be neighboring items.
    Let $\sigma \in [0,1)$ be such that in $h_\sigma$ the item $p$ is non-critical, $q$ is a maximum and $p$, $q$ are contiguous.
    Let $\theta \in (\sigma, 1]$ be such that in $h_\theta$ the item $p$ is a maximum, $q$ is non-critical and $p$, $q$ are contiguous.
    The tree $\UpTree{h_\theta}$ is obtained from $\UpTree{h_\sigma}$ by replacing $q$ with $p$.
\end{lemma}
\begin{proof}
    By \cref{lem:stable-value-order}, the order of maxima by function value is the same in $h_\sigma$ and $h_\theta$, with $q$ replaced by $p$.
    The maps $h_\lambda$ are defined such that they differ in exactly one fixed item $j$, so either $j=p$ or $j=q$.
    Neither $p$ nor $q$ is a minimum in $h_\sigma$ or $h_\theta$, so all minima have equal value in $h_\sigma$ and $h_\theta$.
    It then follows that the minimum $s$ paired with $q$ in $h_\sigma$ is paired with $p$ in $h_\theta$.
    That is, if $\Window{s}{q}$ is a window in $h_\sigma$, then $\Window{s}{p}$ is a window in $h_\theta$.
    Furthermore, $\Window{s}{p}$ is nested into the same window in $h_\theta$ as $\Window{s}{q}$ is in $h_\sigma$.
    No other window is affected by the change in value, since the order of critical items by function value is the same other than $q$ being replaced by $p$.
    By \cref{lem:bananas_and_windows}, $s$ and $q$ span a banana in $h_\sigma$ and $s$ and $p$ span a banana in $h_\theta$.
    Thus, replacing $q$ with $p$ in the up-tree for $h_\sigma$ yields the up-tree for $h_\theta$, as claimed.
\end{proof}

\begin{lemma}[Min-Slide]
    \label{lem:min-slide}
    Let $p$ and $q$ be neighboring items.
    Let $\sigma \in [0,1)$ be such that in $h_\sigma$ the item $p$ is non-critical, $q$ is a minimum and $p$ and $q$ are contiguous.
    Let $\theta \in (\sigma, 1]$ be such that in $h_\theta$ the item $p$ is a minimum, $q$ is non-critical and $p$ and $q$ are contiguous.
    The tree $\UpTree{h_\theta}$ is obtained from $\UpTree{h_\sigma}$ by replacing $q$ with $p$.
\end{lemma}
\begin{proof}
    By \cref{lem:stable-value-order}, the order of minima by function value is the same in $h_\sigma$ and $h_\theta$, with $q$ replaced by $p$.
    Since the values of maxima are equal in $h_\sigma$ and $h_\theta$ it then follows that the maximum $s$ paired with $q$ in $h_\sigma$ is paired with $p$ in $h_\theta$.
    That is, if $\Window{q}{s}$ is a window in $h_\sigma$, then $\Window{p}{s}$ is a window in $h_\theta$.
    Furthermore, $\Window{p}{s}$ is nested into the same window in $h_\theta$ as $\Window{s}{q}$ is in $h_\sigma$.
    No other window is affected by the change in value, since the order of critical items by function value is the same other than $q$ being replaced by $p$.
    By \cref{lem:bananas_and_windows}, $q$, $s$ span a banana in $h_\sigma$ and $p$, $s$ span a banana in $h_\theta$.
    Thus, replacing $q$ with $p$ in the up-tree for $h_\sigma$ yields the up-tree for $h_\theta$, as claimed.
\end{proof}

\subsection{Cancellations.}
Let $p$ be a minimum, $q$ be a maximum and a neighbor of $p$, such that there is a window $\Window{p}{q}$ without any windows nested into it,
i.e., such that $p$ and $q$ span a banana in the banana tree whose trails contain no other critical items.
A cancellation occurs when the value of $p$ increases above the value of $q$ or the value of $q$ decreases below the value of $p$.
Then, both $p$ and $q$ become non-critical and the window spanned by $p$ and $q$ disappears.
To update the up-tree we remove the banana spanned by $p$ and $q$ by connecting $\Up{q}$ to $\Dn{q}$ and discarding $p$ and $q$.

\begin{lemma}[Cancellation]
    \label{lem:cancellation}
    Let $p$, $q$ be neighboring items.
    Let $\sigma \in [0,1)$ be such that in $h_\sigma$ the item $p$ is a minimum, $q$ is a maximum and $p$ and $q$ are contiguous and paired in  $h_\sigma$.
    Let $\theta \in (\sigma, 1]$ be such that in $h_\theta$ the items $p$ and $q$ are non-critical.
    Then nodes $p$ and $q$ span a banana in $\UpTree{h_\sigma}$ that does not contain any other critical items and removing this banana yields $\UpTree{h_\theta}$.
\end{lemma}
\begin{proof}
    Let $\gamma \in (\sigma, \theta]$ be such that in $h_\gamma$ the items $p$ and $q$ are non-critical and contiguous.
    By \cref{lem:stable-value-order}, the order of items by function value is the same in $h_\sigma$ and $h_\gamma$, except that $p$ and $q$ are swapped.
    In particular, the order of critical items by function value is the same in $h_\sigma$ and $h_\gamma$.
    Since $p$, $q$ are non-critical in $h_\gamma$ and in $h_\theta$ it follows that the order of critical items remains the same. 
    Thus, the order of critical items in $h_\sigma$ and $h_\theta$ is identical.

    Since $p$ and $q$ are contiguous in $h_\sigma$, there is no critical item with value between $h_\sigma(p)$ and $h_\sigma(q)$.
    As $p$ and $q$ are neighbors this implies that $p$ and $q$ are paired in $h_\sigma$:
    consider the process to obtain the pairing described in \cref{sec:2.2}; the component born at $p$ in $h_\sigma$ dies at $q$, since the component on the other side of $q$ must be born at a lower value.
    Furthermore, there are no windows nested into $\Window{p}{q}$, since there is no critical item with value in $[h_\sigma(p), h_\sigma(q)]$.
    All other windows in $h_\sigma$ are also windows in $h_{\theta}$, since the critical items have the same order by function value.

    As $\Window{p}{q}$ is a window of $h_\sigma$, by \cref{lem:bananas_and_windows} the nodes $p$ and $q$ span a banana in $\UpTree{h_\sigma}$.
    There are no windows nested into $\Window{p}{q}$ and thus the trails between $p$ and $q$ are empty.
    Since $p$ and $q$ are non-critical in $h_\theta$ they are not present in $\UpTree{h_\theta}$.
    Thus, removing the banana spanned by $p$ and $q$ from $\UpTree{h_\sigma}$ yields $\UpTree{h_\theta}$, as claimed.
\end{proof}

\subsection{Slides and Cancellations Involving Endpoints.}
\label{sec:slide-cancel-endpoint}
When the criticality of the neighbor of an endpoint changes it can also change the endpoint from up-type to down-type or vice versa.
Likewise, a change in the criticality of an endpoint leads to a change in the criticality of its neighbor.
The operations needed to maintain the up-tree in these cases are similar to cancellations and slides.
Let $q \in \{1,m\}$ be an endpoint and $a$ be its neighbor.
In the following we assume $a < q$, i.e., that $q$ is the right endpoint of the interval. The case $q < a$ is symmetric.
Let $0 \leq \sigma < \theta \leq 1$ such that
$h_\sigma(q) > h_\sigma(a),\ h_\theta(q) < h_\theta(a)$ or $h_\sigma(q) < h_\sigma(a),\ h_\theta(q) > h_\theta(a)$.
There are four cases:
\smallskip \begin{enumerate}
    \item $h_\sigma(q) > h_\sigma(a)$, $h_\theta(q) < h_\theta(a)$, $a$ is a minimum in $h_\sigma$ and non-critical in $h_\theta$,
    \item $h_\sigma(q) > h_\sigma(a)$, $h_\theta(q) < h_\theta(a)$, $a$ is non-critical in $h_\sigma$ and a maximum in $h_\theta$,
    \item $h_\sigma(q) < h_\sigma(a)$, $h_\theta(q) > h_\theta(a)$, $a$ is a maximum in $h_\sigma$ and non-critical in $h_\theta$,
    \item $h_\sigma(q) < h_\sigma(a)$, $h_\theta(q) > h_\theta(a)$, $a$ is non-critical in $h_\sigma$ and a minimum in $h_\theta$.
\end{enumerate} \smallskip
We now give the algorithm for the case $a < q$.
\begin{tabbing}
    m\=m\=m\=m\=m\=m\=m\=m\=m\=m\=m\=\kill
    \clreset
    \clbegin                 \> \algcancelendpoint $(a,q)$: \\*
    \clbegin                 \> \> \texttt{if $q$} changes from down-type to up-type \texttt{then} \\*
    \clbegin                 \> \> \> \texttt{if} $a$ becomes non-critical \texttt{then} \\*
    \clbeginlbl{ln:e-case-1} \> \> \> \> remove $q$ and $\Birth{q}$; replace $a$ by $q$ \\*
    \clbegin                 \> \> \> \texttt{else if} $a$ becomes a maximum \texttt{then} \\*
    \clbeginlbl{ln:e-case-2} \> \> \> \> let $p = \Birth{q}$; replace $q$ by $a$; replace $p$ by $q$ \\*
    \clbegin                 \> \> \> \texttt{endif} \\
    \clbegin                 \> \> \texttt{else if} $q$ changes from up-type to down-type \texttt{then} \\*
    \clbegin                 \> \> \> \texttt{if} $a$ becomes non-critical \texttt{then} \\*
    \clbeginlbl{ln:e-case-3} \> \> \> \> replace $q$ by the hook; replace $a$ by $q$ \\*
    \clbegin                 \> \> \> \texttt{else if} $a$ becomes a minimum \texttt{then} \\*
    \clbeginlbl{ln:e-case-4} \> \> \> \> replace $q$ by $a$; insert $q$ as $\In{a}$; add a banana between $q$ and its hook \\*
    \clbegin                 \> \> \> \texttt{endif}; \\*
    \clbegin                 \> \> \texttt{endif}. \\*
\end{tabbing}

\begin{lemma}[Changes Involving Endpoints]
    \label{lem:slide-cancel-endpoint}
    Let $q$ be an endpoint, $a$ be its neighbor and let $\check{q}$ be the hook neighboring $q$.
    Let $0 \leq \sigma < \theta \leq 1$ such that
    \begin{enumerate}
        \item either $h_\sigma(q) > h_\sigma(\check{q}) > h_\sigma(a)$, $h_\theta(q) < h_\theta(\check{q}) < h_\theta(a)$,
        \item or $h_\sigma(q) < h_\sigma(\check{q}) < h_\sigma(a)$, $h_\theta(q) > h_\theta(\check{q}) > h_\theta(a)$,
    \end{enumerate}
    where $q$ and $\check{q}$ are contiguous and $\check{q}$ and $a$ are contiguous in $h_\sigma$ and $h_\theta$.
    Applying the algorithm \algcancelendpoint$(a,q)$ to the up-tree $\UpTree{h_\sigma}$ yields the up-tree $\UpTree{h_\theta}$.
\end{lemma}
\begin{proof}
    Assume $a < q$; the other case is symmetric.
    We show how the change in function value from $h_\sigma$ to $h_\theta$ affects the windows in each of the cases
    and that the algorithm correctly updates the bananas to match the windows of $h_\theta$.
    Recall that the maps $h_\lambda$ are defined such that they differ in a single fixed item.
    Thus, we can assume that either $h_\theta(q) = h_\sigma(q)$ or $h_\theta(a) = h_\sigma(a)$
    and for all $u \neq q,a$ it holds that $h_\theta(u) = h_\sigma(u)$.
    \smallskip \begin{description}
        \item[$h_\sigma(q) > h_\sigma(a)$, $h_\theta(q) < h_\theta(a)$, $a$ is non-critical in $h_\theta$ (\cref{ln:e-case-1}):]
            There is a component in $h_\sigma$ born at $\check{q}$ that dies at $q$ and a component born at $a$ that dies at an item $b$.
            The component born at $\check{q}$ merges into the component born at $a$.
            These components correspond to two windows $\Window{\check{q}}{q}$ and $\Window{a}{b}$, where the former is nested into the latter.
            In $h_\theta$ there are no components born at $\check{q}$ or $a$, as they are no longer homological critical points.
            Instead, there is a component born at $q$, which corresponds to a window $\Window{q}{\cdot}$.
            As $q$, $\check{q}$ and $\check{q}$, $a$ are contiguous in $h_\sigma$ and $h_\theta$
            the order of items by function value in $h_\sigma$ and $h_\theta$ differs only in that $q$, $\check{q}$ and $a$ are reversed by \cref{lem:stable-value-order}.
            Since $q$ neighbors $a$ it follows that the component born at $q$ in $h_\theta$ dies at $b$, i.e., there is a window $\Window{q}{b}$ in $h_\theta$.
            The algorithm reflects this change in windows:
            the banana spanned by $\check{q}$ and $q$ corresponding to $\Window{\check{q}}{q}$ is removed;
            the banana spanned by $a$ and $b$ corresponding to $\Window{a}{b}$ is replaced by a banana spanned by $q$ and $b$ corresponding to $\Window{q}{b}$.
        \item[$h_\sigma(q) > h_\sigma(a)$, $h_\theta(q) < h_\theta(a)$, $a$ is a maximum in $h_\theta$ (\cref{ln:e-case-2}):]
            Going from $h_\sigma$ to $h_\theta$ can be seen as
            relabeling $q$ such that it represents item $a$, relabeling $\check{q}$ such that it represents $q$ and removing the item formerly associated with $a$.
            Changing the function values to reflect the values in $h_\theta$ does not affect the ordering of items, due to contiguity of $q$, $\check{q}$ and $\check{q}$, $a$.
            The corresponding change in the up-tree is to relabel the node representing $q$ to represent item $a$,
            and relabeling the node representing the hook to represent item $q$,
            as the algorithm does in \cref{ln:e-case-2}.
        \item[$h_\sigma(q) < h_\sigma(a)$, $h_\theta(q) > h_\theta(a)$, $a$ is non-critical in $h_\theta$ (\cref{ln:e-case-3}):]
            Similarly to the previous case the change from $h_\sigma$ to $h_\theta$ can be seen as
            insertion of an item $a^* < a$ next to $a$ and contiguous with $q$ in $h_\sigma$,
            then relabeling $q$ to represent $\check{q}$, $a$ to represent $q$ and $a^*$ to represent $a$.
            Adjusting the values of the items to reflect the values in $h_\theta$ does not affect the ordering of items, due to contiguity of $q$, $\check{q}$ and $\check{q}$, $a$.
            The corresponding change in the up-tree is to relabel the node representing $q$ to represent $\check{q}$,
            and relabel the node representing $a$ to represent $q$,
            as the algorithm does in \cref{ln:e-case-3}.
        \item[$h_\sigma(q) < h_\sigma(a)$, $h_\theta(q) > h_\theta(a)$, $a$ is a minimum in $h_\theta$ (\cref{ln:e-case-4}):]
            This is the reverse of the first case: there is a component born at $q$ in $h_\theta$ and it is replaced by a component born at $a$ and another component born at $\check{q}$ in $h_\theta$.
            Reverse to the case above the window $\Window{q}{b}$ is replaced by the window $\Window{a}{b}$ and a new window $\Window{\check{q}}{q}$ appears.
            The change from $\Window{q}{b}$ to $\Window{a}{b}$ is reflected by replacing $q$ by $a$ in the up-tree.
            The new window $\Window{\check{q}}{q}$ must be nested into $\Window{a}{b}$ as $h_\theta(a) < h_\theta(q) < h_\theta(b)$,
            since $a$, $\check{q}$ and $q$ are contiguous in $h_\theta$.
            Furthermore, $\Window{\check{q}}{q}$ must be nested into the in-panel, as $b < a$.
            The nesting of $\Window{\check{q}}{q}$ into the in-panel of $\Window{a}{b}$ is reflected in the algorithm by adding a banana spanned by $\check{q}$ and $q$ into the in-trail of the banana spanned by $a$ and $b$.
    \end{description} \smallskip
    The algorithm updates the bananas to reflect the windows and their nesting hierarchy in $h_\theta$, which implies by \cref{lem:bananas_and_windows} that the new up-tree is $\UpTree{h_\theta}$.
\end{proof}

\subsection{Anti-cancellations.}
Let $p$ and $q$ be neighboring non-critical items.
If the value of one changes such that both become critical, a new window is introduced.
This is the reverse of a cancellation and we call it an \emph{anti-cancellation}.

Denote by $h_\sigma$ and $h_\theta$ the map before and after the anti-cancellation, respectively, and assume that in both $h_\sigma$ and $h_\theta$ the items $p$ and $q$ are contiguous.
Furthermore, assume that $p$ becomes a minimum and $q$ a maximum.
To change $\UpTree{h_\sigma}$ into $\UpTree{h_\theta}$ we need to introduce a banana spanned by $p$ and $q$.
We recall the algorithm for the case $p < q$; the case $p > q$ is symmetric. The algorithm first identifies the maximum $b$ closest to $p$ such that $b < p < q$, i.e., such that $p$ becomes the critical item next to $b$.
It then walks down the path from $b$ to its successor until it reaches the first node $t$ such that $h_\theta(t) < h_\theta(q)$.
The node $q$ is then inserted into this path above $t$, and a new banana spanned by $p$ and $q$ is attached at $q$.
To insert a node $q$ between some nodes $a$ and $b$, where $a$ and $b$ are neighbors on the same trail and $f(a) < f(b)$,
we set $\Up{q} = b$, $\Dn{q} = a$ and make the pointer from $a$ to $b$ and the pointer from $b$ to $a$ point to $q$ instead.

\begin{tabbing}
    m\=m\=m\=m\=m\=m\=m\=m\=m\=m\=m\=\kill
    \clreset
    \clbegin \> \alganticancel $(p,q)$: \\*
    \clbegin \> \> \> find the maximum $b$ closest to $p$ such that $b < p < q$ \\*
    \clbegin \> \> \> \texttt{if} $b < \Birth{b}$ \texttt{then} $t = \Mid{b}$ \texttt{else} $t = \Dn{b}$ \texttt{endif}; \\*
    \clbegin \> \> \> \texttt{while} $h_\theta(t) > h_\theta(q)$ \texttt{do} $t = \In{t}$ \texttt{endwhile}; \\*
    \clbegin \> \> \> \texttt{if} $t$ is a leaf \= \texttt{then} \= \texttt{if} $t = \Birth{b}$ \= \texttt{then} insert $q$ between $t$ and $\Mid{t}$ \\*
    \clbegin \> \> \>                           \>               \>                             \> \texttt{else if} \= $t$ is the global minimum \texttt{and} $t < q$ \texttt{then} \\*
    \clbegin \> \> \>                           \>               \>                             \>                  \> insert $q$ between $t$ and $\Mid{t}$ \\*
    \clbegin \> \> \>                           \>               \>                             \> \texttt{else} insert $q$ between $t$ and $\In{t}$ \\*
    \clbegin \> \> \>                           \>               \> \texttt{endif}; \\*
    \clbegin \> \> \>                           \> \texttt{else} insert $q$ between $t$ and $\Up{t}$ into the trail of $t$ \\*
    \clbegin \> \> \> \texttt{endif}. \\*
    \clbegin \> \> \> Set $\In{p} = \Mid{p} = q$, $\In{q} = \Mid{q} = p$, $\Dth{p} = q$, $\Low{q} = \Low{\Dn{q}}$.
\end{tabbing}

\begin{lemma}[Anti-Cancellation]
    \label{lem:anticancel}
    Let $p$ and $q$ be neighboring items.
    Let $\sigma \in [0,1)$ such that in $h_\sigma$ both $p$ and $q$ are non-critical, $h_\sigma(p) > h_\sigma(q)$, and $p$ and $q$ are contiguous.
    Let $\theta \in (\sigma, 1]$ such that in $h_\theta$ the item $p$ is a minimum, $q$ is a maximum, and $p$ and $q$ are contiguous.
    Algorithm \alganticancel$(p,q)$ applied to the up-tree $\UpTree{h_\sigma}$ yields the up-tree $\UpTree{h_\theta}$.
\end{lemma}
\begin{proof}
    We prove the lemma for the case where $p < q$; the other case is symmetric.
    The tree $\UpTree{h_\sigma}$ satisfies \cref{inv:condition-I,inv:condition-II,inv:condition-III-max}.
    We show that the tree obtained by \alganticancel also satisfies these invariants for $h_\theta$,
    which by \cref{cor:unique-bananas-local} implies that it is the unique up-tree for $h_\theta$.

    The algorithm first finds the maximum $b$ closest to $p$ such that $b < p < q$.
    Since $p$ and $q$ are non-critical in $h_\sigma$, the item $p$ cannot be an endpoint.
    Furthermore, $h_\sigma(p) > h_\sigma(q)$ and thus this maximum $b$ exists.
    After the while-loop terminates $t$ satisfies $h_\theta(t) \leq h_\theta(q)$.
    Since $p$ and $q$ are contiguous in $h_\theta$ it holds that $h_\theta(t) \neq h_\theta(q)$
    and thus $h_\theta(t) < h_\theta(q)$ and $h_\theta(t) < h_\theta(p)$.
    
    The node $q$ is inserted above $t$ such that $\Dn{q} = t$.
    As $\Low{q}$ is set to $\Low{\Dn{q}} = \Low{t}$, we have $h_\theta(\Low{q}) < h_\theta(p)$ 
    and it follows that $h_\theta(\Low{q}) = h_\theta(\Low{\Dth{p}}) < h_\theta(p)$.
    Since the $\Low{\cdot}$ and $\Dth{\cdot}$ pointers of all other nodes are unchanged between $\UpTree{h_\sigma}$ and the new up-tree,
    the new up-tree satisfies \cref{inv:condition-II}.

    We now show that \cref{inv:condition-I} holds in the new up-tree.
    Let $t_0$ be the node $t$ as initialized in the first if-statement and $t_i$ be the node $t$ after the $i$-th iteration of the loop.
    Denote by $I$ the last iteration, such that $t_I$ is the node $t$ after the loop terminates.
    Note that all $t_i$ for $i \geq 0$ are present in $\UpTree{h_\sigma}$.
    As \cref{inv:condition-I} holds for $\UpTree{h_\sigma}$, the initialization of $t$ guarantees $b < t_0$.
    Furthermore, all $t_i$ for $i \geq 1$ are descendants of $t_0$ and thus $b < t_i$ for all $i \geq 0$.
    By definition of $b$ and since $p$ and $q$ are neighbors it also holds that $q < t_i$ for all $i \geq 0$
    and thus $b < p < q < t_i$ for all $i \geq 0$.

    To continue the proof of \cref{inv:condition-I} we need the following claim:
    \begin{claim}
        \label{clm:anticancel-t-order}
        The while-loop terminates.
        For all $0 \leq i \leq I-1$ it holds that $t_i$ is not a leaf, $\Birth{t_i} < t_i$ and $t_{i+1} < t_i$.
    \end{claim}
    \begin{proof}
        There exists a minimum $a$ such that $b < p < q < a$ and such that $h_\sigma(a) < h_\sigma(q)$ and such that there are no critical items $u$ with $q < u < a$.
        This follows from a similar argument as for the existence of $b$:
        $q$ is non-critical in $h_\sigma$ and thus cannot be an endpoint and $h_\sigma(q) < h_\sigma(p)$, so the first critical item greater than $q$ must be a minimum.
        Note that $a$ must be either $t_0$ or a descendant of $t_0$ and that in the latter case it is the leftmost descendant of $t_0$ by \cref{inv:condition-I}.
        As the algorithm modifies the up-tree after the while-loop terminated, we only argue about $\UpTree{h_\sigma}$ in this proof.
        Recall that this up-tree satisfies \cref{inv:condition-I,inv:condition-II,inv:condition-III-max}
        and that by \cref{lem:inv-1-and-3-are-cond-3} it also satisfies Conditions {\sf III.1} and {\sf III.2}.

        By \cref{inv:condition-I} $b = \Up{t_0} < t_0$ and by \cref{inv:condition-I,inv:condition-III-max} $\In{t_0} < t_0 < \Dn{t_0}$, unless $t_0$ is a leaf.
        The while-loop assigns $t_1 = \In{t_0}$ and by Conditions {\sf III.1}, {\sf III.2} and \cref{inv:condition-I}
        this is the leftmost node among the other nodes on the banana with maximum and their descendants.
        \Cref{inv:condition-I,inv:condition-III-max} again $\In{t_1} < t_1$, unless $t_1$ is a leaf.
        This argument can be repeated for $\In{t_1}, \In{\In{t_1}}, \dots$ until we reach a leaf.
        Thus, by following $\In{\cdot}$-pointers from $t_0$ we reach the leftmost leaf that is a descendant of $t_0$, which is $a$.
        As $h_\theta(a) = h_\sigma(a) < h_\sigma(q) \leq h_\theta(q)$, there exists a node $\InOp^*(t_0)$ with $h_\theta(\InOp^*(t_0)) < h_\theta(q)$ and hence the while-loop terminates.
        If the loop terminates in a leaf, then $t_I$ is a leaf and all $t_i$ with $0 \leq i \leq I-1$ are not leaves.
        If the loop terminates at an internal node, then all $t_i$ for $0 \leq i \leq I$ are not leaves.
        We have also seen that $\In{t_i} < t_i$ for all $0 \leq i \leq I$ and this implies $\Birth{t_i} < t_i$ by \cref{inv:condition-I}.
        This proves the claim.
    \end{proof}

    We resume the proof of \cref{lem:anticancel}.
    \Cref{clm:anticancel-t-order} implies that the $t_i$ are ordered such that $t_0 > t_1 > \dots > t_I$.
    For all $0 \leq i \leq I-1$ the nodes $p$ and $q$ are inserted into the banana subtree rooted at $t_i$,
    as either $q = \In{t_i}$ or $q$ is a descendant of $\In{t_i}$.
    Since $\Birth{t_i} < t_i$ and $p<q<t_i$ for $0 \leq i \leq I-1$, it follows that no $t_i$ for $0 \leq i \leq I-1$ violates \cref{inv:condition-I}.
    Recall that $q$ becomes $\Up{t_I}$; this implies that neither the banana subtree rooted at $t_I$ nor the descendants of $\Dn{t_I}$ change
    and thus $t_I$ satisfies the condition of \cref{inv:condition-I} in the new tree.
    Finally, $b$ also satisfies the condition of \cref{inv:condition-I}, as $b < p < q < t_0$ and 
    $p$ and $q$ are inserted into the subtree of $b$ that contains $t_0$.
    It follows that there is no node that violates the condition of \cref{inv:condition-I} and thus the new up-tree satisfies \cref{inv:condition-I}.

    We use the same claim to show that \cref{inv:condition-III-max} is satisfied after inserting $q$ and $p$.
    There are two cases: $t_I \neq t_0$ and $t_I = t_0$.
    In the first case, $q = \In{t_{I-1}}$, $t_I = \Dn{q}$ and $t_{I-1} = \Up{q}$.
    As shown above $t_I > q$ and $t_{I-1} > q$.
    Furthermore, since the while-loop did not terminate for $t = t_{I-1}$ we have $h_\theta(q) < h_\theta(t_{I-1})$.
    It also holds that $h_\theta(t_I) < h_\theta(q)$, as the while-loop terminated for $t = t_I$.
    Thus, $q$ satisfies the conditions in \cref{inv:condition-III-max}.
    In the case $t_I = t_0$, $t_I = \Dn{q}$, $\Up{q} = b$ and either $q = \Mid{b}$ or $q = \Dn{b}$.
    We have $b < q < t_I$, i.e., $\Up{q} < q < \Dn{q}$.
    As in the other case $h_\theta(t_I) < h_\theta(q)$.
    It is also true that $h_\theta(q) < h_\theta(b)$, which follows from $p$ and $q$ being contiguous in $h_\theta$ and the definition of $b$.
    Thus, $q$ again satisfies the conditions in \cref{inv:condition-III-max}.
    For $t_I$ and $t_{I-1}$ the conditions in \cref{inv:condition-III-max} follow from \cref{clm:anticancel-t-order} and the fact that $h_\theta(t_I) < h_\theta(q) < h_\theta(t_{I-1})$.
    No other node can violate \cref{inv:condition-III-max} and thus the new up-tree satisfies \cref{inv:condition-III-max}.

    We have shown that the tree constructed by applying \alganticancel$(p,q)$ to $\UpTree{h_\sigma}$ yields an up-tree that satisfies
    \cref{inv:condition-I,inv:condition-II,inv:condition-III-max} for $h_\theta$.
    By \cref{cor:unique-bananas-local} this is the unique up-tree for $h_\theta$.
\end{proof}

\subsection{Interchanges.}
\paragraph{Interchanges of Maxima.}
An interchange of maxima occurs when the value of a maximum $j$ increases above the value of a maximum $q$
or the value of $q$ decreases below the value of $j$.
This only has structural consequences on the up-tree if $q = \Up{j}$, as this is the only case where one of the invariants, namely \cref{inv:condition-III-max}, is violated.

Denote by $h_\sigma$ and $h_\theta$ the map before and after the interchange,
such that $h_\sigma(j) < h_\sigma(q)$ and $h_\theta(j) > h_\theta(q)$.
Assume $q < j$, which is the situation depicted in Figure~\ref{fig:maxinterchange}.
The case $j < q$ is symmetric.
Let $i$ and $p$ be such that $j = \Dth{i}$ and $q = \Dth{p}$ in $\UpTree{h_\sigma}$.
There are three cases: $j = \Dn{q}$, $j = \In{q}$ and $j = \Mid{q}$.
If $j = \In{q}$ or $j = \Mid{q}$, then $p = \Low{j}$ and by \cref{inv:condition-II} $h_\sigma(i) > h_\sigma(p)$.
If $j = \Dn{q}$ we further distinguish two cases depending on the order of $h_\sigma(i)$ and $h_\sigma(p)$.
We now give the algorithm for processing an interchange of maximum $j$ with maximum $q = \Up{j}$ for the case $q < j$.
\begin{tabbing}
    m\=m\=m\=m\=m\=m\=m\=m\=m\=m\=m\=\kill
    \clreset
    \clbegin \> \texttt{max-interchange}$(j,q)$: \\*
    \clbegin \> \> \> Let $i = \Birth{j}$, $p = \Birth{q}$; \\*
    \clbegin \> \> \> \texttt{if} $j = \In{q}$ \texttt{then} \\*
    \clbegin \> \> \> \> remove $j$ from its trail; insert $j$ as $\Up{q}$ \\*
    \clbegin \> \> \> \texttt{else if} $j = \Mid{q}$ \texttt{then} \\*
    \clbegin \> \> \> \> exchange $j$ and $q$; \\*
    \clbegin \> \> \> \> remove $q$ from its trail; insert $q$ as $\Dn{j}$; \\*
    \clbegin \> \> \> \> swap in- and mid-trail of $i$ and $q$ \\
    \clbegin \> \> \> \texttt{else if} $j = \Dn{q}$ \texttt{then} \\*
    \clbegin \> \> \> \> \texttt{if} $h_\sigma(i) < h_\sigma(p)$ \texttt{then} \\*
    \clbegin \> \> \> \> \> remove $q$ from its trail; insert $q$ as $\In{j}$ \\*
    \clbegin \> \> \> \> \texttt{else} \\*
    \clbegin \> \> \> \> \> exchange $j$ and $q$; \\*
    \clbegin \> \> \> \> \> remove $q$ from its trail; insert $q$ as $\Mid{j}$; \\*
    \clbegin \> \> \> \> \> swap in- and mid-trail of $i$ and $q$ \\*
    \clbegin \> \> \> \> \texttt{endif}; \\*
    \clbegin \> \> \> \texttt{endif}.
\end{tabbing}

The next lemma states that this algorithm correctly maintains the up-tree.
\begin{lemma}[Max-interchange]
    \label{lem:interchange-of-maxima}
    Let $j$ and $q$ be two items.
    Let $\sigma \in [0,1)$ such that in $h_\sigma$ the items $j$ and $q$ are maxima with $q = \Up{j}$
    and such that $j$ and $q$ are contiguous in $h_\sigma$.
    Let $\theta \in (\sigma, 1]$ such that in $h_\theta$ the items $j$ and $q$ are maxima with $h_\theta(j) > h_\theta(q)$
    and such that $j$ and $q$ are contiguous in $h_\theta$.
    Applying \algmaxinterchange$(j,q)$ to the up-tree $\UpTree{h_\sigma}$ yields the up-tree $\UpTree{h_\theta}$.
\end{lemma}
\begin{proof}
    The tree $\UpTree{h_\sigma}$ satisfies \cref{inv:condition-I,inv:condition-II,inv:condition-III-max}.
    Recall that by \cref{lem:inv-1-and-3-are-cond-3} this up-tree also satisfies Conditions~{\sf III.1} and {\sf III.2}.
    We show that the tree obtained by \algmaxinterchange also satisfies these invariants for $h_\theta$,
    which by \cref{cor:unique-bananas-local} implies that it is the unique up-tree for $h_\theta$.

    Write $I_q$ for the nodes internal to the in-trail beginning at $q$
    and $M_q$ for the nodes internal to the mid-trail beginning at $q$ in $\UpTree{h_\sigma}$;
    similarly, write $I_p$, $I_i$, $I_j$ for the nodes internal to the in-trails
    and $M_p$, $M_i$, $M_j$ for the nodes internal to the mid-trails at $p$, $i$, $j$.

    There are three cases in \algmaxinterchange: $j = \In{q}$, $j = \Mid{q}$ and $j = \Dn{q}$.
    As $q = \Up{j}$ there is no other case to consider.
    The case $j = \Dn{q}$ is further divided into two cases based on the values of $\Birth{j}$ and $\Birth{q}$.
    We describe for each case how the trails change from $\UpTree{h_\sigma}$ to the new tree and show that \cref{inv:condition-I,inv:condition-III-max} hold in the new tree.
    Afterwards we prove \cref{inv:condition-II}.
    \smallskip \begin{description}
        \item[Case 1 ($j = \In{q}$):]
            Assume $j < q$; the other case is symmetric.
            Note that $i = \Birth{j} < j$ and $p = \Birth{q} < q$.
            The in-trail and mid-trail between $i$ and $j$ remain the in-trail and mid-trail between $i$ and $j$,
            and the in-trail and mid-trail between $p$ and $q$ remain the in-trail and mid-trail between $p$ and $q$,
            with the exception that $j$ is removed from the in-trail between $p$ and $q$.
            In the new tree the nodes internal to the in-trail between $i$ and $j$ are thus $I_j$ and the nodes internal to the mid-trail between $i$ and $j$ are $M_j$.
            Similarly, the nodes internal to the in-trail between $p$ and $q$ are $I_q \setminus \{j\}$ and the nodes internal to the mid-trail between $p$ and $q$ are $M_q$.

            To see that \cref{inv:condition-I} holds we analyze the subtrees of $j$ and $q$ to show that they satisfy the condition of \cref{inv:condition-I} in the new up-tree.
            For any other node the condition holds in the new tree as the nodes in the respective subtrees do not change.
            First, consider the node $q$.
            The node $\Dn{q}$ is the same in $\UpTree{h_\sigma}$ and in the new tree. The algorithm also does not modify the descendants of $\Dn{q}$.
            Thus, for all descendants $u$ of $\Dn{q}$ including $\Dn{q}$ we have $q < u$ by \cref{inv:condition-I} for $\UpTree{h_\sigma}$, as $p = \Birth{q} < q$.
            The banana subtree rooted at $q$ in the new tree differs from that in $\UpTree{h_\sigma}$ only by the banana subtree rooted at $j$, which the algorithm removes.
            Thus, for all nodes $v$ in the banana subtree rooted at $q$ it holds that $v < q$ by \cref{inv:condition-I} for $\UpTree{h_\sigma}$, as $p = \Birth{q} < q$.
            It follows that $q$ satisfies the condition for \cref{inv:condition-I}.
            Now consider the node $j$.
            The algorithm does not change the banana subtree rooted at $j$ and thus for all $v$ in this subtree we have $v < j$ by \cref{inv:condition-I} for $\UpTree{h_\sigma}$, as $i = \Birth{j} < j$.
            The other subtree of $j$, i.e., the descendants of $\Dn{j}$ differ significantly:
            in $\UpTree{h_\sigma}$ these were the descendants of nodes on the in-trail below $j$;
            in the new tree they are the descendants of $q$.
            For the descendants $u_1$ of $\Dn{q}$ it is easy to see that $j < u_1$, as $j < q < u_1$, which follows from the discussion for $q$ and the assumption that $j < q$.
            For the descendants $u_2$ of nodes on the banana spanned by $p$ and $q$ in the new tree it follows from Conditions {\sf III.1} and {\sf III.2} and \cref{inv:condition-I} that $j < u_2$:
            the node $j$ is the topmost node on the in-trail, and with $j < q$ this implies that the other maxima $t$ on the banana satisfy $j < t$;
            nodes $u_3$ in the banana subtrees rooted at each $t$ are either descendants of $\Dn{j}$, which implies $j < u_3$ by \cref{inv:condition-I} for $\UpTree{h_\sigma}$,
            or $t$ is on the mid-trail between $p$ and $q$, which by $p < q$ and \cref{inv:condition-I} also implies $j < t < u_3$.
            It follows that all descendants $u$ of $\Dn{j} = q$, including $\Dn{j}$ satisfy $j < u$.
            Thus, the condition of \cref{inv:condition-I} holds for $j$.

            We now prove that \cref{inv:condition-III-max} holds by again analyzing $j$ and $q$; for the remaining nodes the conditions of \cref{inv:condition-III-max} follow directly from \cref{inv:condition-III-max} for $\UpTree{h_\sigma}$.
            Let $u_q = \Up{q}$ and $d_q = \Dn{q}$ in $\UpTree{h_\sigma}$.
            In the new tree $d_q = \Dn{q}$, $u_q = \Up{j}$ and $j = \Up{q}$.
            By the assumption $j < q$ and \cref{inv:condition-III-max} for $\UpTree{h_\sigma}$ we have $j < q < d_q$ in the new tree, i.e., $\Up{q} < q < \Dn{q}$.
            If $q = \In{u_q}$ in $\UpTree{h_\sigma}$ then $q < u_q$ and in the new tree $j = \In{u_q}$ and $j < u_q$.
            Thus, if $j = \In{\Up{j}} = \In{u_q}$ in the new tree then $j < \Up{j}$ and $j < \Dn{j} = q$.
            Otherwise, if $q \neq \In{u_q}$ in $\UpTree{h_\sigma}$ then $u_q < q$ and $u_q < j$ by \cref{inv:condition-I} for $\UpTree{h_\sigma}$.
            In the new tree this implies $j \neq \In{\Up{j}} = \In{u_q}$ and $u_q = \Up{j} < j < \Dn{j} = q$.
            The inequalities $h_\theta(\Up{j}) > h_\theta(j) > h_\theta(\Dn{j}) = h_\theta(q)$ and $h_\theta(j) = h_\theta(\Up{q}) > h_\theta(q) > h_\theta(\Dn{q})$ follow directly from the fact that $j$ and $q$ are contiguous in $h_\theta$ and from \cref{inv:condition-III-max} for $\UpTree{h_\sigma}$.
            Thus, $j$ and $q$ satisfy the conditions of \cref{inv:condition-III-max} in the new tree and it follows that \cref{inv:condition-III-max} holds in the new tree.
        \item[Case 2 ($j = \Mid{q}$):]
            Assume $j < q$; the other case is symmetric.
            Note that $i = \Birth{j} > j$ and $p = \Birth{q} < q$.
            The nodes $q$ and $j$ exchange their partners, i.e., in the new up-tree $\Birth{q} = i$ and $\Birth{j} = p$.
            The in-trail and mid-trail between $i$ and $j$ become the mid- and in-trail between $i$ and $q$ (notice that the trails change from in to mid and vice versa);
            the in-trail and mid-trail between $p$ and $q$ become the in- and mid-trail between $p$ and $q$, with $j$ removed from the mid-trail.
            That is, in the new tree the nodes internal to the in-trail between $i$ and $q$ are the nodes $M_j$ and the nodes internal to the mid-trail between $i$ and $q$ are the nodes $I_j$.
            The nodes internal to the in-trail between $p$ and $j$ are the nodes $I_q$
            and the nodes internal to the mid-trail between $p$ and $j$ are the nodes $M_q \setminus \{j\}$.

            The proof of \cref{inv:condition-III-max} is identical to that of Case 1.
            To see that \cref{inv:condition-I} holds we again analyze the subtrees of $j$ and $q$ to show that they satisfy the condition of \cref{inv:condition-I}.
            As in the previous case the remaining nodes satisfy the condition in the new tree, as their subtrees are unchanged by the algorithm.
            For the node $q$ observe that $\Dn{q}$ and its descendants are the same in $\UpTree{h_\sigma}$ and the new tree.
            Furthermore, the nodes in the banana subtree rooted at $q$ in the new tree are also in the banana subtree rooted at $q$ in $\UpTree{h_\sigma}$.
            The condition of \cref{inv:condition-I} thus holds for $q$ by \cref{inv:condition-I} for $\UpTree{h_\sigma}$.
            Now we consider the node $j$.
            The descendants of $\Dn{j}$ are (1) the descendants of $\Dn{q}$, including $\Dn{q}$, and (2) the node $q$ along with the banana subtree rooted at $q$.
            For the descendants $u_1$ of $\Dn{q}$, including $\Dn{q}$, it holds that $j < u_1$ by \cref{inv:condition-I} for $\UpTree{h_\sigma}$.
            The banana subtree rooted at $q$ in the new tree is the banana subtree rooted at $j$ in $\UpTree{h_\sigma}$,
            and for nodes $u_2$ in this banana subtree $j < u_2$ by \cref{inv:condition-I} for $\UpTree{h_\sigma}$, since $j < i = \Birth{j}$.
            With the assumption $j < q$ it follows that the descendants $u$ of $\Dn{j}$ including $\Dn{j}$ satisfy $j < u$ in the new tree.
            The nodes $v$ in the banana subtree $j$ in the new tree are descendants of maxima $t \neq j$ on a trail between $p$ and $q$.
            By Conditions {\sf III.1}, {\sf III.2} and \cref{inv:condition-I} for $\UpTree{h_\sigma}$ these nodes satisfy $v < j$:
            the node $j$ is the topmost node of the right trail of the banana spanned by $p$ and $q$;
            the nodes $t_1$ on the left trail satisfy $t_1 < j$ by Conditions {\sf III.1} and {\sf III.2} and the nodes $v_1$ in the banana subtrees rooted at $t_1$ satisfy $v_1 < t_1 < j$ by \cref{inv:condition-I};
            the nodes $t_2$ on the right trail below $j$ are descendants of $\Dn{j}$, and they along with their descendants $v_2$ satisfy $t_2 < j$ and $v_2 < j$ by \cref{inv:condition-I}.
            It follows that $j$ satisfies the condition of \cref{inv:condition-I}.
            Since $j$ and $q$ satisfy the condition of \cref{inv:condition-I} it follows that \cref{inv:condition-I} holds in the new tree.
        \item[Case 3.1 ($j = \Dn{q}$ and $h_\sigma(i) < h_\sigma(p)$):]
            Assume $q < j$; the other case is symmetric.
            Note that $i = \Birth{j} < j$ and $p = \Birth{q} < q$.
            The in-trails and mid-trails remain the same, with the exception that $q$ is added to the in-trail between $i$ and $j$.
            That is, in the new tree the nodes internal to the in-trail between $i$ and $j$ are $I_j \cup \{q\}$,
                the nodes internal to the mid-trail between $i$ and $j$ are $M_j$,
                the nodes internal to the in-trail between $p$ and $q$ are $I_q$
                and the nodes internal to the mid-trail between $p$ and $q$ are $M_q$. 

            To see that \cref{inv:condition-I} holds we again analyze the subtrees of $j$ and $q$, as in the previous cases.
            The banana subtree rooted at $q$ is the same in $\UpTree{h_\sigma}$ and the new tree.
            The set of descendants of $\Dn{q}$ including $\Dn{q}$ in the new tree is a subset of that in $\UpTree{h_\sigma}$.
            Thus, the condition of \cref{inv:condition-I} holds for $q$ in the new tree by \cref{inv:condition-I} for $\UpTree{h_\sigma}$.
            For the node $j$ the $\Dn{j}$ and its descendants do not change.
            The banana subtree rooted at $j$ in the new tree is equal to that in $\UpTree{h_\sigma}$, with the addition of $q$ and the banana subtree rooted at $q$.
            By assumption $q < j$ and by \cref{inv:condition-I} for $\UpTree{h_\sigma}$ the nodes $v_1$ in the banana subtree rooted at $q$ satisfy $v_1 < q$ and thus $v_1 < j$.
            The remaining nodes $v_2$ in the banana subtree of $j$ satisfy $v_2 < j$ by \cref{inv:condition-I} for $\UpTree{h_\sigma}$.
            Thus, for all nodes $v$ in the banana subtree rooted at $j$ we have $v < j$.
            It follows that $j$ satisfies the condition of \cref{inv:condition-I} and that the new up-tree satisfies \cref{inv:condition-I}.

            We now show that \cref{inv:condition-III-max} holds.
            Let $i_j = \In{j}$, $u_q = \Up{q}$ in $\UpTree{h_\sigma}$.
            Observe that $q$ is inserted at the top of the left trail between $i$ and $j$, i.e., $q = \In{j}$ in the new tree.
            By \cref{inv:condition-I} and the assumption that $q < j$ it follows that $q < \Up{q} = j$ and $q < \Dn{q}$.
            Note that $u_q = \Up{j}$ in the new tree.
            The inequality $j < \Dn{j}$ holds in $\UpTree{h_\sigma}$ since $j$ is on a left trail, and holds in the new tree since $\Dn{j}$ does not change.
            If $q = \In{u_q}$ in $\UpTree{h_\sigma}$, then $q < u_q$ and by \cref{inv:condition-I} $j < u_q$.
            In the new tree $j = \In{u_q}$ and thus $j < \Up{j} = u_q$ and $j < \Dn{j}$.
            Otherwise, if $q \neq \In{u_q}$ in $\UpTree{h_\sigma}$, then $u_q < q < j$.
            Thus, in the new tree $u_q = \Up{j} < j < \Dn{j}$.
            The inequalities $h_\theta(\Up{q}) > h_\theta(q) > h_\theta(\Dn{q})$ and $h_\theta(\Up{j}) > h_\theta(j) > h_\theta(\Dn{j})$
            follow directly from \cref{inv:condition-III-max} for $\UpTree{h_\sigma}$ and the fact that $j$ and $q$ are contiguous in $h_\theta$.
        \item[Case 3.2 ($j = \Dn{q}$ and $h_\sigma(i) > h_\sigma(p)$):]
            Assume $q < j$; the other case is symmetric.
            Note that $i = \Birth{j} < j$ and $p = \Birth{q} < q$.
            The nodes $q$ and $j$ exchange their partners, i.e., in the new up-tree $\Birth{q} = i$ and $\Birth{j} = p$.
            The in-trail and mid-trail between $i$ and $j$ become the mid-trail and in-trail between $i$ and $p$ (notice that the trails change from in to mid and vice versa);
            the in-trail and mid-trail between $p$ and $q$ become the in-trail and mid-trail between $p$ and $j$.
            Note that $q$ is added to the mid-trail between $p$ and $j$.
            Thus, in the new tree the nodes internal to the in-trail between $p$ and $j$ are $I_q$,
                the nodes internal to the mid-trail between $p$ and $j$ are $M_q \cup \{q\}$,
                the nodes internal to the in-trail between $i$ and $q$ are $M_j$
                and the nodes internal to the mid-trail between $i$ and $q$ are $I_q$.

            The proof for \cref{inv:condition-III-max} is identical to that of Case 3.1.
            The proof for \cref{inv:condition-I} is similar to that of Case 3.1 and we point out the differences.
            The banana subtree rooted at $q$ in the new tree is the banana subtree rooted at $j$ in $\UpTree{h_\sigma}$.
            Since the nodes $u$ in this subtree were descendants of $\Dn{q}$ in $\UpTree{h_\sigma}$ it follows from \cref{inv:condition-I} for $\UpTree{h_\sigma}$ that $q < u$.
            The descendants of $\Dn{q}$ including $\Dn{q}$ are nodes formerly in the banana subtree rooted at $q$, and these nodes $v$ satisfy $v < q$ by \cref{inv:condition-I} for $\UpTree{h_\sigma}$.
            As in the previous case, the descendants of $\Dn{j}$ including $\Dn{j}$ are unchanged.
            The banana subtree rooted at $j$ in the new tree is the banana subtree rooted at $j$ in $\UpTree{h_\sigma}$ combined with the banana subtree rooted at $q$ in $\UpTree{h_\sigma}$.
            \Cref{inv:condition-I} follows by the same argument as above.
    \end{description} \smallskip
    It remains to prove that \cref{inv:condition-II} holds in the new tree.
    Observe that $\Low{\cdot}$ pointers only change for $q$ and $j$.
    Consequently, we need to show that $p$ and $i$ satisfy $h_\theta(\Low{\Dth{p}} < h_\theta(p)$ and $h_\theta(\Low{\Dth{i}}) < h_\theta(i)$.
    We consider the three cases above:
    \smallskip \begin{description}
        \item[Case 1 ($j = \In{q}$):] We have $\Low{j} = p$ in $\UpTree{h_\sigma}$ and $\Low{j} = \Low{q}$ in the new tree.
            Recall that $j = \Dth{i}$ and $q = \Dth{p}$ in both $\UpTree{h_\sigma}$ and in the new tree.
            Since $\UpTree{h_\sigma}$ satisfies \cref{inv:condition-II}, we know that
            $h_\sigma(i) > h_\sigma(\Low{j}) = h_\sigma(p)  > h_\sigma(\Low{q})$.
            As $h_\sigma(u) = h_\theta(u)$ for $u \neq j,q$, it holds that $h_\theta(i) > h_\theta(\Low{j}) = h_\theta(\Low{q})$ and $h_\theta(p) > h_\theta(\Low{q})$.
        \item[Case 2 ($j = \Mid{q}$):] We have $\Low{j} = p$ in $\UpTree{h_\sigma}$ and $\Low{j} = \Low{q}$ in the new tree.
            In the new tree $\Dth{i} = q$ and $\Dth{p} = j$.
            It holds that $h_\sigma(i) > h_\sigma(p) > h_\sigma(\Low{q})$ by \cref{inv:condition-II} for $\UpTree{h_\sigma}$ and thus $h_\theta(i) > h_\theta(\Low{q})$.
            Similarly $h_\theta(p) = h_\sigma(p) > h_\sigma(\Low{q}) = h_\theta(\Low{q}) = h_\theta(\Low{j})$, i.e., $h_\theta(p) > h_\theta(\Low{j})$.
        \item[Case 3 ($j = \Dn{q}$):] We distinguish the two cases $h_\sigma(i) < h_\sigma(p)$ and $h_\sigma(i) > h_\sigma(p)$.
            \begin{description}
                \item[$h_\sigma(i) < h_\sigma(p)$:] We have $\Low{q} = \Low{j}$ in $\UpTree{h_\sigma}$ and $\Low{j} \neq \Low{q} = i$ in the new tree.
                    Note that $\Low{j}$ is the same in both trees.
                    In the new tree it holds that $h_\theta(p) = h_\sigma(p) > h_\sigma(\Low{q}) = h_\sigma(i) = h_\theta(i)$, i.e., $h_\theta(p) > h_\theta(i)$.
                    As $\Low{j}$ does not change and $\Dth{i} = j$ in both trees, $h_\theta(i) > h_\theta(\Low{j})$ in the new tree.
                \item[$h_\sigma(i) > h_\sigma(p)$:] We have $\Low{q} = \Low{j}$ in $\UpTree{h_\sigma}$ and $\Low{j} \neq \Low{q} = p$ in the new tree.
                    In the new tree we also have $\Dth{i} = q$ and $\Dth{p} = j$.
                    Note again that $\Low{j}$ is the same in both trees.
                    In the new tree it holds that $h_\theta(p) = h_\sigma(p) > h_\sigma(\Low{j}) = h_\theta(\Low{j})$,
                    since in $\UpTree{h_\sigma}$ $h_\sigma(p) > h_\sigma(\Low{q})$ and $\Low{j} = \Low{q}$.
                    Furthermore, in the new tree $h_\theta(i) = h_\sigma(i) > h_\sigma(p) = h_\theta(p)$.
            \end{description}
    \end{description} \smallskip
    In all cases $p$ and $i$ satisfy the required condition and thus \cref{inv:condition-II} is satisfied in the new tree.
    We have shown that the new up-tree obtained by applying \algmaxinterchange to $\UpTree{h_\sigma}$ satisfies \cref{inv:condition-I,inv:condition-II,inv:condition-III-max}, which implies that it is indeed the unique up-tree for $h_\theta$.
\end{proof}

\paragraph{Interchanges of Minima.}
Let $i$ be a minimum and $p = \Low{\Dth{i}}$.
In an up-tree satisfying \cref{inv:condition-II} the function value of $i$ is greater than the function value of $p$.
If this relation changes such that the function value of $p$ becomes greater than the function value of $i$, \cref{inv:condition-II} is violated.
To fix this violation we perform an \emph{interchange of minima}.

Let $\sigma \in [0,1)$ such that $h_\sigma(i) > h_\sigma(\Low{\Dth{i}}) = h_\sigma(p)$ and such that $i$ and $p$ are contiguous.
Let $\theta \in (\sigma, 1]$ such that $h_\theta(i) < h_\theta(p)$ and such that $i$ and $p$ are contiguous.
Write $j = \Dth{i}$ and $q = \Dth{p}$.
We give the detailed algorithm for the case $q < p$; the other case is symmetric.
\begin{tabbing}
    m\=m\=m\=m\=m\=m\=m\=m\=m\=m\=m\=\kill
    \clreset
    \clbegin                  \> \algmininterchange$(i,p)$: \\*
    \clbegin                  \> \> \texttt{if} $\Low{\Dth{i}} \neq p$ \texttt{then exit endif}. \\*
    \clbegin                  \> \> Let $j = \Dth{i}$, $q = \Dth{p}$. \\*
    \clbeginlbl{ln:j-nabors}  \> \> Let $u_j = \Up{j}$, $d_j = \Dn{j}$, $i_j = \In{j}$, $m_j = \Mid{j}$. \\*
    \clbeginlbl{ln:split}     \> \> From $q$ down along the trail not containing $j$ find \\*
    \clbegin                  \> \> \> 1. first node $s^-$ such that $h_\sigma(s^-) < h_\sigma(j)$ \\*
    \clbegin                  \> \> \> 2. last node $s^+$ such that $h_\sigma(s^+) > h_\sigma(j)$. \\*
    \clbeginlbl{ln:j-ptrs}    \> \> Set $\Dn{j} \gets m_j$, $\Mid{j} \gets d_j$, $\In{j} \gets s^-$, $\Up{j} \gets s^+$. \\*
    \clbeginlbl{ln:mid-vs-in} \> \> \texttt{if} $q < j < p$ \texttt{then} \\*
    \clbegin                  \> \> \> Swap $\In{i}$ and $\Mid{i}$; \\*
    \clbegin                  \> \> \> \texttt{if} $i_j = i$ \texttt{then} $\Mid{i_j} = u_j$ \texttt{else} $\Up{i_j} = u_j$ \texttt{endif}; \\*
    \clbegin                  \> \> \> \texttt{if} $u_j = q$ \texttt{then} $\Mid{u_j} = i_j$ \texttt{else} $\Dn{u_j} = i_j$ \texttt{endif}; \\*
    \clbeginlbl{ln:splus-m}   \> \> \> \texttt{if} $s^+ = q$ \texttt{then} $\In{s^+} = j$ \texttt{else} $\Dn{s^+} = j$ \texttt{endif}; \\*
    \clbegin                  \> \> \> \texttt{if} $s^- = p$ \texttt{then} $\In{s^-} = j$ \texttt{else} $\Up{s^-} = j$ \texttt{endif} \\*
    \clbegin                  \> \> \texttt{else if} $q < p < j$ \\
    \clbegin                  \> \> \> Swap $\In{p}$ and $\Mid{p}$; \\*
    \clbegin                  \> \> \> \texttt{if} $i_j = i$ \texttt{then} $\In{i_j} = u_j$ \texttt{else} $\Up{i_j} = u_j$ \texttt{endif}; \\*
    \clbegin                  \> \> \> \texttt{if} $u_j = q$ \texttt{then} $\In{u_j} = i_j$ \texttt{else} $\Dn{u_j} = i_j$ \texttt{endif}; \\*
    \clbeginlbl{ln:splus-i}   \> \> \> \texttt{if} $s^+ = q$ \texttt{then} $\Mid{s^+} = j$ \texttt{else} $\Dn{s^+} = j$ \texttt{endif}; \\*
    \clbegin                  \> \> \> \texttt{if} $s^- = p$ \texttt{then} $\In{s^-} = j$ \texttt{else} $\Up{s^-} = j$ \texttt{endif} \\*
    \clbegin                  \> \> \texttt{endif}. \\*
    \clbeginlbl{ln:ass-dth}   \> \> Set $\Dth{i} = q$, $\Dth{p} = j$. \\*
    \clbegin                  \> \> Set $\Low{u} = i$ for nodes $u \neq q$ between $u_j$ and $q$ and $j$ and $q$, including $u_j$ and $j$. \\*
\end{tabbing}

\begin{lemma}[Min-interchange]
    Let $p$ and $i$ be items.
    Let $\sigma \in [0,1)$ such that $p$ and $i$ are minima and contiguous in $h_\sigma$, $h_\sigma(i) > h_\sigma(p)$ and $p = \Low{\Dth{i}}$ in $\UpTree{h_\sigma}$.
    Let $\theta \in (\sigma, 1]$ such that $p$ and $i$ are minima and contiguous in $h_\theta$, and $h_\theta(i) < h_\theta(p)$.
    Applying \algmininterchange($i, p$) to $\UpTree{h_\sigma}$ yields the up-tree for $h_\theta$.
\end{lemma}
\begin{proof}
    The algorithm defines $j = \Dth{i}$, $q = \Dth{p}$.
    Assume $q < p$; the other case is symmetric.
    We show that by applying \algmininterchange to $\UpTree{h_\sigma}$, which satisfies \cref{inv:condition-I,inv:condition-II,inv:condition-III-max},
    we obtain an up-tree satisfying \cref{inv:condition-I,inv:condition-II,inv:condition-III-max} for $h_\theta$.
    By \cref{cor:unique-bananas-local} this is the unique up-tree for $h_\theta$, $\UpTree{h_\theta}$.
    Note that by definition of $h_\lambda$ the maps $h_\sigma$ and $h_\theta$ differ in either the value of $i$ or the value of $p$,
    and are equal in the values of all other items.

    In \cref{ln:split} finds two nodes on the trail between $p$ and $q$ that does not contain $j$:
    $s^-$ with $h_\sigma(s^-) < h_\sigma(j)$ and $s^+$ with $h_\sigma(s^+) > h_\sigma(j)$,
    with $s^-$ being the highest such node and $s^+$ being the lowest such node along the trail.
    Since by \cref{inv:condition-III-max} nodes are ordered by function value along trails, $s^-$ and $s^+$ are neighbors along the trail.
    
    It may be the case that $s^-$ is a node internal to the trail or that $s^- = p$,
    and similarly $s^+$ is either internal to the trail or $s^+ = q$.
    Similar statements hold for the nodes defined in \cref{ln:j-nabors}:
    \smallskip \begin{itemize}
        \item $u_j = \Up{j}$ is either internal to the trail between $p$ and $q$ or $u_j = q$,
        \item $d_j = \Dn{j}$ is either internal to the trail between $p$ and $q$ or $d_j = p$,
        \item $i_j = \In{j}$ is either internal to the in-trail between $i$ and $j$ or $i_j = i$,
        \item $m_j = \Mid{j}$ is either internal to the mid-trail between $i$ and $j$ or $m_j = i$.
    \end{itemize} \smallskip
    The if-statement in \cref{ln:mid-vs-in} distinguishes two cases: $q < j < p$ and $q < p < j$.
    The cases correspond to $j$ being on the mid-trail between $p$ and $q$ and $j$ being on the in-trail between $p$ and $q$.
    As $j = \Dth{i}$ and $p = \Low{\Dth{i}}$ the node $j$ must be in one of the trails between $p$ and $q$.
    Thus, these two cases are the only cases that can occur.

    We first prove \cref{inv:condition-III-max}.
    Recall that by \cref{lem:inv-1-and-3-are-cond-3}, $\UpTree{h_\sigma}$ satisfies Conditions {\sf III.1} and {\sf III.2}.
    We consider each node for which the $\Up{\cdot}$ or $\Dn{\cdot}$ pointer changes individually.
    \smallskip \begin{description}
        \item[Node $j$:] in the new tree $\Up{j} = s^+$ and $\Dn{j} = m_j$.
            By definition of $h_\lambda$ it holds that $h_\sigma(s^+) = h_\theta(s^+)$ and $h_\sigma(j) = h_\theta(j)$.
            As $h_\sigma(s^+) > h_\sigma(j)$ it follows that $h_\theta(\Up{j}) > h_\theta(j)$.
            If $m_j \neq i$ then $h_\sigma(m_j) = h_\theta(m_j)$; if $m_j = i$ then $h_\sigma(m_j) \geq h_\theta(m_j)$ as either $h_\sigma(i) = h_\theta(i)$ or $h_\theta(i) < h_\theta(p) = h_\sigma(p)$.
            Thus, $h_\theta(m_j) < h_\theta(j)$.
            It follows that $h_\theta(\Dn{j}) < h_\theta(j) < h_\theta(\Up{j})$.
            For the other conditions of \cref{inv:condition-III-max} we consider the two cases of \cref{ln:mid-vs-in} separately:
            \begin{description}
                \item[$q < j < p$:]
                    the node $j$ is on the mid-trail between $p$ and $q$, which by Conditions {\sf III.1} and {\sf III.2} implies that $j < \Dn{j}$ in $\UpTree{h_\sigma}$.
                    It follows that $m_j < j$ by \cref{inv:condition-I} for $\UpTree{h_\sigma}$, since in that tree $m_j$ is in the banana tree rooted at $j$.
                    Thus, in the new tree $m_j = \Dn{j} < j$.
                    If $s^+ = q$, then $q = s^+ = \Up{j} < j$. Since in this case $j = \In{s^+}$ (see \cref{ln:splus-m}) in the new tree $j$ satisfies point 3 of \cref{inv:condition-III-max}.
                    If $s^+ \neq q$, then $j < s^+$ by Conditions {\sf III.1} and {\sf III.2}, as $s^+$ is in the in-trail between $p$ and $q$.
                    Thus $\Dn{j} < j < \Up{j}$ in the new tree and hence $j$ satisfies point 2 of \cref{inv:condition-III-max}.
                \item[$q < p < j$:]
                    the node $j$ is on the in-trail between $p$ and $q$, which by Conditions {\sf III.1} and {\sf III.2} implies that $\Dn{j} < j$ in $\UpTree{h_\sigma}$.
                    It follows that $j < m_j$ by \cref{inv:condition-I} for $\UpTree{h_\sigma}$.
                    Thus, in the new tree $j < \Dn{j} = m_j$.
                    If $s^+ = q$, then $q = s^+ = \Up{j} < j$. Since in this case $j = \Mid{s^+}$ (see \cref{ln:splus-i}) in the new tree $j$ satisfies point 2 of \cref{inv:condition-III-max}.
                    If $s^+ \neq q$, then $s^+ < j$ by Conditions {\sf III.1} and {\sf III.2}, as $s^+$ is in the mid-trail between $p$ and $q$.
                    Thus, $\Up{j} < j < \Dn{j}$ in the new tree and hence $j$ satisfies point 2 of \cref{inv:condition-III-max}.
            \end{description}
        \item[Node $i_j$:] we only need to consider the case where $i_j \neq i$, as $i$ is a minimum and does not affect \cref{inv:condition-III-max}.
            Since $\Dn{i_j}$ is the same in $\UpTree{h_\sigma}$ and in the new tree, by \cref{inv:condition-III-max} it holds that $h_\sigma(\Dn{i_j}) < h_\sigma(i_j)$.
            If $\Dn{i_j} \neq i$ then $h_\sigma(\Dn{i_j}) = h_\theta(\Dn{i_j})$; otherwise if $\Dn{i_j} = i$ then $h_\theta(\Dn{i_j}) \leq h_\sigma(\Dn{i_j})$ by the assumptions on $i$ and $p$.
            Hence, we have $h_\theta(\Dn{i_j}) < h_\sigma(i_j) = h_\theta(i_j)$.
            In the new tree $\Up{i_j} = u_j$.
            By \cref{inv:condition-III-max} for $\UpTree{h_\sigma}$ and the definition of $h_\lambda$ it holds that
            $h_\theta(u_j) = h_\sigma(u_j) > h_\sigma(j) > h_\sigma(i_j) = h_\theta(i_j)$.
            Thus, $h_\theta(\Dn{i_j}) < h_\theta(i_j) < h_\theta(\Up{i_j})$.
            For the other conditions of \cref{inv:condition-III-max} we consider the two cases of \cref{ln:mid-vs-in} separately:
            \begin{description}
                \item[$q < j < p$:] $j$ is on the mid-trail and thus by Conditions {\sf III.1} and {\sf III.2} $j < \Dn{j}$ in $\UpTree{h_\sigma}$.
                    By \cref{inv:condition-I} it follows that in $\UpTree{h_\sigma}$ $i_j = \In{j} < j$ and by \cref{inv:condition-III-max} $i_j < \Dn{i_j}$.
                    Also by \cref{inv:condition-III-max} $u_j = \Up{j} < j$ in $\UpTree{h_\sigma}$.
                    As $i_j$ and $j$ are in the same subtree of $u_j$ in $\UpTree{h_\sigma}$ \cref{inv:condition-I} implies that $u_j < i_j$ since $u_j < j$.
                    Thus, in the new tree $u_j = \Up{i_j} < i_j < \Dn{i_j}$.
                \item[$q < p < j$:] $j$ is on the in-trail and thus by Conditions {\sf III.1} and {\sf III.2} $\Dn{j} < j$ in $\UpTree{h_\sigma}$.
                    By \cref{inv:condition-I} and \cref{inv:condition-III-max} it follows that in $\UpTree{h_\sigma}$ $j = \Up{i_j} < i_j$ and $\Dn{i_j} < i_j$.
                    If $u_j = q$ then by \cref{inv:condition-III-max} $u_j = \Up{j} < j$ in $\UpTree{h_\sigma}$ and thus in the new tree $u_j = \Up{j} < i_j$;
                        in this case $i_j = \In{\Up{i_j}}$ in the new tree and hence $i_j$ satisfies point 3 of \cref{inv:condition-III-max}.
                    If $u_j \neq q$ then by \cref{inv:condition-III-max} $u_j = \Up{j} > j$ in $\UpTree{h_\sigma}$.
                        As $j$ and $i_j$ are in the same subtree of $u_j$ in $\UpTree{h_\sigma}$ \cref{inv:condition-I} implies that $u_j > i_j$.
                        It follows that in the new $\Dn{i_j} < i_j < \Up{i_j} = u_j$ and hence $i_j$ satisfies point 2 of \cref{inv:condition-III-max}.
            \end{description}
        \item[Node $u_j$:] we only need to consider the case where $u_j \neq q$, as in the other case $\Up{u_j}$ and $\Dn{u_j}$ are not changed.
            Note that only $\Dn{u_j}$ changes and that $\Up{u_j}$ is the same in $\UpTree{h_\sigma}$ and the new tree.
            By \cref{inv:condition-I} and the definition of $h_\lambda$ we have $h_\theta(u_j) = h_\sigma(u_j) < h_\sigma(\Up{u_j}) = h_\theta(\Up{u_j})$.
            As shown for $i_j$ we have $h_\theta(i_j) = h_\theta(\Dn{u_j}) < h_\theta(u_j)$ in the new tree and thus
            $h_\theta(\Dn{u_j}) < h_\theta(u_j) < h_\theta(\Up{u_j})$.
            We now consider the two cases of \cref{ln:mid-vs-in}:
            \begin{description}
                \item[$q < j < p$:] as shown for node $i_j$ it holds that $u_j < \Dn{u_j} = i_j$ in the new tree.
                    As $j$ is on the mid-trail between $p$ and $q$ in $\UpTree{h_\sigma}$ so is $u_j$ and thus $\Up{u_j} < u_j$.
                    Hence, $\Up{u_j} < u_j < \Dn{u_j}$ in the new tree.
                \item[$q < p < j$:] recall that $u_j \neq q$. We have shown for node $i_j$ that $\Dn{u_j} = i_j < u_j$.
                    As $j$ is on the in-trail between $p$ and $q$ in $\UpTree{h_\sigma}$ so is $u_j$ and thus either $q = \Up{u_j} < u_j$ or $\Up{u_j} > u_j$.
                    In the first case $u_j$ satisfies point 3 of \cref{inv:condition-III-max} and point 2 of \cref{inv:condition-III-max} in the new tree.
            \end{description}
        \item[Node $s^+$:] we only need to consider the case where $s^+ \neq q$, as in the other case $\Up{s^+}$ and $\Dn{s^+}$ are not changed.
            This implies that in the new tree $\Dn{s^+} = j$.
            The details are similar to those for node $j$.
            We have seen there that $h_\theta(s^+) > h_\theta(j)$ and thus in the new tree $h_\theta(s^+) > h_\theta(\Dn{s^+})$.
            As the $\Up{s^+}$ is the same in the new tree as in $\UpTree{h_\sigma}$ it follows that $h_\theta(\Up{s^+}) > h_\theta(s^+)$ also in the new tree.
            Now consider the two cases of \cref{ln:mid-vs-in}:
            \begin{description}
                \item[$q < j < p$:] as above this implies $j = \Dn{s^+} < s^+$. 
                    If $\Up{s^+} = q$ then $s^+ = \In{\Up{s^+}}$, $\Up{s^+} < s^+$ and hence $s^+$ satisfies point 3 of \cref{inv:condition-III-max}.
                    If $\Up{s^+} \neq q$ then $s^+ < \Up{s^+}$ by Conditions {\sf III.1} and {\sf III.2} and hence $s^+$ satisfies point 2 of \cref{inv:condition-III-max}.
                \item[$q < p < j$:] as above this implies $j = \Dn{s^+} > s^+$.
                    As $s^+$ is on the mid-trail from $p$ to $q$ in $\UpTree{h_\sigma}$ it follows that $s^+ > \Up{s^+}$.
                    Thus, $s^+$ satisfies point 2 of \cref{inv:condition-III-max}.
            \end{description}
        \item[Node $s^-$:] we only need to consider the case where $s^- \neq p$, as $p$ has a minimum and does not affect \cref{inv:condition-III-max}.
            This implies that in the new tree $\Up{s^-} = j$.
            The pointer $\Dn{s^-}$ remains unchanged and thus $h_\theta(\Dn{s^-}) < h_\theta(s^-)$.
            As $h_\theta(s^-) = h_\sigma(s^-) < h_\sigma(j) = h_\theta(j)$ we also have $h_\theta(s^-) < h_\theta(\Up{s^-})$.
            Since $s^- = \In{j} = \In{\Up{s^-}}$ in the new tree it remains to show that $\Up{s^-} = j$ and $\Dn{s^-}$ are on the same side of $s^-$ in both cases of \cref{ln:mid-vs-in}.
            \begin{description}
                \item[$q < j < p$:] $j$ is on the mid-trail and $s^-$ is on the in-trail between $p$ and $q$, and thus by Conditions {\sf III.1} and {\sf III.2} $j < s^-$.
                    These conditions also imply that $\Dn{s^-} < s^-$.
                \item[$q < p < j$:] $j$ is on the in-trail and $s^-$ is on the mid-trail between $p$ and $q$, and thus by Conditions {\sf III.1} and {\sf III.2} $s^- < j$.
                    These conditions also imply that $s^- < \Dn{s^-}$
            \end{description}
            In both cases $s^-$ satisfies point 3 of \cref{inv:condition-III-max}.
    \end{description} \smallskip
    All nodes for which the $\Up{\cdot}$ or $\Dn{\cdot}$ pointer changed satisfy the conditions of \cref{inv:condition-III-max} in the new tree,
    and thus the new tree satisfies \cref{inv:condition-III-max}.

    \smallskip
    We now show that the new tree also satisfies \cref{inv:condition-I}.
    For any maximum, if the contents of the subtrees of that maximum are unchanged, then the condition of \cref{inv:condition-I} continues to hold for this maximum.
    That is, for a maximum $b$, if the set of descendants of $\Dn{b}$ is the same in $\UpTree{h_\sigma}$ and in the new tree,
    and if the set of nodes in the banana subtree rooted at $b$ is the same in $\UpTree{h_\sigma}$ and in the new tree,
    and if $\Birth{b} < b$ (or $\Birth{b} > b$) in $\UpTree{h_\sigma}$ and in the new tree,
    then $b$ satisfies the condition of \cref{inv:condition-I} in the new tree, since it also satisfied it in $\UpTree{h_\sigma}$.
    This is the case for most of the nodes:
    \smallskip \begin{itemize}
        \item $\Birth{q} = p$ in $\UpTree{h_\sigma}$ and $\Birth{q} = i$ in the new tree, but both $q < p$ and $q < i$.
            The set of nodes in either subtree of $q$ remains unchanged.
        \item For any node on the trail from $u_j$ to $q$ (except $u_j$ and $q$) and any node on the trail from $s^+$ to $q$ (excluding $s^+$ and $q$)
            the set of nodes in either subtree remains unchanged, and so does the $\Birth{\cdot}$ of these nodes.
        \item The same holds for any descendant of $j$ in $\UpTree{h_\sigma}$, for $s^-$ and any descendant of $s^-$.
        \item For any node not in the banana subtree of $q$ the contents of the subtrees and $\Birth{\cdot}$ is unchanged.
    \end{itemize} \smallskip
    Thus, we only need to show that $u_j$, $j$ and $s^+$ do not violate the condition of \cref{inv:condition-I}.
    Furthermore, we only need to consider the case where $u_j \neq q$ and $s^+ \neq q$.
    \medskip \begin{description}
        \item[Node $u_j$:] the banana subtree of $u_j$ is the same in $\UpTree{h_\sigma}$ and in the new tree, as is $\Birth{u_j}$.
            The node $\Dn{u_j} = i_j$ in the new tree is a descendant of $j$ in $\UpTree{h_\sigma}$ and $j = \Dn{u_j}$ in $\UpTree{h_\sigma}$.
            The descendants of $i_j$ are descendants of $\Dn{u_j}$ in both $\UpTree{h_\sigma}$ and the new tree.
            As $u_j$ does not gain any descendants, it follows that $u_j$ satisfies the condition of \cref{inv:condition-I} in the new tree.
        \item[Node $s^+$:] the banana subtree of $s^+$ is the same in $\UpTree{h_\sigma}$ and in the new tree, as is $\Birth{s^+}$.
            Through the change of $\Dn{s^+}$ to $j$, $\Dn{s^+}$ gains new descendants in the new tree.
            These are the nodes on the trail between $p$ and $j$, along with their banana subtrees, and the nodes on the trail from $i$ to $j$.
            By \cref{inv:condition-I} and Conditions {\sf III.1} and {\sf III.2} for $\UpTree{h_\sigma}$ these are all on the same side of $s^+$ as $\Dn{s^+}$ in $\UpTree{h_\sigma}$.
            It follows that $s^+$ satisfies the condition of \cref{inv:condition-I} in the new tree.
        \item[Node $j$:] The two cases $q < j < p$ and $q < p < j$ are symmetric. We give the proof for the first case and assume $q < j < p$.
            The pointers $\Mid{j}$ and $\Dn{j}$ are swapped going from $\UpTree{h_\sigma}$ to the new tree (see the assignment in \cref{ln:j-ptrs}).
            The banana subtree rooted at $j$ in the new tree consists of descendants of $s^+$ and the descendants of $\Dn{j}$ in $\UpTree{h_\sigma}$.
            By \cref{inv:condition-I} and Conditions {\sf III.1} and {\sf III.2} for $\UpTree{h_\sigma}$, each node $v$ in this set satisfies $j < v$.
            Note that $\Birth{j} = p$ in the new tree and thus $j < \Birth{j}$ by assumption.
            The descendants of $\Dn{j}$ in the new tree are descendants of nodes on the mid-trail of $j$ in $\UpTree{h_\sigma}$.
            Each node $t$ in this set satisfies $t < j$ by \cref{inv:condition-I} for $\UpTree{h_\sigma}$.
            Thus, in the new tree $j$ satisfies the condition of \cref{inv:condition-I}.
    \end{description} \medskip
    We have shown that all nodes satisfy the condition of \cref{inv:condition-I} and thus the new up-tree satisfies \cref{inv:condition-I}.

    \smallskip
    It remains to show that the new up-tree also satisfies \cref{inv:condition-II}.
    \Cref{inv:condition-II} can be violated only by minima for which $\Low{\Dth{\cdot}}$ changes.
    This changes for $i$, $p$ and any node $u$ with $\Low{\Dth{\cdot}} = p$ in $\UpTree{h_\sigma}$.
    \smallskip \begin{itemize}
        \item In the new tree $\Low{\Dth{p}} = i$, and $h_\theta(p) > h_\theta(i)$ by definition of $h_\theta$.
        \item As the algorithm sets $\Dth{i} = q$ (see \cref{ln:ass-dth}) in the new tree $\Low{\Dth{i}} = \Low{q}$.
            Since either $h_\sigma(i) = h_\theta(i)$ or $h_\sigma(p) = h_\theta(p)$ and since $i$ and $p$ are contiguous in both $h_\sigma$ and $h_\theta$,
            $h_\sigma(p) > h_\sigma(\Low{q})$ implies that $h_\theta(i) > h_\sigma(\Low{q}) = h_\theta(\Low{q}) = h_\theta(\Low{\Dth{i}})$.
        \item The nodes $u$ for which we set $\Low{\Dth{u}} = i$ have $\Low{\Dth{u}} = p$ in $\UpTree{h_\sigma}$.
            They satisfy $h_\sigma(u) = h_\theta(u) > h_\sigma(\Low{\Dth{p}}$,
            and contiguity of $i$ and $p$ in $h_\sigma$ they satisfy $h_\theta(u) > h_\sigma(i)$.
            As $h_\theta(i) \leq h_\sigma(i)$, it follows that $h_\theta(u) > h_\theta(i) = h_\theta(\Low{\Dth{u}}$.
    \end{itemize} \smallskip
    All minima satisfy the condition of \cref{inv:condition-II} and thus the new tree satisfies \cref{inv:condition-II}.
    We have shown that the new tree satisfies \cref{inv:condition-I,inv:condition-II,inv:condition-III-max} for $h_\theta$,
    which concludes the proof that it is the unique up-tree for $h_\theta$.
\end{proof}

\subsection{Correctness of Scenarios~A and B.}
\label{sec:scenario-correctness}
So far, we have described the local operations in terms of small changes in function value, where the involved items are contiguous in function value before and after the operation.
However, Scenarios~A and B described in \cref{sec:4.2} also need to deal with larger changes in value, involving not necessarily contiguous items.
We start by describing the conditions under which a change in the value of a maximum (or a minimum by \cref{lem:coupling_of_interchanges}) does not affect the corresponding banana tree.
This is stated formally in the following lemma.

\begin{lemma}[Unaffected Banana Tree]
    \label{lem:change-of-max-is-local}
    Let $j$ be a maximum.
    Let $h_\sigma$ and $h_\theta$ be maps that differ only in the value of $j$ such that all other items have the same criticality.
    If it holds that $\max\{h_\sigma(\In{j}), h_\sigma(\Mid{j}), h_\sigma(\Dn{j})\} < h_\theta(j) < h_\sigma(\Up{j})$ then $\UpTree{h_\sigma}$ and $\UpTree{h_\theta}$ are identical.
\end{lemma}
\begin{proof}
    By \cref{inv:condition-I,inv:condition-II,inv:condition-III-max}, we know that  $\max\{h_\sigma(\In{j}), h_\sigma(\Mid{j}), h_\sigma(\Dn{j})\} < h_\sigma(j) < h_\sigma(\Up{j})$.
    The up-tree $\UpTree{h_\sigma}$ satisfies \cref{inv:condition-I} for the map $h_\theta$, since the order of items along the interval is independent of $h_\sigma$ and $h_\theta$.
    It also satisfies \cref{inv:condition-II} for $h_\theta$, as $h_\sigma(u) = h_\theta(u)$ for all minima $u$ in $h_\sigma$ and $h_\theta$,
    and the set of minima is the same in $h_\sigma$ and $h_\theta$.
    Finally, it also satisfies \cref{inv:condition-III-max}, which follows from $\max\{h_\sigma(\In{j}), h_\sigma(\Mid{j}), h_\sigma(\Dn{j})\} < h_\theta(j) < h_\sigma(\Up{j})$ and the fact that the order of items along the interval is independent of $h_\sigma$ and $h_\theta$.
    By \cref{cor:unique-bananas-local} there is a unique banana tree for $h_\theta$ that satisfies \cref{inv:condition-I,inv:condition-II,inv:condition-III-max},
    and thus $\UpTree{h_\sigma} = \UpTree{h_\theta}$.
\end{proof}

We now show that the algorithms given for Scenarios~A and B indeed transform $\UpTree{f}$ into $\UpTree{g}$, where the maps $f$ and $g$ differ in the value of item $j$.
Recall that in Scenario~A item $j$ is non-critical in $f$ and $f(j) < g(j)$,
and that in Scenario~B item $j$ is a maximum in $f$ and decreases its value until it is non-critical in $g$.
For Scenarios~A and B we assume that $j$ is not an endpoint, i.e., $j \in [2,m-1]$ and we give an additional algorithm for updating the value of an endpoint below.

We restate the algorithm for Scenario~A, where we now also account for updates to items that lead to an endpoint changing from down-type to up-type (see \cref{ln:A-endpoint}).
\begin{tabbing}
  m\=m\=m\=m\=m\=m\=m\=m\=m\=\kill
  \clreset
  \clbegin                   \> \> \texttt{if} $g(j) > f(j+1)$ \texttt{then} \\*
  \clbeginlbl{ln:A-endpoint} \> \> \> \texttt{if} $j+2$ is a hook \texttt{then} cancel or slide endpoint $j+1$ to $j$ in $\UpTree{f}$ and $\DnTree{f}$ \\*
  \clbeginlbl{ln:A-ac-slide} \> \> \> \texttt{else if} $f(j+1) < f(j+2)$ \= \texttt{then} anti-cancel $j$ and $j+1$ in $\UpTree{f}$ and $\DnTree{f}$ \\*
  \clbegin                   \> \> \>                                    \> \texttt{else} slide $j+1$ to $j$ in $\UpTree{f}$ and $\DnTree{f}$ \\*
  \clbegin                   \> \> \> \texttt{endif}; \\*
  \clbegin                   \> \> \> set $q = \Up{j}$ in $\UpTree{f}$; \\*
  \clbeginlbl{ln:A-while}    \> \> \> \texttt{while} $f(q) < g(j)$ \texttt{do} \=
                              interchange $j$ and $q$ in $\UpTree{f}$ and $\DnTree{f}$; \\*
  \clbegin                   \> \> \> \> $q = \Up{j}$ in $\UpTree{f}$ \\*
  \clbegin                   \> \> \> \texttt{endwhile}; \\
  \clbegin                   \> \> \texttt{endif}.
\end{tabbing}
The following lemma states that this algorithm correctly updates $\UpTree{f}$ to obtain $\UpTree{g}$.
\begin{lemma}[Correctness of Scenario A]
    \label{lem:scenario-A}
    Let $f$ and $g$ be maps with $f(j) < g(j)$ and $f(i) = g(i)$ for all $i \neq j$,
    and $f(j-1) < f(j) < f(j+1)$.
    Applying the algorithm for Scenario A to $\UpTree{f}$ yields $\UpTree{g}$.
\end{lemma}
\begin{proof}
    We show how the up-tree $\UpTree{f}$ is modified throughout the algorithm following the homotopy $h_\lambda$ until the up-tree becomes $\UpTree{g}$.
    If $g(j) < f(j+1)$ then item $j$ is non-critical in $f$ and $g$, $\UpTree{f} = \UpTree{g}$, and the algorithm does not change the up-tree.
    Assume $g(j) > f(j+1)$ for the remainder of the proof.
    By this assumption there exists an $h_{\theta_0}$ with $0 < \theta_0 \leq 1$ such that $j$ and $j+1$ are contiguous in $h_{\theta_0}$.
    If $j+2$ is a hook, then $j+1$ is down-type in $f$ and up-type in $h_{\theta_0}$ and by \cref{lem:slide-cancel-endpoint} applying $\algcancelendpoint(j,j+1)$ yields $\UpTree{h_{\theta_0}}$.
    If $j+2$ is not a hook and $f(j+1) < f(j+2)$, then $j$ and $j+1$ are non-critical in $f$ and a maximum and minimum, respectively, in $h_{\theta_0}$.
    By \cref{lem:anticancel} executing an anti-cancellation between $j$ and $j+1$ yields the up-tree $\UpTree{h_{\theta_0}}$.
    If $j+2$ is not a hook and $f(j+1) > f(j+2)$, then $j+1$ is a maximum in $f$, non-critical in $h_{\theta_0}$ and $j$ is a maximum in $h_{\theta_0}$.
    By \cref{lem:max-slide} executing a slide from $j+1$ to $j$ yields the up-tree $\UpTree{h_{\theta_0}}$.

    Let $q_i$ be $\Up{j}$ at the beginning of the $i$-th iteration of the loop in \cref{ln:A-while}.
    We show by induction that after every iteration $i$ of the loop either (1) the current up-tree corresponds to a map $h_{\theta_{i+1}}$ where $q_i$ and $j$ are contiguous
    and $h_{\theta_{i+1}}(q_i) < h_{\theta_{i+1}}(j)$
    or (2) the current up-tree is $\UpTree{g}$.
    For the base case $i = 1$, if $f(q_1) > g(j)$, then the up-tree $\UpTree{h_{\theta_0}} = \UpTree{g}$ by \cref{lem:change-of-max-is-local}.
    Otherwise, there exists a $\theta_1 \in (\theta_0, 1]$ such that $h_{\theta_1}(q_1) < h_{\theta_1}(j)$ and $q_1$ and $j$ are contiguous in $h_{\theta_1}$.
    By \cref{lem:change-of-max-is-local} and \cref{lem:interchange-of-maxima} interchanging $q_1$ and $j$ yields $\UpTree{h_{\theta_1}}$.
    Now assume that at the beginning of some iteration $i$ the current up-tree is $\UpTree{h_{\theta_{i-1}}}$, with $h_{\theta_{i-1}}(j) \in [\max \{ f(\In{j}), f(\Mid{j}), f(\Dn{j}) \}, f(q_i)]$.
    If $f(q_i) > g(j)$, then the up-tree $\UpTree{h_{\theta_i}} = \UpTree{g}$ by \cref{lem:change-of-max-is-local}.
    Else, there exists a $\theta_{i+1} \in (\theta_i, 1]$ such that $h_{\theta_{i+1}}(q_i) < h_{\theta_{i+1}}(j)$ and $q_i$ and $j$ are contiguous in $h_{\theta_{i+1}}$.
    Again, interchanging $q_i$ and $j$ yields the up-tree $\UpTree{h_{\theta_{i+1}}}$.

    There exists an iteration $I$ where $f(q_I) > g(j)$, at the latest when $q_I$ is the special root, and the loop terminates.
    Then, by \cref{lem:change-of-max-is-local} the $\UpTree{h_{\theta_{I-1}}} = \UpTree{g}$.
\end{proof}

The following algorithm updates the banana trees in Scenario~B, now including updates that change up-type items to down-type items (see \cref{ln:B-endpoint}).
\begin{tabbing}
  m\=m\=m\=m\=m\=m\=m\=m\=m\=\kill
  \clreset
  \clbegin                    \> \> \texttt{loop} $q = \arg\max\{f(\Dn{j}), f(\In{j}), f(\Mid{j})\}$ in $\UpTree{f}$; \\*
  \clbegin                    \> \> \> \texttt{if} $f(q) > f(j+1)$ \= \texttt{then}
                              interchange $q$ and $j$ in $\UpTree{f}$ and $\DnTree{f}$ \\*
  \clbegin                    \> \> \>                             \> \texttt{else} \texttt{exit} \texttt{endif} \\*
  \clbegin                    \> \> \texttt{forever}; \\
  \clbeginlbl{ln:B-endpoint}  \> \> \texttt{if} $j+2$ is a hook \texttt{then} cancel or slide endpoint $j+1$ to $j$ in $\UpTree{f}$ and $\DnTree{f}$ \\*
  \clbegin                    \> \> \texttt{else if} $f(j+1) < f(j+2)$ \= \texttt{then}
                                     cancel $j$ with $j+1$ in $\UpTree{f}$ and $\DnTree{f}$ \\*
  \clbegin                    \> \>                               \> \texttt{else}
                                     slide $j$ to $j+1$ in $\UpTree{f}$ and $\DnTree{f}$ \\*
  \clbegin                    \> \> \texttt{endif}.
\end{tabbing}
The next lemma states that this algorithm correctly updates $\UpTree{f}$ to obtain $\UpTree{g}$.
\begin{lemma}[Correctness of Scenario B]
    \label{lem:scenario-B}
    Let $f$ and $g$ be maps with $f(j) > g(j)$ and $f(i) = g(i)$ for all $i \neq j$,
    and $f(j-1) < f(j) > f(j+1)$.
    Applying the algorithm for Scenario B to $\UpTree{f}$ yields $\UpTree{g}$.
\end{lemma}
\begin{proof}
    The proof is analogous to that of \cref{lem:scenario-A}:
    after each iteration $i$ of the loop, there is a $\theta_i$ such that $q_i$ and $j$ are contiguous in $h_{\theta_i}$ and the current up-tree is $\UpTree{h_{\theta_i}}$.
    Let $\UpTree{h_{\gamma}}$ be the up-tree after the loop,
    where in $h_\gamma$ the items $j$ and $q$ are contiguous in value.
    If $j+2$ is a hook, then $j+1$ is an up-type endpoint in $h_\gamma$ and a down-type endpoint in $g$, since $g(j) < g(j+1)$.
    Then, by \cref{lem:slide-cancel-endpoint} and the fact that non-critical items are not represented in the up-tree applying $\algcancelendpoint(j,j+1)$ to $\UpTree{h_\gamma}$ yields $\UpTree{g}$.
    Similarly, if $j+2$ is not a hook and $f(j+1) < f(j+2)$, then $j$ and $j+1$ are non-critical in $g$ as $f(j-1) < g(j) < f(j+1) < f(j+2)$, and canceling $j$ with $j+1$ yields $\UpTree{g}$ by \cref{lem:cancellation}.
    Finally, if $j+2$ is not a hook and $f(j+1) > f(j+2)$, then $j+1$ was non-critical in $f$ and is a maximum in $g$, and sliding $j$ to $j+1$ yields $\UpTree{g}$ by \cref{lem:max-slide}.
\end{proof}

To conclude, we state an algorithm for updating the banana trees under the change of value of the endpoint $m$.
The algorithm for the other endpoint of the interval is symmetric.
We assume that $m$ is up-type in $f$ and down-type in $g$, i.e., $f(m) < f(m-1) < g(m)$.
There is also the reverse case, in which $m$ is down-type in $f$ and up-type in $g$.
\begin{tabbing}
  m\=m\=m\=m\=m\=m\=m\=m\=m\=\kill
  \clreset
  \clbegin                   \> \> Set $q = \Dn{m}$ in $\DnTree{f}$; \\*
  \clbegin                   \> \> \texttt{while} $q$ is a maximum in $-f$ \texttt{and} $f(q) < g(m)$ \texttt{do} \\*
                             \> \>      \> interchange $m$ and $q$ in $\UpTree{f}$ and $\DnTree{f}$; \\*
  \clbegin                   \> \>      \> $q = \Dn{m}$ in $\DnTree{f}$ \\*
  \clbegin                   \> \> \texttt{endwhile}; \\*
  \clbegin                   \> \> \texttt{if} $f(m-1) < g(m)$ \texttt{then} cancel or slide endpoint $m$ to $m-1$ \texttt{endif}; \\*
  \clbegin                   \> \> Set $q = \Up{m}$ in $\UpTree{f}$; \\*
  \clbegin                   \> \> \texttt{while} $f(q) < g(m)$ \texttt{do}
                                        \= interchange $m$ and $q$ in $\UpTree{f}$ and $\DnTree{f}$; \\*
  \clbegin                   \> \>      \> $q = \Up{m}$ in $\UpTree{f}$ \\*
  \clbegin                   \> \> \texttt{endwhile}.
\end{tabbing}
In the first loop we exploit the coupling of interchanges between the up-tree and down-tree:
we find the maximum with which to interchange $m$ in $\DnTree{f}$,
and perform the corresponding interchange of minima in $\UpTree{f}$.
The purpose of this loop is to prepare the banana tree for switching $m$ from an up-type to a down-type item, as described in \cref{sec:slide-cancel-endpoint}.
The value of the hook $m+1$ is defined to have value $f(m+1) = m + \varepsilon$ and $g(m+1) = g(m) - \varepsilon$, i.e., its value is adjusted alongside that of $m$,
which justifies interchanging $m$ only with $\Dn{m}$ in the first loop, rather than with $\In{m}$, $\Mid{m}$ and $\Dn{m}$ as we have done in Scenario~B.
For completeness we state in the next lemma that this algorithm correctly maintains the up-tree; we omit the proof, which is analogous to that of \cref{lem:scenario-A,lem:scenario-B}.
\begin{lemma}[Local Change at an Endpoint]
    Let $f$ and $g$ be maps differing only in the value of the endpoint $m$, such that $f(m) < f(m-1) < g(m)$.
    Applying the previous algorithm to $\UpTree{f}$ yields $\UpTree{g}$.
\end{lemma}

\section{Correctness of Topological Maintenance}
\label{sec:topo-correctness}

\subsection{Splitting  Banana Trees.}
\label{sec:split-correctness}
In this section we prove that the algorithm \textsc{Split} described in \cref{sec:4.3} correctly splits a banana tree.
The proof will be by induction over the iterations of the loop in \textsc{Split},
with the base case in \cref{lem:split-iteration-1} and the induction step in \cref{lem:split-iteration-i}.
We first prove these two lemmas and then state the result in \cref{thm:correctness-of-split}.

Given a list of items, its associated map $f$ and an item $\ell$, the algorithm \texttt{Split} computes from $\UpTree{f}$ two new up-trees $\UpTree{g}$ and $\UpTree{h}$, where the former contains the items up to $\ell$ and the latter the remaining items.
We define $x$ to be the midpoint of $\ell$ and $\ell + 1$ with $f(x) = \frac{1}{2} (f(\ell) + f(\ell + 1))$.
We write $\Stack_i$, $p_i$, $q_i$ for $\Stack$, $p$, $q$ returned by \textsc{TopBanana} in the $i$-th iteration of \textsc{Split}.
The up-trees $\UpTree{g_i}$ and $\UpTree{h_i}$ denote the left and right up-trees after the $i$-th iteration,
with initial up-trees $\UpTree{g_0}$, $\UpTree{h_0}$.
One of the initial up-trees consists of the new special root and the dummy leaf $\alpha$.
We assume that this is the tree $\UpTree{h_0}$ with special root $\beta_h$ and that $\UpTree{g_0} = \UpTree{f}$.
This is equivalent to assuming that the top banana on the stacks is on the right spine of $\UpTree{f}$,
where the top banana is the banana closest to the root such that $x$ lies in some panel of the corresponding window.
The other case where the top banana is on the left spine of $\UpTree{f}$ is symmetric.
We further define $\alpha$ such that it is the rightmost item of $\UpTree{g_i}$ if it is in $\UpTree{g_i}$ and the leftmost item of $\UpTree{h_i}$ if it is in $\UpTree{h_i}$.
Finally, let $I$ be the total number of bananas on the stack, i.e., the number of iterations of \textsc{Split}.

\begin{lemma}[Base Case]
    \label{lem:split-iteration-1}
    At the beginning of the first iteration it holds that
    \begin{enumerate}
        \item for all nodes $u$ in $\UpTree{g_0}$ and all nodes $v$ in $\UpTree{h_0}$ we have $u < v$, \label{itm:split-base-clm-1}
        \item $\UpTree{g_0}$ and $\UpTree{h_0}$ satisfy \cref{inv:condition-I,inv:condition-II,inv:condition-III-max}, \label{itm:split-base-clm-2}
        \item any node $u$ that is not on any stack and does not have an ancestor on any stack is in $\UpTree{g_0}$ if $u < x$ and in $\UpTree{h_0}$ if $u > x$, \label{itm:split-base-clm-3}
        \item $f(\alpha) < f(x)$, $f(q_j) < f(\In{\alpha})
        $ and $f(q_j) < f(\Mid{\alpha})$ at the beginning of the iteration for all $j \in [1,I]$ with $(p_j,q_j) \in \Lup \cup \Rup \cup \Mup$, \label{itm:split-base-clm-4}
        \item all bananas $(p_j,q_j) \in \Lup \cup \Rup \cup \Mup$ for $j \in [1,I]$ are in $\UpTree{g_0}$ if $\alpha$ is in $\UpTree{h_0}$ and in $\UpTree{h_1}$ if $\alpha$ is in $\UpTree{g_1}$, \label{itm:split-base-clm-5}
        \item if $\Stack_1 \in \{\Ldn, \Rdn\}$ then $q_1$ is in the same tree as $\alpha$, \label{itm:split-base-clm-6}
        \item the in-trail between $\alpha$ and $\Dth{\alpha}$ is empty and $\alpha$ is on the spine. \label{itm:split-base-clm-7}
    \end{enumerate}
\end{lemma}
\begin{proof}
    Assume that $\UpTree{g_0} = \UpTree{f}$.
    Claim \ref{itm:split-base-clm-1} holds by definition of $\alpha$ and the special root of $\UpTree{h_0}$.
    The up-tree $\UpTree{f} = \UpTree{g_0}$ is the input to the algorithm and satisfies \cref{inv:condition-I,inv:condition-II,inv:condition-III-max} by assumption.
    The topmost banana $(p_1,q_1)$ has $q_1$ on the right spine of $\UpTree{g_0}$ and $q_1 < x$, which implies that nodes $u$ that are not descendants of $q_1$ satisfy $u < q_1 < x$.
    These are precisely the nodes in $\UpTree{g_0}$ that have no ancestor on any stack.
    The nodes in $\UpTree{h_0}$ also have no ancestor on any stack, but $\alpha > x$ by definition, and so is the special root of $\UpTree{h_0}$.
    This implies claim \ref{itm:split-base-clm-3}.
    The inequality $f(\alpha) < f(x)$ follows from $f(\alpha) < f(p_1) < f(x)$, where the last inequality holds since $x$ is in the in-panel or mid-panel of $\Window{p_1}{q_1}$.
    The pointers $\In{\alpha}$ and $\Mid{\alpha}$ both point to the special root of $\UpTree{h_0}$, which by definition has greater value than all $q_j$ for $j \in [1,I]$ with $q_j \in \Lup \cup \Rup \cup \Mup$.
    Thus, all three inequalities of claim \ref{itm:split-base-clm-4} hold.
    To see that claim \ref{itm:split-base-clm-5} holds simply note that all nodes on the stacks are in $\UpTree{g_0}$ and $\alpha$ is in $\UpTree{h_0}$.
    Claim \ref{itm:split-base-clm-6} holds trivially, as the topmost banana $(p_1,q_1)$ cannot be on $\Ldn$ or $\Rdn$:
    this would put $x$ in the in-panel or mid-panel of the banana that $(p_1,q_1)$ is nested in, which contradict $(p_1,q_1)$ being the topmost banana.
    Finally, $\alpha$ is the only other node in $\UpTree{h_0}$ besides the special root, and thus both trails between $\alpha$ and $\Dth{\alpha}$ are empty and $\alpha$ is on the spine.
\end{proof}

\begin{lemma}[Inductive Step]
    \label{lem:split-iteration-i}
    If at the beginning of any iteration $i$ of \textsc{Split} it holds that
    \smallskip \begin{enumerate}
        \item for all nodes $u$ in $\UpTree{g_{i-1}}$ and all nodes $v$ in $\UpTree{h_{i-1}}$ we have $u < v$, \label{itm:split-ass-1}
        \item $\UpTree{g_{i-1}}$ and $\UpTree{h_{i-1}}$ satisfy \cref{inv:condition-I,inv:condition-II,inv:condition-III-max}, \label{itm:split-ass-2}
        \item any node $u$ that is not on any stack and does not have an ancestor on any stack is in $g_{i-1}$ if $u < x$ and in $h_{i-1}$ if $u > x$, \label{itm:split-ass-3}
        \item $f(\alpha) < f(x)$, $f(q_j) < f(\In{\alpha}$ and $f(q_j) < f(\Mid{\alpha})$ at the beginning of the iteration for all $j \in [i,I]$ with $(p_j,q_j) \in \Lup \cup \Rup \cup \Mup$, \label{itm:split-ass-4}
        \item all bananas $(p_j,q_j) \in \Lup \cup \Rup \cup \Mup$ for $j \in [i,I]$ are in $\UpTree{g_i}$ if $\alpha$ is in $\UpTree{h_i}$ and in $\UpTree{h_i}$ if $\alpha$ is in $\UpTree{g_i}$, \label{itm:split-ass-5}
        \item if $\Stack_i \in \{\Ldn, \Rdn\}$ then $q_i$ is in the same tree as $\alpha$, \label{itm:split-ass-6}
        \item the in-trail between $\alpha$ and $\Dth{\alpha}$ is empty and $\alpha$ is on the spine, \label{itm:split-ass-7}
    \end{enumerate} \smallskip
    then after the $i$-th iteration
    \smallskip \begin{enumerate}
        \item for all nodes $u$ in $\UpTree{g_i}$ and all nodes $v$ in $\UpTree{h_i}$ we have $u < v$, \label{itm:split-clm-1}
        \item $\UpTree{g_i}$ and $\UpTree{h_i}$ satisfy \cref{inv:condition-I,inv:condition-II,inv:condition-III-max}, \label{itm:split-clm-2}
        \item any node $u$ that is not on any stack and does not have an ancestor on any stack is in $g_{i}$ if $u < x$ and in $h_i$ if $u > x$, \label{itm:split-clm-3}
        \item $f(\alpha) < f(x)$, $f(q_j) < f(\In{\alpha})$ and $f(q_j) < f(\Mid{\alpha})$ for all $j \in [i+1,I]$ with $(p_j,q_j) \in \Lup \cup \Rup \cup \Mup$, \label{itm:split-clm-4}
        \item all bananas $(p_j,q_j) \in \Lup \cup \Rup \cup \Mup$ for $j \in [i+1,I]$ are in $\UpTree{g_i}$ if $\alpha$ is in $\UpTree{h_i}$ and in $\UpTree{h_i}$ if $\alpha$ is in $\UpTree{g_i}$, \label{itm:split-clm-5}
        \item if \textsc{TopBanana} returns $\Stack_{i+1} \in \{\Ldn, \Rdn\}$ then $q_{i+1}$ is in the same tree as $\alpha$, \label{itm:split-clm-6}
        \item the in-trail between $\alpha$ and $\Dth{\alpha}$ is empty and $\alpha$ is on the spine. \label{itm:split-clm-7}
    \end{enumerate}
\end{lemma}
\begin{proof}
    There are in total six cases:
    \begin{enumerate}
        \item $\Stack = \Lup$,
        \item $\Stack = \Rup$,
        \item $\Stack = \Mup$ and $q < x < p$,
        \item $\Stack = \Mup$ and $p < x < q$,
        \item $\Stack = \Ldn$,
        \item $\Stack = \Rdn$.
    \end{enumerate}
    The first two cases are symmetric, as are the two cases with $\Stack = \Mup$ and the last two cases.
    We prove the claim for cases 1, 3 and 5.
    In the following we assume that the seven conditions hold.
    \begin{description}
        \item[$\Stack_i = \Lup$:]
            \textsc{Split} executes \textsc{DoInjury}. By definition of $\Lup$ $q_i < p_i < x$.
            The node $q_i$ is on a spine, as otherwise the banana it is nested in would contain $x$ in its in- or mid-panel.
            This is not possible, since this banana would then be on the stacks above $(p_i,q_i)$.
            Furthermore, $q_i < p_i$ implies that $q_i$ is on a right spine and it is thus in the left tree.
            The fourth claim follows immediately from the fourth assumption, as the value of $\alpha$ remains unchanged and $\Lup$, $\Rup$, $\Mup$ can be merged into a sorted stack by \cref{lem:sorted_stacks}.

            Nodes $j$ with $j > x$ on the in-trail between $p_i$ and $q_i$ are removed and inserted between $\alpha$ and $\Mid{\alpha}$.
            These nodes and the nodes in their banana subtrees are the rightmost nodes of $\UpTree{g_i}$, which implies claim \ref{itm:split-clm-1}.
            Since we do not modify the in-trail between $\alpha$ and $\Dth{\alpha}$ and it is assumed to be empty at the beginning of the iteration,
            this trail is also empty after the iteration, which proves claim \ref{itm:split-clm-7}.
            
            Denote by $j^-$ the lowest and highest among the moved nodes in terms of function value, respectively.
            Conditions {\sf III.1} and {\sf III.2} hold in $\UpTree{g_{i-1}}$ and $\UpTree{h_{i-1}}$, and thus the nodes $j$ form a contiguous section of the in-trail with $j^+ = \In{q_i}$.
            Write $a$ for the highest node on the in-trail between $p_i$ and $q_i$ that is not moved.
            Both the in-trail originally containing the nodes $j$ and the mid-trail between $\alpha$ and $\Dth{\alpha}$ are right trails.
            By the assumption that nodes in $\UpTree{g_i}$ are less than nodes in $\UpTree{h_i}$ it follows that $\alpha = \Dn{j^-} < j^- < \dots < j^+ < \Up{j^+}$.
            The item $x$ is in the in-panel of the window $\Window{p_i}{q_i}$ corresponding to the banana $(p_i,q_i)$ in $\UpTree{g_{i-1}}$, and by Conditions {\sf III.1}, {\sf III.2} and \cref{inv:condition-III-max} $x < j^-$ also implies $f(x) < f(j^-)$.
            This, together with the fourth assumption implies that $f(\alpha) = f(\Dn{j^-}) < f(x) < f(j^-) < \dots < f(j^+) < f(\Up{j^+})$.
            Thus, \cref{inv:condition-III-max} holds in the new trees.

            Since nodes pointed to by $\Low{\Dth{\cdot}}$ only change for $\Birth{j}$ where $j$ is one of the moved nodes, namely from $p_i$ to $\alpha$,
            \cref{inv:condition-II} holds in the new trees by the assumption that $f(\alpha) < f(p_i)$.

            \Cref{inv:condition-I} continues to hold in the tree that originally contained the nodes $j$, since nodes are only removed from subtrees.
            The nodes $j$ themselves only gain $\alpha$ as a descendant of $\Dn{j}$ or as $\Dn{j}$ itself.
            In addition, $\alpha < j$ by definition of $\alpha$ and $j < \Birth{j}$ by definition of right trail.
            Thus, the nodes $j$ also satisfy the condition of \cref{inv:condition-I}.
            The nodes $j$ are inserted as descendants of $\Dn{t}$, where $t = \Mid{\alpha}$ before the insertion.
            Since $j < \Dn{t}$ by assumption \cref{itm:split-ass-1}, the condition of \cref{inv:condition-I} continues to holds for $t$,
            and also holds for the remaining nodes of this tree.
            It follows that both trees satisfy \cref{inv:condition-I} after the $i$-th iteration.

            The nodes $u$ in $\UpTree{h_{i-1}}$ that are not on any stack and do not have an ancestor on any stack satisfy $u > x$.
            The nodes $j$ and the nodes in their banana subtrees are also greater than $x$.
            Thus, $\UpTree{h_{i-1}}$ contains no node $u$ that is not on any stack and has no ancestor on any stack with $u < x$.
            In the left tree $\UpTree{g_i}$ the node $a$ and the nodes in its subtree are the rightmost nodes of $a$, since $a$ is on the right spine.
            Other nodes in $\UpTree{g_i}$ are thus less than $x$. Note that $a < x$ by definition of $a$.
            The node $a$ is either a minimum, a maximum and not on any stack, or a maximum and on some stack.
            In the first two cases all nodes in $\UpTree{g_i}$ are less than $x$.
            In the third case, all nodes except possibly descendants of $a$ are less than $x$, and descendants of $a$ have an ancestor on the stack.
            Thus, claim \ref{itm:split-clm-3} holds.
            
            To see that claim \ref{itm:split-clm-5} holds, recall that the bananas in $\Lup$, $\Rup$ and $\Mup$ are nested into each other.
            Call the topmost banana on these stacks $(p_\ell,q_\ell)$.
            It must be nested into $(p_i,q_i)$ and is thus on the in-trail between $p_i$ and $q_i$ in $\UpTree{g_i}$, as otherwise $x$ would be in the mid-panel of $\Window{p_i}{q_i}$, which contradicts $\Stack_i = \Lup$.
            This implies $q_i < p_i$ and $q_i < x$.
            It follows that $q_i$ remains in the left tree, while $\alpha$ remains in the right tree.
            This proves claim \ref{itm:split-clm-5}.
            
            We now prove claim \ref{itm:split-clm-6}.
            Assume that $\Stack_{i+1} = \Ldn$. Then, by definition $p_{i+1} < q_{i+1}$.
            Thus, $q_{i+1}$ cannot be on the in-trail between $p_i$ and $q_i$.
            It can also not be on any other trail, since then $q_{i+1} < p_i < x$, which implies that $x$ is not in the out-panel of $\Window{p_{i+1}}{q_{i+1}}$.
            It follows that $p_{i+1}$ and $q_{i+1}$ cannot be in the left tree.
            Now assume that $\Stack_{i+1} = \Rdn$. Then, by definition $q_{i+1} < p_{i+1}$.
            If $q_{i+1}$ is in $\UpTree{g_i}$, then $q_{i+1} < x$, as all maxima greater than $x$ have been moved to $\UpTree{h_i}$.
            This contradicts $\Window{p_{i+1}}{q_{i+1}}$ having $x$ in its out-panel, which implies that $q_{i+1}$ must be in the right tree.
            In both cases, $q_{i+1}$ is in the right tree, which is the same tree as $\alpha$.
               
        \item[$\Stack_i = \Mup$ and $q_i < x < p_i$:]
            \textsc{Split} executes \textsc{DoFatality}. By definition of $\Mup$ $q_i < x < p_i$.
            This case is similar to the case $\Stack_i = \Lup$ and we mainly point out the differences.
            As above, $q_i$ is on the left spine of $\UpTree{g_{i-1}}$.
            Since $x$ is in a panel of the window $(p_i,q_i)$ it holds that $f(p_i) < f(x)$.
            The fourth claim then follows immediately from the third assumption, as $f(\alpha) < f(p_i)$ after the iteration and $\Lup$, $\Rup$, $\Mup$ can be merged into a sorted stack by \cref{lem:sorted_stacks}.

            All nodes internal to the in-trail between $p_i$ and $q_i$ are moved to the right tree, as are the nodes $j$ with $j > x$ on the mid-trail between $p_i$ and $q_i$.
            The node $\alpha$ replaces $p_i$ as $\Birth{q_i}$.
            As $q_i$ is on a spine the nodes that are moved to the right tree are the rightmost nodes of $\UpTree{g_{i-1}}$,
            and together with assumption \ref{itm:split-ass-1} this implies claim \ref{itm:split-clm-1}, i.e., that the nodes in $\UpTree{h_i}$ are all greater than the nodes in $\UpTree{g_i}$.
            The in-trail between $\alpha$ and $\Dth{\alpha}$ in $\UpTree{g_i}$ is the in-trail between $\alpha$ and $q_i$, which is empty as all nodes on the in-trail beginning at $q_i$ have been moved to $\UpTree{h_i}$.

            \Cref{inv:condition-I,inv:condition-III-max} hold in $\UpTree{g_i}$ as $f(\alpha) < f(p_i) <  f(q_i)$ by assumption \ref{itm:split-ass-4} and because the order of maxima and nodes in their subtrees is unchanged.
            By the assumption that $f(q_i) < f(\In{\alpha}))$ and $f(q_i) < f(\Mid{\alpha})$, where $\In{\alpha}$ and $\Mid{\alpha}$ are the pointers in $\UpTree{h_{i-1}}$,
            the nodes $j$ that are moved to the right tree satisfy $f(j) < f(\Up{j})$.
            Those nodes $j$ that are on a left trail in $\UpTree{g_{i-1}}$ are on a left trail in $\UpTree{h_i}$, and those on a right trail are also on a right trail in $\UpTree{h_i}$.
            It follows that $\UpTree{h_i}$ also satisfies \cref{inv:condition-I,inv:condition-III-max}.
            In $\UpTree{g_i}$, only for nodes with $\Low{\Dth{\cdot}} = p_i$ does this pointer change, namely to $\alpha$, and $f(\alpha)$ is set to $f(p_i) - \varepsilon$.
            These nodes satisfy the condition of \cref{inv:condition-II}, as does $\alpha$ itself.
            In $\UpTree{h_i}$ the node $p_i$ satisfies the condition of \cref{inv:condition-II}, as $f(p_i) > f(\alpha)$.
            The nodes that are moved to the right tree also satisfy this condition, since $\Low{\Dth{\cdot}}$ remains unchanged.
            Other nodes $u$ with $\Low{\Dth{u}} = p_i$ in $\UpTree{g_i}$ were already on the mid-trail between $\alpha$ and $\Dth{\alpha}$ in $\UpTree{h_{i-1}}$.
            If $f(\Low{\Dth{u}}) < f(p_i)$ then the windows $\Window{u}{\Dth{u}}$ would have $x$ in their out-panel, which places them in $\Ldn$ or $\Rdn$.
            But by assumption \ref{itm:split-ass-4} $f(\Dth{u}) > f(q_i)$ which implies that they would have been processed in an earlier iteration.
            This is a contradiction and thus $f(\Low{\Dth{u}}) > f(p_i)$, i.e., they satisfy the condition of \cref{inv:condition-II}.
            For no other node of $\UpTree{h_i}$ does the node $\Low{\Dth{\cdot}}$ change from what it was in $\UpTree{h_{i-1}}$,
            so \cref{inv:condition-II} holds in $\UpTree{h_i}$.
            This proves claim \ref{itm:split-clm-2}.
    
            For claim \ref{itm:split-clm-3} observe that nodes $u$ in $\UpTree{g_i}$ satisfy $u < x$, as any nodes with $u > x$ are moved to $h_i$.
            Nodes $v$ with $v < x$ in $\UpTree{h_i}$ have as ancestor the highest node that is moved from the mid-trail;
            if such nodes $v$ exist, then this highest node itself must be in either $\Rup$ or $\Mup$, and thus the only nodes in $\UpTree{h_i}$ with $v < x$ have an ancestor on some stack.
            This highest node can also be the only next node in $\Lup \cup \Rup \cup \Mup$, so claim \ref{itm:split-clm-5} follows.
            
            To see that claim \ref{itm:split-clm-6} holds note that the next banana with $x$ in the out-panel of the corresponding window can only be on the mid-trail between $\alpha$ and $q_i$ in $\UpTree{g_i}$,
            similar to the case $\Stack_i = \Lup$.
        \item[$\Stack_i = \Ldn$:]
            \textsc{Split} executes \textsc{DoScare}, which sets $f(\alpha) = f(p_i) + \varepsilon$ and then executes an interchange of minima between $\alpha$ and $p_i$.
            The nodes $q_i$ and $\alpha$ are in the same tree by assumption, and they must be in $\UpTree{g_{i-1}}$, as $x$ is in the out-panel of $\Window{p_i}{q_i}$ and the in-trail between $\alpha$ and $\Dth{\alpha}$ is assumed to be empty.
            Going from $\UpTree{g_{i-1}}$ to $\UpTree{g_i}$ and from $\UpTree{h_{i-1}}$ to $\UpTree{h_i}$ does not change the set of nodes contained in each tree.
            In fact, $\UpTree{h_{i-1}} = \UpTree{h_i}$. 
            Claims \ref{itm:split-clm-1}  and \ref{itm:split-clm-5} follow immediately from the respective assumption.
            The set of nodes internal to the in-trail between $\alpha$ and $\Dth{\alpha}$ after the interchange of minima is a subset of the nodes internal to the in-trail between $\alpha$ and $\Dth{\alpha}$ before the iteration.
            The latter is empty by assumption \ref{itm:split-ass-7}.
            Furthermore, $\alpha$ is on the spine before iteration $i$ by \ref{itm:split-ass-7} and it remains on the spine through the interchange, so the claim \ref{itm:split-clm-7} holds.

            We now show claim \ref{itm:split-clm-2}.
            Note that $q_i$ must be on the mid-trail between $\Dth{\alpha}$ and $\alpha$, i.e., $\Low{q_i} = \alpha$, as otherwise $x$ would not be in the out-panel of $\Window{p_i}{q_i}$.
            We first show that there is no other minimum with value between $f(\alpha)$ and $f(p_i)$ whose $\Dth{\cdot}$ pointer points to a maximum on the banana spanned by $\alpha$ and $\Dth{\alpha}$.
            The in-trail between $\alpha$ and $\Dth{\alpha}$ is empty by assumption.
            Let $s$ be a minimum with $f(\alpha) < f(s) < f(p_i)$ and with $\Dth{s}$ on the mid-trail between $\alpha$ and $\Dth{\alpha}$.
            The banana $(s,\Dth{s})$ has $x$ in its out-panel and $s < \Dth{s}$, and thus must be on the stack $\Ldn$.
            It must be on $\Ldn$ below $(p_i,q_i)$, as otherwise it would have been processed before and $f(\alpha) > f(s)$ at the beginning of the iteration.
            Since $\Ldn$ is sorted by \cref{lem:sorted_stacks} this implies $f(s) > f(p_i)$, which is a contradiction.
            It follows that there is no other minimum that could be interchanged with $\alpha$ in place of $p_i$.
            Setting the value of $\alpha$ to $f(p_i) + \varepsilon$ leads to a violation of \cref{inv:condition-II}, but this is fixed by the interchange of minima.
            Thus, $\UpTree{g_i}$ and $\UpTree{h_i}$ satisfy \cref{inv:condition-I,inv:condition-II,inv:condition-III-max}.

            To prove claim \ref{itm:split-ass-3} we consider the nodes which no longer have an ancestor on any stack or are no longer on any stack themselves.
            The nodes no longer on any stack are $p_i$ and $q_i$, which were taken of $\Ldn$. These satisfy $p_i < x$ and $q_i < x$ by definition of $\Ldn$.
            The remaining nodes to consider are nodes with ancestor $q_i$.
            Descendants $u$ of $q_i$ that were in the banana subtree of $q_i$ in $\UpTree{g_{i-1}}$ satisfy $u < q_i$ by \cref{inv:condition-I} for $\UpTree{g_{i-1}}$.
            The remaining descendants have an ancestor on the mid-trail between $\alpha$ and $q_i$ in the new tree $\UpTree{g_i}$.
            The window spanned by this ancestor and its birth has $x$ in its out-panel and is thus on $\Ldn$.

            The interchange of minima does not modify $\Mid{\alpha}$, but does modify $\In{\alpha}$. More precisely, $\In{\alpha} = q_i$ after the interchange.
            By assumption \ref{itm:split-ass-4} and \cref{lem:sorted_stacks} it follows for all $j \in [i+1,I]$ with $(p_j,q_j) \in \Lup \cup \Rup \cup \Mup$ that $f(q_j) < f(\In{\alpha})$ and $f(q_j) < f(\Mid{\alpha})$.
            By \cref{lem:sorted_stacks} and for all $j \in [i+1,I]$ with $(p_j,q_j) \in \Lup \cup \Rup \cup \Mup$ it holds that $f(q_j) < f(q_i)$, which implies $f(q_j) < f(\In{\alpha})$.
            The value $\varepsilon$ is defined to be smaller than the difference between the values of any two items,
            and thus there is no item with value in $[f(p_i), f(\alpha)]$.
            Because $x$ is in the out-panel of $\Window{p_i},{q_i}$ it holds that $f(p_i) < f(x)$ and thus $f(\alpha) < f(x)$.

            For the sixth claim we first prove that $\Stack_{i+1} \neq \Rdn$.
            If $\Stack_{i+1} \in \Rdn$, then $x < q_{i+1} < p_{i+1}$ and $x$ is in the out-panel of $\Window{p_{i+1}}{q_{i+1}}$ by definition of $\Rdn$.
            This implies that $q_{i+1}$ cannot be on the same trail as $q_i$ and in order for $x$ to be in the out-panel of $\Window{p_{i+1}}{q_{i+1}}$ it must in fact be either in the in-trail between $\alpha$ and $\Dth{\alpha}$,
            or in a right spine of the right subtree.
            The former is impossible, as the in-trail between $\alpha$ and $\Dth{\alpha}$ is empty.
            Now assume that $q_{i+1}$ is on a right trail in the right up-tree.
            Then $p' = \Low{q_{i+1}} < q_{i+1}$ and $\Dth{p'}$ must be on the left spine of the right tree, as otherwise $x$ would not be in the out-panel of $\Window{p_{i+1}}{q_{i+1}}$.
            This implies that $p' < x < \Dth{p'}$, which in turn implies that $(p',\Dth{p'}) \in \Mup$.
            As $\Mup$ and $\Rdn$ are sorted and $f(p') < f(p_{i+1})$ by \cref{inv:condition-II}, this implies that $\Stack_{i+1}$ cannot be $\Rdn$.
            Thus, if $\Stack_{i+1} \in \{ \Ldn, \Rdn \}$, then $\Stack_{i+1} = \Ldn$.
            In this case $p_{i+1} < q_{i+1} < x$, which implies that $q_{i+1}$ must be in the left tree and thus in the same tree as $\alpha$. %
    \end{description}
    This concludes the proof of the inductive step.
\end{proof}

\begin{theorem}[Correctness of \textsc{Split}]
    \label{thm:correctness-of-split}
    Given a list of $m$ items and corresponding map $f$, and an item $\ell$ with $2 \leq \ell \leq m - 1$,
    \textsc{Split} splits the up-tree $\UpTree{f}$ into two up-trees $\UpTree{g}$ and $\UpTree{h}$,
    where $\UpTree{g}$ contains nodes $u$ with $u \leq \ell$ and $\UpTree{h}$ contains the nodes $u \geq \ell$.
\end{theorem}
\begin{proof}
    Recall that we defined $x$ to be the midpoint between $\ell$ and $\ell+1$.
    The proof is by induction over the iterations of the loop in \textsc{Split}, taking \cref{lem:split-iteration-1} as the base case
    and \cref{lem:split-iteration-i} as the induction step.
    This proves the invariants (i) and (ii) introduced in \cref{sec:4.3}, which imply the theorem.
\end{proof}

\subsection{Gluing Two Banana Trees.}
\label{sec:glue-correctness}

In this section we prove that the algorithm \textsc{Glue} presented in \cref{sec:4.3} correctly glues two banana trees $\UpTree{g}$ and $\UpTree{h}$ into $\UpTree{f}$, where $f = g \cdot h$. 
First, we show how to pre-process the functions $g$ and $h$ to obtain the case where $g$ ends in an up-type item on the right, $h$ begins with a down-type item on the left
and this endpoint of $g$ has lower value than the endpoint of $h$.
Write $\ell$ for the rightmost item of $g$ and $\ell'$ for the leftmost item of $h$.
We assume $f(\ell) < f(\ell')$. In the case $f(\ell) > f(\ell')$ we obtain the symmetric case to the one described in section \cref{sec:4.3} and the algorithm is symmetric.
There are four cases depending on whether $\ell$ and $\ell'$ are up-type or down-type items:
(1) $\ell$ and $\ell'$ are up-type; (2) $\ell$ is up-type and $\ell'$ is down-type; (3) $\ell$ is down-type and $\ell'$ is up-type; (4) $\ell$ and $\ell'$ are down-type.
Case (2) is the one described in \cref{sec:4.3} and we show how to reduce the other cases to this case.
\smallskip \begin{description}
    \item[(1) $\rightarrow$ (2):] $\ell'$ becomes non-critical in $g \cdot h$ and we replace it by the dummy leaf $\alpha$.
    \item[(3) $\rightarrow$ (2):] $\ell$ and $\ell'$ become non-critical in $g \cdot h$.
        Since $\ell$ is down-type it is paired with a hook in $\UpTree{g}$ and we simply remove it along with the hook.
        As in the previous case we replace $\ell'$ by the dummy leaf $\alpha$.
    \item[(4) $\rightarrow$ (2):] $\ell$ becomes non-critical in $g \cdot h$.
        As in the previous $\ell$ is paired with a hook in $\UpTree{g}$ and we remove it along with the hook.
\end{description} \smallskip
Whenever $\ell'$ is a down-type item it is paired with a hook and we replace the hook by the dummy leaf $\alpha$.
In all cases we begin with the situation described in \cref{sec:4.3}.

For the remainder of the section we assume that we are in case (2), and we now prove that the algorithm \textsc{Glue} yields the tree $\UpTree{f}$.
The proof will be by induction over the iterations of the loop, with the base case and induction step in \cref{lem:glue-base-case,lem:glue-induction-step}, respectively.
We combine the two and show that the algorithm terminates correctly in \cref{thm:correctness-of-glue}.
Write $\UpTree{g_i}$ and $\UpTree{h_i}$ for the left and right tree after iteration $i$, respectively, with $\UpTree{g_0} = \UpTree{g}$ and $\UpTree{h_0} = \UpTree{h}$.
Define $\alpha$ to be greater than the other nodes of $\UpTree{g_i}$ and less than the other nodes of $\UpTree{h_i}$.
Furthermore, we write $\gluepbth_i$, $\glueplow_i$, $q_i$, $q_i'$ for $\gluepbth$, $\glueplow$, $q$, $q'$ in the $i$-th iteration.

\begin{lemma}[Base Case]
    \label{lem:glue-base-case}
    At the beginning of the first iteration it holds that
    \begin{enumerate}
        \item $\UpTree{g_0}$ and $\UpTree{h_0}$ satisfy \cref{inv:condition-I,inv:condition-II,inv:condition-III-max}, \label{itm:glue-1-clm-inv}
        \item the in-trail between $\alpha$ and $\Dth{\alpha}$ is empty,
            $\alpha \in \{ \Birth{q_1}, \Birth{q_1'} \}$,
            $f(\alpha) > f(\glueplow_1)$,
            $f(\alpha) > f(\gluepbth_1)$, \label{itm:glue-1-clm-alpha}
        \item $\alpha$ is the last node on the right spine of $\UpTree{g_0}$ or on the left spine of $\UpTree{0}$,
            $q_1$ is on the right spine of $\UpTree{g_0}$ and $q_1'$ is on the left spine of $\UpTree{h_0}$ \label{itm:glue-1-clm-spine}
        \item for all nodes $u$ in $\UpTree{g_1}$ and all nodes $v$ in $\UpTree{h_1}$ we have $u < v$, \label{itm:glue-1-clm-separate}
        \item $\textsc{UndoInjury}(\gluepbth_1,q_1)$ being called implies $f(\Mid{\alpha}) > f(\In{q_1})$ and \\
              $\textsc{UndoInjury}(\gluepbth_1,q_1')$ being called implies $f(\Mid{\alpha}) > f(\In{q_1'})$, \label{itm:glue-1-clm-injury}
        \item $\textsc{UndoFatality}(\gluepbth_1,q_1)$ being called implies $f(\Mid{\alpha}) > f(\In{q_1'})$ and \\
              $\textsc{UndoFatality}(\gluepbth_1,q_1')$ being called implies $f(\Mid{\alpha}) > f(\In{q_1})$, \label{itm:glue-1-clm-fatality}
        \item $\textsc{UndoScare}(\glueplow_1,q_1)$ being called implies $\Birth{q_1} = \alpha$, $f(\Mid{\alpha}) > f(\In{q_1'})$ and \\
              $\textsc{UndoScare}(\glueplow_1,q_1')$ being called implies $\Birth{q_1'} = \alpha$, $f(\Mid{\alpha}) > f(\In{q_1})$. \label{itm:glue-1-clm-scare}
    \end{enumerate}
\end{lemma}
\begin{proof}
    The trees $\UpTree{g_0}$ and $\UpTree{h_0}$ are the input to the algorithm, i.e., $\UpTree{g}$ and $\UpTree{h}$, and thus satisfy \cref{inv:condition-I,inv:condition-II,inv:condition-III-max}.
    Since $\alpha$ is a hook, the in-trail between $\alpha$ and $\Dth{\alpha}$ is empty.
    By definition $\alpha = \Birth{b_0'} = \Birth{q_1'}$.
    The inequalities $f(\alpha) > f(\gluepbth_i)$ and $f(\alpha) > f(\glueplow_1)$ also follow from the assumptions on $g$ and $h$.
    The nodes $q_1 = b_0$ and $q_1' = b_0'$ are on the right spine of $\UpTree{g_0}$ and the left spine of $\UpTree{h_0}$, respectively.
    The node $\alpha = \Birth{q_1'}$ is thus also on the spine of $\UpTree{h_0}$ and since it is a leaf it is the last node on the spine.
    Claim \ref{itm:glue-1-clm-separate} holds by the assumption that items in $g$ are all less than items in $h$.

    We now prove claims \ref{itm:glue-1-clm-injury}, \ref{itm:glue-1-clm-fatality} and \ref{itm:glue-1-clm-scare}
    If $f(q_1) < f(q_1)'$, then $\gluepbth_1 = a_0$ and by \cref{inv:condition-II} $f(\gluepbth_1) > f(\glueplow_1)$.
    Thus $\textsc{UndoInjury}(\gluepbth_1,q_1)$ is executed.
    Since $\Birth{q_1} = a_0$ is an up-type item it also holds that $\In{q_1} = a_0$. As by definition $f(\alpha) > f(a_0)$ the inequality $f(\Mid{\alpha}) > f(\In{q_1})$ holds.
    This implies claim \ref{itm:glue-1-clm-injury}.
    If $f(q_1) > f(q_1')$, then $\gluepbth_1 = a_0$ and $\gluepbth_1 \neq \Birth{q_1'}$.
    Either $\textsc{UndoFatality}(\gluepbth_1,q_1')$ or $\textsc{UndoScare}(\glueplow_1, q_1)$ is executed.
    The inequality $f(\Mid{\alpha}) > f(\In{q_1})$ holds as we have just shown and this proves claims \ref{itm:glue-1-clm-fatality} and \ref{itm:glue-1-clm-scare}.

    We have shown that all claims hold at the beginning of the first iteration and this proves the lemma.
\end{proof}

\begin{lemma}[Inductive Step]
    \label{lem:glue-induction-step}
    If at the beginning of the $i$-th iteration the following conditions hold
    \smallskip \begin{enumerate}
        \item $\UpTree{g_{i-1}}$ and $\UpTree{h_{i-1}}$ satisfy \cref{inv:condition-I,inv:condition-II,inv:condition-III-max}, \label{itm:glue-ass-inv}
        \item the in-trail between $\alpha$ and $\Dth{\alpha}$ is empty,
            $\alpha \in \{ \Birth{q_i}, \Birth{q_i'} \}$,
            $f(\alpha) > f(\glueplow_i)$,
            $f(\alpha) > f(\gluepbth_i)$, \label{itm:glue-ass-alpha}
        \item $\alpha$ is the last node on the right spine of $\UpTree{g_{i-1}}$ or on the left spine of $\UpTree{h_{i-1}}$,
            $q_i$ is on the right spine of $\UpTree{g_{i-1}}$ and $q_i'$ is on the left spine of $\UpTree{h_{i-1}}$ \label{itm:glue-ass-spine}
        \item for all nodes $u$ in $\UpTree{g_i}$ and all nodes $v$ in $\UpTree{h_i}$ we have $u < v$, \label{itm:glue-ass-separate}
        \item $\textsc{UndoInjury}(\gluepbth_i,q_i)$ being called implies $f(\Mid{\alpha}) > f(\In{q_i})$ and \\
              $\textsc{UndoInjury}(\gluepbth_i,q_i')$ being called implies $f(\Mid{\alpha}) > f(\In{q_i'})$, \label{itm:glue-ass-injury}
        \item $\textsc{UndoFatality}(\gluepbth_i,q_i)$ being called implies $f(\Mid{\alpha}) > f(\In{q_i'})$ and \\
              $\textsc{UndoFatality}(\gluepbth_i,q_i')$ being called implies $f(\Mid{\alpha}) > f(\In{q_i})$, \label{itm:glue-ass-fatality}
        \item $\textsc{UndoScare}(\glueplow_i,q_i)$ being called implies $\Birth{q_i} = \alpha$, $f(\Mid{\alpha}) > f(\In{q_i'})$ and \\
              $\textsc{UndoScare}(\glueplow_i,q_i')$ being called implies $\Birth{q_i'} = \alpha$, $f(\Mid{\alpha}) > f(\In{q_i})$, \label{itm:glue-ass-scare}
    \end{enumerate} \smallskip
    then at the beginning of the next iteration, if it exists, it holds that
    \smallskip \begin{enumerate}
        \item $\UpTree{g_i}$ and $\UpTree{h_i}$ satisfy \cref{inv:condition-I,inv:condition-II,inv:condition-III-max}, \label{itm:glue-clm-inv}
        \item the in-trail between $\alpha$ and $\Dth{\alpha}$ is empty,
            $\alpha \in \{ \Birth{q_{i+1}}, \Birth{q_{i+1}'} \}$,
            $f(\alpha) > f(\glueplow_{i+1})$,
            $f(\alpha) > f(\gluepbth_{i+1})$, \label{itm:glue-clm-alpha}
        \item $\alpha$ is the last node on the right spine of $\UpTree{g_i}$ or on the left spine of $\UpTree{h_i}$, 
            $q_{i+1}$ is on the right spine of $\UpTree{g_i}$ and $q_{i+1}'$ is on the left spine of $\UpTree{h_i}$\label{itm:glue-clm-spine}
        \item for all nodes $u$ in $\UpTree{g_{i+1}}$ and all nodes $v$ in $\UpTree{h_{i+1}}$ we have $u < v$, \label{itm:glue-clm-separate}
        \item $\textsc{UndoInjury}(\gluepbth_{i+1},q_{i+1})$ being called implies $f(\Mid{\alpha}) > f(\In{q_{i+1}})$ and \\
              $\textsc{UndoInjury}(\gluepbth_{i+1},q_{i+1}')$ being called implies $f(\Mid{\alpha}) > f(\In{q_{i+1}'})$, \label{itm:glue-clm-injury}
        \item $\textsc{UndoFatality}(\gluepbth_{i+1},q_{i+1})$ being called implies $f(\Mid{\alpha}) > f(\In{q_{i+1}'})$ and \\
              $\textsc{UndoFatality}(\gluepbth_{i+1},q_{i+1}')$ being called implies $f(\Mid{\alpha}) > f(\In{q_{i+1}})$, \label{itm:glue-clm-fatality}
        \item $\textsc{UndoScare}(\glueplow_{i+1},q_{i+1})$ being called implies $\Birth{q_{i+1}} = \alpha$, $f(\Mid{\alpha}) > f(\In{q_{i+1}'})$ and \\
              $\textsc{UndoScare}(\glueplow_{i+1},q_{i+1}')$ being called implies $\Birth{q_{i+1}'} = \alpha$, $f(\Mid{\alpha}) > f(\In{q_{i+1}})$. \label{itm:glue-clm-scare}
    \end{enumerate}
\end{lemma}
\begin{proof}
    Assume that the conditions hold.
    We prove the claims for $f(q_i) < f(q_i')$ by considering separately the three cases that the algorithm distinguishes.
    The case $f(q_i) > f(q_i')$ is symmetric.

    \smallskip\noindent\textbf{Case 1 (\boldmath$f(\gluepbth_i) < f(\glueplow_i)$ and \boldmath$\gluepbth_i = \Birth{q_i}$):}
        $\textsc{UndoInjury}(\gluepbth_i, q_i)$ is called, which removes nodes $j$ with $f(j) < f(q_i)$ from the mid-trail between $\alpha$ and $\Dth{\alpha}$ and
        inserts them into the in-trail between $\gluepbth_i$ and $q_i$, specifically between the nodes $\In{q_i}$ and $q_i$.

        We first prove that $\UpTree{g_i}$ and $\UpTree{h_i}$ satisfy \cref{inv:condition-I,inv:condition-II,inv:condition-III-max}.
        By assumption \ref{itm:glue-ass-inv} these invariants holds for $\UpTree{g_{i-1}}$ and $\UpTree{h_{i-1}}$.
        To see that they also hold for $\UpTree{h_i}$, observe that removing maxima along with its banana subtree does not affect the invariants.
        For $\UpTree{g_i}$ we first show \cref{inv:condition-III-max}.
        Write $j^-$ and $j^+$ for the lowest and highest of the nodes are moved, respectively.
        That is, in $\UpTree{g_{i-1}}$ $j^- = \Mid{\alpha}$ and $j^+$ is the highest node on the mid-trail between $\alpha$ and $\Dth{\alpha}$ such that $f(j^+) < f(q_i)$.
        Write $j_q$ for $\In{q_i}$ in $\UpTree{g_{i-1}}$.
        In $\UpTree{g_i}$ these nodes are inserted such that $j^- = \Up{j_q}$ and $j^+ = \In{q_i}$.
        By the assumption that $f(\Mid{\alpha}) > f(\In{q_i})$ (assumption \ref{itm:glue-ass-injury}) we thus have $f(\Dn{j^-}) < f(j^-)$ in $\UpTree{g_i}$.
        Furthermore, $f(j^+) < f(\Up{j^+}) = f(q_i)$.
        By assumption \ref{itm:glue-ass-separate} and \cref{inv:condition-III-max} for $\UpTree{h_{i-1}}$
        it holds in $\UpTree{g_i}$ that $\Dn{j^-} < j^- < \Up{j^-} < \dots < \Dn{j^+} < j^+$ and 
        $\Up{j^+} = q_i < j^+$.
        Thus \cref{inv:condition-III-max} holds in $\UpTree{g_i}$.
        To see that \cref{inv:condition-I} is satisfied note that the $j$ that are inserted into $\UpTree{g_i}$ satisfy $q_i < j$ by assumption \ref{itm:glue-ass-separate}.
        Since $q_i$ is on the right spine $q_i < \gluepbth_i = \Birth{q_i}$ and thus the condition of \cref{inv:condition-I} holds for $q_i$ and its ancestors.
        The nodes $j$ gain descendants $u$ in the same subtree $\Dn{j}$, and these nodes $u$ satisfy $u < j$, again by assumption \ref{itm:glue-ass-separate}.
        Thus, these nodes $j$ also satisfy the condition for \cref{inv:condition-I}.
        For no other node do the subtrees change and thus \cref{inv:condition-I} holds in $\UpTree{g_i}$.
        Finally, we prove \cref{inv:condition-II} by analyzing the nodes for which the nodes at $\Dth{\Low{\cdot}}$ changes.
        These are the nodes $v$ with $\Dth{v} = j$ for some moved node $j$. They originally have $\Dth{\Low{v}} = \alpha$ and this changes to $\Dth{\Low{v}} = \gluepbth_i$.
        Assumption \ref{itm:glue-ass-alpha} that $f(\alpha) > f(\gluepbth_i)$ thus implies $f(\Dth{\Low{v}}) > f(\alpha) > f(\gluepbth_i)$.
        This implies \cref{inv:condition-II}.
        In the remainder of the proof for this case we assume \cref{inv:condition-I,inv:condition-II,inv:condition-III-max} for $\UpTree{g_i}$ and $\UpTree{h_i}$.

        The in-trail between $\alpha$ and $\Dth{\alpha}$ is empty at the beginning of the $i$-th iteration,
        is not modified by \textsc{UndoInjury} and is thus also empty at the beginning of the next iteration.
        As $q_{i+1}' = q_i'$ and $\Dth{\alpha} = q_i'$ is unchanged it holds that $\alpha = \Birth{q_{i+1}'}$ at the beginning of the next iteration.
        If in the next iteration $f(q_{i+1}) < f(q_{i+1}')$, then $\gluepbth_{i+1} = \Birth{q_{i+1}} = \Low{q_i} = \glueplow_i$ and $\glueplow_{i+1} = \Low{q_{i+1}}$.
        By \cref{inv:condition-II} this implies $f(\glueplow_{i+1}) < f(\gluepbth_{i+1}) < f(\gluepbth_i) < \alpha$, where the last inequality follows from assumption \ref{itm:glue-ass-alpha}.
        Otherwise, if in the next iteration $f(q_{i+1}) > f(q_{i+1}')$, then $\gluepbth_{i+1} = \Birth{q_{i+1}} = \Low{q_i}$ (since $\Birth{q_{i+1}'} = \alpha$) and $\glueplow_{i+1} = \Low{q_{i+1}'} = \Low{q_i'}$.
        By \cref{inv:condition-II} and assumption \ref{itm:glue-ass-alpha} this implies $f(\glueplow_{i+1}) < f(\alpha)$ and $f(\gluepbth_{i+1}) < f(\gluepbth_i) < f(\alpha)$.
        This proves claim \ref{itm:glue-clm-alpha}.
        Claim \ref{itm:glue-clm-spine} follows immediately from the assumption \ref{itm:glue-ass-spine}:
        $\alpha$ was the last node on the left spine of $\UpTree{h_{i-1}}$ at the beginning of the iteration,
        and the removal of nodes from its mid-trail does not change that;
        $q_{i+1} = \Dth{\Low{q_i}}$, which is on the right spine of $\UpTree{g_i}$ since $q_i$ is on the right spine of $\UpTree{g_{i-1}}$
        and $q_{i+1}' = q_i'$ is on the left spine of $\UpTree{h_i}$ because $q_i'$ is on the left spine of $\UpTree{h_{i-1}}$.

        Since $\alpha$ is the last node on the left spine of $\UpTree{h_{i-1}}$ and its in-trail is empty,
        the nodes that are moved to $\UpTree{g_i}$ are the leftmost nodes of $\UpTree{h_{i-1}}$ other than $\alpha$.
        Assumption \ref{itm:glue-ass-separate} and the definition of $\alpha$ thus imply the claim \ref{itm:glue-clm-separate}.

        We now prove claims \ref{itm:glue-clm-injury}, \ref{itm:glue-clm-fatality} and \ref{itm:glue-clm-scare}.
        Recall that $q_{i+1}' = q_{i+1}$ and $q_{i+1}' = \Dth{\Low{q_i}}$.
        In total there are six cases, but only three can occur:
        \smallskip \begin{enumerate}[(i)]
            \item $f(q_{i+1}) < f(q_{i+1}')$, $f(\gluepbth_{i+1}) > f(\glueplow_{i+1})$ and $\gluepbth_{i+1} = \Birth{q_{i+1}}$ \\
                $\implies$ call to $\textsc{UndoInjury}(\gluepbth_{i+1}, q_{i+1})$
            \item $f(q_{i+1}) > f(q_{i+1}')$, $f(\gluepbth_{i+1}) > f(\glueplow_{i+1})$ and $\gluepbth_{i+1} = \Birth{q_{i+1}}$ \\
                $\implies$ call to $\textsc{UndoFatality}(\gluepbth_{i+1}, q_{i+1}')$
            \item $f(q_{i+1}) > f(q_{i+1}')$, $f(\glueplow_{i+1}) > f(\gluepbth_{i+1})$ \\
                $\implies$ call to $\textsc{UndoScare}(\glueplow_{i+1}, q_{i+1}')$
        \end{enumerate} \smallskip
        Since $\gluepbth_{i+1} = \Birth{q_{i+1}} \neq \alpha$, if $f(q_{i+1}) < f(q_{i+1})'$, then \textsc{UndoFatality} will not be called.
        By \cref{inv:condition-II} it also holds that $f(\glueplow_{i+1}) < f(\gluepbth_{i+1})$, which implies that \textsc{UndoScare} will not be called.
        Finally, $\Birth{q_{i+1}'} = \alpha$ implies that if $f(q_{i+1}) > f(q_{i+1}')$, then \textsc{UndoInjury} will not be called.
        We analyze the three possible cases separately:
        \smallskip \begin{enumerate}[(i)]
            \item $f(q_{i+1}) < f(q_{i+1}')$, $f(\gluepbth_{i+1}) > f(\glueplow_{i+1})$ and $\gluepbth_{i+1} = \Birth{q_{i+1}}$ and $\textsc{UndoInjury}(\gluepbth_{i+1}, q_{i+1})$ is called.
                We have shown above that $q_{i+1}$ is on the spine, as is $q_i$, and thus $q_{i+1} = \Dth{\Low{q_i}}$ implies $\In{q_{i+1}} = q_i$.
                Since $f(\Mid{\alpha})$ was not moved to $\UpTree{g_i}$ and $\Dth{\alpha} = f(q_i') > f(q_i)$ it follows that $f(\Mid{\alpha}) > f(q_i) = f(\In{q_{i+1}})$.
            \item $f(q_{i+1}) > f(q_{i+1}')$, $f(\gluepbth_{i+1}) > f(\glueplow_{i+1})$ and $\gluepbth_{i+1} = \Birth{q_{i+1}}$ and $\textsc{UndoFatality}(\gluepbth_{i+1}, q_{i+1}')$ is called.
                The inequality $f(\Mid{\alpha}) > f(\In{q_{i+1}})$ holds by the same argument as in the previous case.
            \item $f(q_{i+1}) > f(q_{i+1}')$, $f(\glueplow_{i+1}) > f(\gluepbth_{i+1})$ and $\textsc{UndoScare}(\glueplow_{i+1}, q_{i+1}')$ is called.
                Since $q_{i+1}' = q_i'$ and $\alpha = \Birth{q_i'}$ it holds that $\alpha = \Birth{q_{i+1}'}$.
                The inequality $f(\Mid{\alpha}) > f(\In{q_{i+1}})$ holds by the same argument as in the other cases.
        \end{enumerate} \smallskip
        In all three cases the claims hold.

    \smallskip\noindent\textbf{Case 2 (\boldmath$f(\gluepbth_i) < f(\glueplow_i)$ and \boldmath$\gluepbth_i = \Birth{q_i'}$):}
        $\textsc{UndoFatality}(\gluepbth_i,q_i)$ is called, which takes on the in-trail between $\gluepbth_i$ and $q_i'$
        and the nodes $j$ with $f(j) < f(q_i)$ on the mid-trail between $\gluepbth_i$ and $q_i'$
        and swaps them with $\alpha = \Birth{q_i}$,
        such that the nodes formerly on the in-trail extend the mid-trail of $q_i$ and the nodes $j$ extend the in-trail of $q_i$.

        We first prove that $\UpTree{g_i}$ and $\UpTree{h_i}$ satisfy \cref{inv:condition-I,inv:condition-II,inv:condition-III-max}.
        By assumption \ref{itm:glue-ass-inv} these invariants hold for $\UpTree{g_{i-1}}$ and $\UpTree{h_{i-1}}$.
        To see that they also hold for $\UpTree{h_i}$, observe that replacing a part of the banana $(\gluepbth_i,q_i')$ by $\alpha$ does not affect \cref{inv:condition-I,inv:condition-III-max}.
        By assumption \ref{itm:glue-ass-alpha} $f(\alpha) > f(\gluepbth_i)$ and thus $f(\alpha) > f(\gluepbth_i) > f(\Low{\Dth{\alpha}}$, where the second inequality follows from \cref{inv:condition-II} for $\UpTree{h_{i-1}}$.
        As $\Low{\Dth{\cdot}}$ changes for no other node \cref{inv:condition-II} also holds for $\UpTree{h_i}$.
        Write $j^+$ for the topmost node on the mid-trail between $\gluepbth_i$ and $q_i'$ in $\UpTree{h_{i-1}}$ that is moved to $\UpTree{g_i}$,
        and $k^+$ for the topmost node on the in-trail between $\gluepbth_i$ and $q_i'$ in $\UpTree{h_{i-1}}$.
        As the entire in-trail is moved to the left tree $k^+ = \In{q_i'}$ in $\UpTree{h_{i-1}}$.
        Write $m_\alpha$ for the node $\Mid{\alpha}$ in $\UpTree{g_{i-1}}$.
        In $\UpTree{g_i}$ $\Dn{m_\alpha} = k^+$.
        By assumption \ref{itm:glue-ass-fatality} we have $f(m_\alpha) > f(k^+)$,
        and by definition of $j^+$ we have $f(j^+) < f(q_i)$.
        Furthermore, by assumption \ref{itm:glue-ass-separate} it holds that $m_\alpha < k^+$ and $q_i < j^+$.
        Since these are the only node where $\Up{\cdot}$ and $\Dn{\cdot}$ pointers change it follows that \cref{inv:condition-III-max} holds in $\UpTree{g_i}$.
        To see that \cref{inv:condition-I} is satisfied note that the nodes $u$ that are inserted into $\UpTree{g_i}$ satisfy $u > v$ for any node $v$ that is already in this tree from the previous iteration.
        With $q_i < \Birth{q_i}$ as $q_i$ is on the right spine and $s < \Dn{s}$ for any node on the mid-trail beginning at $q_i$ \cref{inv:condition-I} follows.
        The nodes for which $\Low{\Dth{\cdot}}$ changes are those nodes $t$ with $\Dth{t}$ on the mid-trail beginning at $q_i$ and the node $\gluepbth_i$.
        It changes from $\alpha$ to $\gluepbth_i$. By assumption \ref{itm:glue-ass-alpha} $f(\gluepbth_i) < f(\alpha)$, so $f(\Low{\Dth{t}}) > f(\alpha) > f(\gluepbth_i)$.
        Furthermore, for $\gluepbth_i$ we have $\Low{\Dth{\gluepbth_i}} = \glueplow_i$ and they satisfy $f(\gluepbth_i) > f(\glueplow)$.
        This implies \cref{inv:condition-II} for $\UpTree{g_i}$ and concludes the proof of claim \ref{itm:glue-ass-inv}.
        In the remainder of the proof for this case we assume \cref{inv:condition-I,inv:condition-II,inv:condition-III-max} for $\UpTree{g_i}$ and $\UpTree{h_i}$.

        After the iteration $\Dth{\alpha} = q_i'$. As the nodes on the in-trail beginning at $q_i'$ are removed, this trail is empty.
        Since $q_{i+1}' = q_i'$ it holds that $\alpha = \Birth{q_{i+1}'}$.
        For the inequalities $f(\alpha) > f(\gluepbth_{i+1})$ and $f(\alpha) > f(\glueplow_{i+1})$ we refer to the proof for the previous case.
        It follows that claim \ref{itm:glue-clm-alpha} holds.
        The node $q_i'$ is on the left spine of $\UpTree{h_{i-1}}$ and thus $q_{i+1}' = q_i'$ is on the left spine of $\UpTree{h_i}$.
        Since $\alpha = \Birth{q_{i+1}'}$ in $\UpTree{h_i}$ it also follows that $\alpha$ is on the left spine of $\UpTree{h_i}$.
        Similarly, in $\UpTree{g_{i-1}}$ and $\UpTree{g_i}$ the node $q_i$ is on the right spine and thus $q_{i+1} = \Dth{\Low{q_i}}$ is also on the right spine.
        This proves claim \ref{itm:glue-clm-spine}.
        For the proof of claim \ref{itm:glue-clm-separate} we again refer to the previous case, since the proof is similar.

        We now prove claims \ref{itm:glue-clm-injury}, \ref{itm:glue-clm-fatality} and \ref{itm:glue-clm-scare}.
        Again, in the next iteration only three out of six cases can occur:
        \smallskip \begin{enumerate}[(i)]
            \item $f(q_{i+1}) < f(q_{i+1}')$, $f(\gluepbth_{i+1}) > f(\glueplow_{i+1})$ and $\gluepbth = \Birth{q_{i+1}}$ \\
                $\implies$ call to $\textsc{UndoInjury}(\gluepbth_{i+1}, q_{i+1})$
            \item $f(q_{i+1}) > f(q_{i+1}')$, $f(\gluepbth_{i+1}) > f(\glueplow_{i+1})$ and $\gluepbth = \Birth{q_{i+1}}$ \\
                $\implies$ call to $\textsc{UndoFatality}(\gluepbth_{i+1}, q_{i+1}')$
            \item $f(q_{i+1}) > f(q_{i+1}')$, $f(\glueplow_{i+1}) > f(\gluepbth_{i+1})$ \\
                $\implies$ call to $\textsc{UndoScare}(\glueplow_{i+1}, q_{i+1}')$
        \end{enumerate} \smallskip
        Note that these are the same cases as in Case 1. The proof that the remaining three cases cannot occur is identical.
        The remainder of the proof is also similar to the proof in case 1:
        \smallskip \begin{enumerate}[(i)]
            \item $f(q_{i+1}) < f(q_{i+1}')$, $f(\gluepbth_{i+1}) > f(\glueplow_{i+1})$ and $\gluepbth = \Birth{q_{i+1}}$ and $\textsc{UndoInjury}(\gluepbth_{i+1}, q_{i+1})$ is called.
                Again, $q_{i+1}$ and $q_i$ are on the spine and thus $\In{q_{i+1}} = q_i$.
                $f(\Mid{\alpha})$ is one of the nodes that have not been moved from the mid-trail, and thus $f(\Mid{\alpha}) > f(q_i) = f(\In{q_{i+1}})$.
            \item $f(q_{i+1}) > f(q_{i+1}')$, $f(\gluepbth_{i+1}) > f(\glueplow_{i+1})$ and $\gluepbth = \Birth{q_{i+1}}$ and $\textsc{UndoFatality}(\gluepbth_{i+1}, q_{i+1}')$ is called.
                The inequality $f(\Mid{\alpha}) > f(\In{q_{i+1}})$ holds by the same argument as in the previous case.
            \item $f(q_{i+1}) > f(q_{i+1}')$, $f(\glueplow_{i+1}) > f(\gluepbth_{i+1})$ and $\textsc{UndoScare}(\glueplow_{i+1}, q_{i+1}')$ is called.
                As before $\alpha = \Birth{q_i'} = \Birth{q_{i+1}'}$.
                The inequality $f(\Mid{\alpha}) > f(\In{q_{i+1}}$ holds as shown in the other cases.
        \end{enumerate} \smallskip
        This proves claims \ref{itm:glue-ass-injury}, \ref{itm:glue-ass-fatality} and \ref{itm:glue-ass-scare}.

    \smallskip\noindent\textbf{Case 3 (\boldmath$f(\glueplow_i) < f(\gluepbth_i)$):}
        $\textsc{UndoScare}(\glueplow_i,q_i)$ is called, which sets $f(\alpha) = f(\glueplow_i) - \varepsilon$ and then performs an interchange of minima.
        Note that $\alpha = \Birth{q_i}$, as $f(\gluepbth_i) < f(\glueplow_i)$, which by \cref{inv:condition-II} for $\UpTree{g_i}$ cannot be if $\gluepbth_i = \Birth{q_i}$.
        Thus, $\alpha \neq \Birth{q_i'}$ and by assumption \ref{itm:glue-ass-alpha} this implies $\alpha = \Birth{q_i}$.
        Claim \ref{itm:glue-clm-separate} follows directly from assumption \ref{itm:glue-ass-separate}, as the contents of the trees do not change.
        \Cref{inv:condition-I,inv:condition-II,inv:condition-III-max} immediately follow for $\UpTree{g_i}$ from the correctness of the interchange of minima.
        As $\UpTree{h_i} = \UpTree{h_{i-1}}$ these invariants hold $\UpTree{h_i}$ by assumption \ref{itm:glue-ass-inv}.
        In the remainder of the proof for this case we assume \cref{inv:condition-I,inv:condition-II,inv:condition-III-max} for $\UpTree{g_i}$ and $\UpTree{h_i}$.

        The in-trail between $\alpha$ and $\Dth{\alpha}$ is empty in $\UpTree{g_{i-1}}$ by assumption \ref{itm:glue-ass-alpha}.
        As $\Dth{\alpha} = q_i$ is on a spine the in-trail remains empty after the interchange of minima.
        The interchange also results in $\Dth{\alpha} = q_{i+1}$ in $\UpTree{g_i}$.
        If $f(q_{i+1}) < f(q_{i+1})'$, then $\gluepbth_{i+1} = \Birth{q_{i+1}'} = \Birth{q_i'} = \gluepbth_i$.
        As $f(\gluepbth_i) < f(\glueplow_i)$ and $\varepsilon$ is defined to be sufficiently small it follows that $f(\gluepbth_{i+1}) < f(\alpha)$.
        Furthermore, $f(\glueplow_{i+1}) < f(\alpha)$ by \cref{inv:condition-II}.
        Otherwise, if $f(q_{i+1}) > f(q_{i+1})'$, then $\gluepbth_{i+1} = \Birth{q_{i+1}'} = \Birth{q_i'} = \gluepbth_i$,
        which we have just shown to satisfy $f(\gluepbth_{i+1}) < f(\alpha)$.
        By \cref{inv:condition-II} $f(\glueplow_{i+1}) < f(\gluepbth_{i+1})$, which implies $f(\glueplow_{i+1}) < f(\alpha)$.
        This proves claim \ref{itm:glue-clm-alpha}.

        The node $q_{i+1}$ is on the spine since it is already on the spine in $\UpTree{g_{i-1}}$ and the interchange of minima does not change its ancestors.
        The node $\alpha$ remains on the spine since its in-trail remains empty and $q_{i+1} = \Dth{\alpha}$ is on the left spine in $\UpTree{g_{i+1}}$.
        Since it is also a leaf it is the last node on the spine.
        Finally, $q_{i+1}'$ is on the spine since $q_{i+1}' = q_i'$.
        This proves \ref{itm:glue-clm-spine}.

        It remains to prove claims \ref{itm:glue-clm-injury}, \ref{itm:glue-ass-fatality} and \ref{itm:glue-clm-scare}.
        Again, in the next iteration only three out of six cases can occur:
        \smallskip \begin{enumerate}[(i)]
            \item $f(q_{i+1}) < f(q_{i+1}')$, $f(\gluepbth_{i+1}) > f(\glueplow_{i+1})$ and $\gluepbth_{i+1} = \Birth{q_{i+1}'}$ \\
                $\implies$ call to $\textsc{UndoFatality}(\gluepbth_{i+1}, q_{i+1})$
            \item $f(q_{i+1}) < f(q_{i+1}')$, $f(\glueplow_{i+1}) > f(\gluepbth_{i+1})$ \\
                $\implies$ call to $\textsc{UndoScare}(\glueplow_{i+1}, q_{i+1})$
            \item $f(q_{i+1}) > f(q_{i+1}')$, $f(\gluepbth_{i+1}) > f(\glueplow_{i_1})$ and $\gluepbth_{i+1} = \Birth{q_{i+1}'}$ \\
                $\implies$ call to $\textsc{UndoInjury}(\gluepbth_{i+1}, q_{i+1}')$
        \end{enumerate} \smallskip
        Since $\Birth{q_{i+1}} = \alpha$ we cannot have $\gluepbth_{i+1} = \Birth{q_{i+1}}$ and thus if $f(q_{i+1}) < f(q_{i+1}')$ then \textsc{UndoInjury} is not called.
        If $f(q_{i+1}) > f(q_{i+1}')$ then $f(\gluepbth_{i+1}) > f(\glueplow_{i+1})$ by \cref{inv:condition-II}, and so \textsc{UndoScare} is not called.
        \textsc{UndoFatality} is not called in this case as $\gluepbth_{i+1} = \Birth{q_{i+1}'}$.
        We now analyze the three possible cases:
        \smallskip \begin{enumerate}[(i)]
            \item $f(q_{i+1}) < f(q_{i+1}')$, $f(\gluepbth_{i+1}) > f(\glueplow_{i+1})$ and $\gluepbth_{i+1} = \Birth{q_{i+1}'}$ and $\textsc{UndoFatality}(\gluepbth_{i+1}, q_{i+1})$ is called.
                The node $\Mid{\alpha}$ is the same in $\UpTree{g_{i-1}}$ and $\UpTree{g_i}$.
                Furthermore $\In{q_{i+1}}' = \In{q_i'}$, as $\UpTree{h_i} = \UpTree{h_{i-1}}$.
                The assumption that $f(\Mid{\alpha}) > f(\In{q_{i}'}$ thus implies $f(\Mid{\alpha}) > f(\In{q_{i+1}'})$.
            \item $f(q_{i+1}) < f(q_{i+1}')$, $f(\glueplow_{i+1}) > f(\gluepbth_{i+1})$ and $\textsc{UndoScare}(\glueplow_{i+1}, q_{i+1})$ is called.
                We have shown above that $\Dth{\alpha} = q_{i+1}$ which implies $\alpha = \Birth{q_{i+1}}$.
            \item $f(q_{i+1}) > f(q_{i+1}')$, $f(\gluepbth_{i+1}) > f(\glueplow_{i_1})$ and $\gluepbth_{i+1} = \Birth{q_{i+1}'}$ and $\textsc{UndoInjury}(\gluepbth_{i+1}, q_{i+1}')$ is called.
                The inequality $f(\Mid{\alpha}) > f(\In{q_{i+1}'})$ holds by the same argument as in the first case.
        \end{enumerate}
        This concludes the proof of the inductive step.
\end{proof}
\begin{theorem}[Correctness of \textsc{Glue}]
    \label{thm:correctness-of-glue}
    Given two lists of items and corresponding maps $g$ and $h$
    \textsc{Glue} applied to $\UpTree{g}$ and $\UpTree{h}$ yields the up-tree $\UpTree{f} = \UpTree{g \cdot h}$.
\end{theorem}
\begin{proof}
    We show by induction over the iterations of \textsc{Glue} that at the beginning of the $i$-th iteration it holds that
    \smallskip \begin{enumerate}
        \item $\UpTree{g_{i-1}}$ and $\UpTree{h_{i-1}}$ satisfy \cref{inv:condition-I,inv:condition-II,inv:condition-III-max},
        \item the in-trail between $\alpha$ and $\Dth{\alpha}$ is empty,
            $\alpha \in \{ \Birth{q_i}, \Birth{q_i'} \}$,
            $f(\alpha) > f(\glueplow_i)$,
            $f(\alpha) > f(\gluepbth_i)$,
        \item $\alpha$ is the last node on the right spine of $\UpTree{g_{i-1}}$ or on the left spine of $\UpTree{h_{i-1}}$,
            $q_i$ is on the right spine of $\UpTree{g_{i-1}}$ and $q_i'$ is on the left spine of $\UpTree{h_{i-1}}$
        \item for all nodes $u$ in $\UpTree{g_i}$ and all nodes $v$ in $\UpTree{h_i}$ we have $u < v$,
        \item $\textsc{UndoInjury}(\gluepbth_i,q_i)$ being called implies $f(\Mid{\alpha}) > f(\In{q_i})$ and \\
              $\textsc{UndoInjury}(\gluepbth_i,q_i')$ being called implies $f(\Mid{\alpha}) > f(\In{q_i'})$,
        \item $\textsc{UndoFatality}(\gluepbth_i,q_i)$ being called implies $f(\Mid{\alpha}) > f(\In{q_i'})$ and \\
              $\textsc{UndoFatality}(\gluepbth_i,q_i')$ being called implies $f(\Mid{\alpha}) > f(\In{q_i})$,
        \item $\textsc{UndoScare}(\glueplow_i,q_i)$ being called implies $\Birth{q_i} = \alpha$, $f(\Mid{\alpha}) > f(\In{q_i'})$ and \\
              $\textsc{UndoScare}(\glueplow_i,q_i')$ being called implies $\Birth{q_i'} = \alpha$, $f(\Mid{\alpha}) > f(\In{q_i})$.
    \end{enumerate} \smallskip
    \Cref{lem:glue-base-case} proves the base case for the first iteration $i=1$, and the induction step is in \cref{lem:glue-induction-step}.
    We assume these claims for all iterations.

    We now show that there exists an iteration after which one tree is empty.
    Assume that the global maximum of $h$ has greater value than the global maximum of $g$ and both trees have nodes other than the special root and $\alpha$.
    The case where the global maximum of $g$ has greater value than that of $h$ is symmetric.
    Consider any iteration $I$ in which $f(q_I') > f(u)$ for any $u \neq \beta_g$ in $\UpTree{g_{I-1}}$ and $q_I = \beta_g$.
    Since in some iteration $q_I' = \beta_h$ this iteration exsists.
    In this iteration $f(q_I') < f(q_I)$, either since $q_I$ is a special root and $q_I'$ is not, or since $f(\beta_h) < f(\beta_g)$ by definition.
    We distinguish two cases: $q_I' \neq \beta_h$ and $q_I = \beta_h$.
    \smallskip \begin{description}
        \item[$q_I' \neq \beta_h$:] There are three cases to consider:
            (1) $\Birth{q_I'} \neq \alpha$ and $\gluepbth_I = \Birth{q_I'}$;
            (2) $\Birth{q_I'} = \alpha$ and $f(\gluepbth_I) > f(\glueplow_I)$;
            (3) $\Birth{q_I'} = \alpha$ and $f(\glueplow_I) > f(\gluepbth_I)$.
            In case (1) $f(\gluepbth_I) > f(\glueplow_I)$ by \cref{inv:condition-II} and thus $\textsc{UndoInjury}(\gluepbth_I,q_I')$ is called.
            In case (2) $\textsc{UndoFatality}(\gluepbth_I,q_I')$ is called.
            In case (3) $\textsc{UndoScare}(\glueplow_I,q_I')$ is called.
            If case (1) occurs, then $\alpha = \Birth{q_I} = \Birth{\beta_g}$,
            the in-trail between $\alpha$ and $\beta_g$ is empty,
            and all nodes from the mid-trail between $\alpha$ and $\beta_g$ are moved to $\UpTree{h_{I}}$.
            This leaves $\In{\beta_g} = \Mid{\beta_g} = \alpha$ and thus the loop terminates.
            If case (2) occurs, then $\gluepbth_I = \Birth{q_I} = \Birth{\beta_g}$ and all nodes on trails between $\gluepbth_I$ and $\beta_g$ except $\beta_g$ are moved to $\UpTree{h_I}$,
            and $\alpha$ is moved to $\UpTree{g_I}$ such that $\In{\beta_g} = \Mid{\beta_g} = \alpha$.
            Thus, the loop terminates.
            In case (3) $f(\alpha)$ is set to $f(\glueplow_I) - \varepsilon$ and an interchange of minima between $\alpha$ and $\glueplow_I$ is performed.
            Thus, both trees are not empty.
            However, $q_{I+1}' = \Up{q_I}$ is the next node upwards on the spine towards $\beta_h$ and the next iteration is another iteration where $f(q_I') > f(u)$ for all $u \neq \beta_g \in \UpTree{g_i}$.
        \item[$q_I' = \beta_h$:]
            This case is eventually reached, since the $q_i'$ move upwards on the spine to $\beta_h$.
            We have defined $f(\beta_h) < f(\beta_g)$.
            There are again three cases to consider:
            (1) $\Birth{q_I'} \neq \alpha$ and $\gluepbth_I = \Birth{q_I'}$;
            (2) $\Birth{q_I'} = \alpha$ and $f(\gluepbth_I) > f(\glueplow_I)$;
            (3) $\Birth{q_I'} = \alpha$ and $f(\glueplow_I) > f(\gluepbth_I)$.
            These are the same cases as above and in case (1) and (2) the loop terminates with $\In{\beta_g} = \Mid{\beta_g} = \alpha$, as above.
            The third case, where $f(\glueplow_I) > f(\gluepbth_I)$, cannot occur since we defined $f(\glueplow_I) = f(\Low{\beta_h}) = f(\texttt{nil}) = -\infty$.
            Thus, when $q_I' = \beta_h$, the algorithm terminates.
    \end{description} \smallskip
    The choice that $f(\beta_h) < f(\beta_g)$ is arbitrary and we could have also chosen $f(\beta_g) < f(\beta_h)$.
    In this case we would get three symmetric cases for $q_I' = \beta_h$ with the same outcome.
    
    Once the loop terminates after iteration $I^*$ with $\In{\beta_g} = \Mid{\beta_g} = \alpha$, the tree $\UpTree{h_{I^*}}$ contains all items from $\UpTree{g}$ and $\UpTree{h}$ except for the dummy $\alpha$.
    Thus, $\UpTree{h_{I^*}} = \UpTree{g \cdot h}$, which concludes the proof.
\end{proof}

\clearpage
\section*{Acknowledgement}
\thefunding

\end{document}